%% file: matching_flow.tex
\documentclass[11pt, letterpaper]{article}
\usepackage{fullpage}

\usepackage[T1]{fontenc}

\usepackage[colorlinks=true,pdfpagemode=none,linkcolor=blue,citecolor=blue]{hyperref}
\usepackage{amsmath,amsfonts,amsthm}
\usepackage[figure,boxed,lined]{algorithm2e}
\usepackage{color}
\usepackage{multirow}
\usepackage{enumerate}
\usepackage{graphicx}
\input{files/defs.tex}

\begin{document}

\clubpenalty=10000
\widowpenalty = 10000

\title{Navigating Central Path with Electrical Flows: from Flows to Matchings, and Back \\{\bf (Preliminary draft)}}

\author{Aleksander M\k{a}dry\thanks{Part of this work was done when the author was with Microsoft Research New England.}\\
       {EPFL}\\
       { aleksander.madry@epfl.ch}}
\date{}

\maketitle
\begin{abstract}

We present an $\tO{m^{\frac{10}{7}}}=\tO{m^{1.43}}$-time\footnote{We recall that $\tO{f}$ denotes $O(f \log^c f)$, for some constant $c$.} algorithm for the maximum $s$-$t$ flow and the minimum $s$-$t$ cut problems in directed graphs with unit capacities. This is the first improvement over the sparse-graph case of the long-standing $O(m\min\{\sqrt{m},n^{2/3}\})$ running time bound due to Even and Tarjan \cite{EvenT75} and Karzanov \cite{Karzanov73}. By well-known reductions, this also establishes an $\tO{m^{\frac{10}{7}}}$-time algorithm for the maximum-cardinality bipartite matching problem. That, in turn, gives an improvement over the celebrated $O(m\sqrt{n})$ running time bound of Hopcroft and Karp \cite{HopcroftK73} and Karzanov \cite{Karzanov73} whenever the input graph is sufficiently sparse. 

At a very high level, our results stem from acquiring a deeper understanding of interior-point methods -- a powerful tool in convex optimization -- in the context of flow problems, as well as, utilizing certain interplay between maximum flows and bipartite matchings.

The core of our approach comprises a primal-dual algorithm for {(near-)}perfect bipartite $\bb$-matching problem. This algorithm is inspired by path-following interior-point methods and employs electrical flow computations to gradually improve the quality of maintained solution by advancing it toward (near-)optimality along so-called central path. To analyze this process, we establish a formal connection that ties its convergence rate to the structure of corresponding electrical flows. Then, we exploit that connection to obtain a convergence guarantee for our algorithm that improves upon the well-known barrier of $\Omega(\sqrt{m})$ iterations corresponding to the generic worst-case performance bounds for interior-point-method-based algorithms. This improvement is based on refining certain insights into behavior of electrical flows that stem from the work of Christiano et al. \cite{ChristianoKMST11} and combining them with a new technique for preconditioning primal-dual solutions.

The final ingredient of our approach is a simple reduction of the maximum $s$-$t$ flow problem to the bipartite $\bb$-matching problem. This reduction is then composed with the recent sub-linear-time algorithm for finding perfect matchings in regular graphs of Goel et al. \cite{GoelKK10}, to derive an efficient procedure for rounding fractional $s$-$t$ flows and bipartite matchings.

\end{abstract}

\thispagestyle{empty}
\newpage
\setcounter{page}{1}

\input{files/introduction.tex}
\input{files/preliminaries.tex}

\input{files/outline.tex}
\input{files/reduction.tex}
\input{files/basic_algorithm.tex}
\input{files/improved_algorithm.tex}
\input{files/improved_algorithm2.tex}

\input{files/improved_algorithm3.tex}

\input{files/interior_point.tex}

\input{files/rounding.tex}
\input{files/acknowledgments.tex}
\bibliographystyle{my}
\bibliography{../main}
\appendix
\input{files/appendix.tex}
\input{files/app_reduction.tex}

\input{files/app_basic_algorithm.tex}

\input{files/app_improved_algorithm.tex}

\input{files/app_improved_algorithm2.tex}
\input{files/app_interior_point.tex}

\end{document}

%% file: files/defs.tex
\newtheorem{theorem}{Theorem}[section]
\newtheorem{corollary}[theorem]{Corollary}
\newtheorem{lemma}[theorem]{Lemma}

\newtheorem{definition}[theorem]{Definition}

\newtheorem{fact}[theorem]{Fact}

\newtheorem{invariant}[theorem]{Invariant}

\newcommand{\cdelta}{C_{\delta}}
\newcommand{\cendecrease}{C_{P}}
\newcommand{\crestrict}{C_{R}}
\newcommand{\cenergy}{C_{E}}
\newcommand{\ceta}{C_{\eta}}

\newcommand{\cheavy}{C_{H}}
\newcommand{\fheavy}{F_{H}}
\newcommand{\cfreeze}{C_{F}}
\newcommand{\ckstar}{C_{K}}

\newcommand{\cincrease}{C_{S}}
\newcommand{\cauxiliary}{C_{A}}
\newcommand{\fauxiliary}{F_{A}}

\newcommand{\todo}[1]{}
\newcommand{\bbR}{\mathbb{R}}
\newcommand{\onev}{\mathbf{1}}

\newcommand{\ceil}[1]{\lceil #1 \rceil}
\newcommand{\floor}[1]{\lfloor #1 \rfloor}

\newcommand{\norm}[2]{\|#1\|_{#2}}
\newcommand{\onorm}[1]{|#1|_{1}}

\newcommand{\inorm}[1]{\|#1\|_{\infty}}

\newcommand{\tO}[1]{\widetilde{O}(#1)}
\newcommand{\tOm}[1]{\widetilde{\Omega}(#1)}

\newcommand{\hA}{\widehat{A}}
\newcommand{\oG}{\bar{G}}
\newcommand{\cG}{{G}}

\newcommand{\hG}{\widehat{G}}
\newcommand{\oV}{\bar{V}}

\newcommand{\hV}{\widehat{V}}
\newcommand{\oE}{\bar{E}}

\newcommand{\hE}{\widehat{E}}

\newcommand{\hF}{\widehat{F}}
\newcommand{\hm}{\widehat{m}}

\newcommand{\cm}{{m}}
\newcommand{\hn}{\hat{}}
\newcommand{\cn}{{n}}
\newcommand{\om}{\bar{m}}
\newcommand{\on}{\bar{n}}
\newcommand{\oq}{\bar{q}}
\newcommand{\op}{\bar{p}}

\newcommand{\oS}{\bar{S}}
\newcommand{\hS}{\hat{S}}

\newcommand{\hT}{\widehat{T}}

\newcommand{\pinv}[1]{{#1}^{\dagger}}

\newcommand{\energy}[2]{\mathcal{E}_{#1}(#2)}
\newcommand{\dist}[2]{\mathrm{dist}(#1,#2)}

\newcommand{\Cset}[2]{S_{#1}(#2)}

\newcommand{\eps}{\varepsilon}

\newcommand{\tf}{\tilde{f}}
\newcommand{\hf}{\hat{f}}
\newcommand{\of}{\bar{f}}

\newcommand{\os}{\bar{s}}
\newcommand{\tr}{\tilde{r}}

\newcommand{\ov}{\bar{v}}
\newcommand{\onu}{\bar{\nu}}
\newcommand{\hmu}{\hat{\mu}}
\newcommand{\ohmu}{\hat{\bar{\mu}}}

\newcommand{\hl}{\hat{l}}

\newcommand{\hgamma}{\hat{\gamma}}
\newcommand{\hsigma}{\hat{\sigma}}
\newcommand{\htheta}{\hat{\theta}}

\newcommand{\okappa}{\bar{\kappa}}
\newcommand{\hkappa}{\hat{\kappa}}

\newcommand{\hlambda}{\hat{\lambda}}

\newcommand{\vphi}{\boldsymbol{\mathit{\phi}}}

\newcommand{\vrho}{\boldsymbol{\mathit{\rho}}}
\newcommand{\vmu}{\boldsymbol{\mathit{\mu}}}
\newcommand{\vnu}{\boldsymbol{\mathit{\nu}}}
\newcommand{\ovnu}{\boldsymbol{\bar{\mathit{\nu}}}}
\newcommand{\hvmu}{\boldsymbol{\mathit{\hat{\mu}}}}
\newcommand{\ohvmu}{\boldsymbol{\mathit{\hat{\bar{\mu}}}}}

\newcommand{\vsigma}{\boldsymbol{\mathit{\sigma}}}
\newcommand{\ovsigma}{\boldsymbol{\mathit{\bar{\sigma}}}}
\newcommand{\tvsigma}{\boldsymbol{\mathit{\tilde{\sigma}}}}
\newcommand{\hvsigma}{\boldsymbol{\mathit{\hat{\sigma}}}}
\newcommand{\tsigma}{\tilde{\sigma}}
\newcommand{\vkappa}{\boldsymbol{\mathit{{\kappa}}}}
\newcommand{\ovkappa}{\boldsymbol{\mathit{\bar{\kappa}}}}

\newcommand{\hphi}{\widehat{\phi}}

\newcommand{\tphi}{\widetilde{\phi}}
\newcommand{\hvphi}{\boldsymbol{\widehat{\phi}}}
\newcommand{\tvphi}{\boldsymbol{\tilde{\phi}}}
\newcommand{\hvkappa}{\boldsymbol{\hat{\kappa}}}

\newcommand{\vlambda}{\boldsymbol{\lambda}}

\renewcommand{\aa}{\boldsymbol{\mathit{a}}}
\newcommand{\bb}{\boldsymbol{\mathit{b}}}
\newcommand{\obb}{\boldsymbol{\mathit{\bar{b}}}}

\newcommand{\ff}{\boldsymbol{\mathit{f}}}
\newcommand{\tff}{\boldsymbol{\mathit{\tilde{f}}}}
\newcommand{\off}{\boldsymbol{\mathit{\bar{f}}}}

\newcommand{\hff}{\boldsymbol{\mathit{\hat{f}}}}

\newcommand{\hll}{\boldsymbol{\mathit{\hat{l}}}}

\newcommand{\rr}{\boldsymbol{\mathit{r}}}

\newcommand{\trr}{\boldsymbol{\mathit{\tilde{r}}}}
\renewcommand{\ss}{\boldsymbol{\mathit{s}}}
\newcommand{\oss}{\boldsymbol{\bar{\mathit{s}}}}

\newcommand{\uu}{\boldsymbol{\mathit{u}}}

\newcommand{\xx}{\boldsymbol{\mathit{x}}}
\newcommand{\oxx}{\boldsymbol{\bar{\mathit{x}}}}
\newcommand{\yy}{\boldsymbol{\mathit{y}}}

\newcommand{\BB}{\boldsymbol{\mathit{B}}}

\newcommand{\LL}{\boldsymbol{\mathit{L}}}
\newcommand{\MM}{\boldsymbol{\mathit{M}}}

\newcommand{\RR}{\boldsymbol{\mathit{R}}}

%% file: files/introduction.tex
\section{Introduction}

  The maximum $s$-$t$ flow problem and its dual, the minimum $s$-$t$ cut problem,
  are two of the most fundamental and extensively studied graph problems in combinatorial optimization ~\cite{Schrijver03,AhujaMO93}. They have a wide range of applications (see~\cite{AhujaMOR95}), are often used as subroutines in other algorithms (see, e.g., \cite{AroraHK05,Sherman09}), and a number of other important problems  -- e.g., bipartite matching problem \cite{CormenLRS09} -- can be reduced to them.  Furthermore, these two problems were often a testbed for development of fundamental algorithmic tools and concepts. Most prominently, the Max-Flow Min-Cut theorem \cite{EliasFS56,FordF56} constitutes the prototypical primal-dual relation.
  
  Several decades of extensive work resulted in a number of developments on these problems (see Goldberg and Rao \cite{GoldbergR98} for an overview) and  many of their generalizations and special cases. Still, despite all this effort, the basic problem of computing maximum $s$-$t$ flow and minimum $s$-$t$ cut in general graphs has resisted progress for a long time. In particular, the current best running time bound of $O(m\min\{m^{\frac{1}{2}},n^{\frac{2}{3}}\}\log (n^2/m) \log U)$ (with $U$ denoting the largest integer arc capacity) was established over 15 years ago in a breakthrough paper by Goldberg and Rao \cite{GoldbergR98} and this bound, in turn, matches the $O(m\min\{m^{\frac{1}{2}},n^{\frac{2}{3}}\})$ bound for unit-capacity graphs that Even and Tarjan \cite{EvenT75} -- and, independently, Karzanov \cite{Karzanov73} -- put forth over 35 years ago. 
  
  Recently, however, important progress was made in the context of undirected graphs. Christiano et al. \cite{ChristianoKMST11} developed an algorithm that allows one to compute a $(1+\eps)$-approximation to the undirected maximum $s$-$t$ flow (and the minimum $s$-$t$ cut) problem in $\tO{mn^{\frac{1}{3}} \eps^{-11/3}}$ time. Their result relies on devising a new approach to the problem that combines electrical flow computations with multiplicative weights update method (see \cite{AroraHK05}). Later, Lee et al. \cite{LeeRS13} presented a quite different -- but still electrical-flow-based -- algorithm that employs purely gradient-descent-type view to obtain an $\tO{mn^{1/3}\eps^{-2/3}}$-time $(1+\eps)$-approximation for the case of unit capacities. Finally, very recently, this line of work was culminated by Sherman \cite{Sherman13} and Kelner et al. \cite{KelnerLOS13} who independently showed how to integrate non-Euclidean gradient-descent methods with fast poly-logarithmic-approximation algorithms for cut problems of M\k{a}dry \cite{Madry10b} to get an $O(m^{1+o(1)}\eps^{-2})$-time $(1+\eps)$-approximation to the undirected maximum flow problem. 
  
Finally, we note that, in parallel to the above work that is focused on designing weakly-polynomial algorithms for the maximum $s$-$t$ flow and minimum $s$-$t$ cut problems, there is also a considerable interest in obtaining running time bounds that are strongly-polynomial, i.e., that do not depend on the values of arc capacities. The current best such bound is $O(mn)$ and it follows by combining the algorithms of King et al. \cite{KingRT94} and Orlin \cite{Orlin13}.

  \paragraph{Bipartite Matching Problem.}
  
  Another problem that we will be interested in is the (maximum-cardinality) bipartite matching problem -- a fundamental assignment problem with numerous applications (see, e.g., \cite{AhujaMO93,LovaszP86}) and long history. Already in 1931, K\"onig \cite{Konig31} and Egerv\'ary \cite{Egervary31} provided first constructive characterization of maximum matchings in bipartite graphs. This characterization can be turned into a polynomial-time algorithm. Then, in 1973, Hopcroft and Karp \cite{HopcroftK73} -- and, independently, Karzanov \cite{Karzanov73} -- devised the celebrated $O(m\sqrt{n})$-time algorithm. Till date, this bound is the best one known in the regime of relatively sparse graphs. It can be improved, however, when the input graph is dense, i.e., when $m$ is close to $n^2$. In this case, one can combine the algebraic approach of Rabin and Vazirani \cite{RabinV89} -- that itself builds on the work of Tutte \cite{Tutte47} and Lov\'asz \cite{Lovasz79} -- with matrix-inversion techniques of Bunch and Hopcroft \cite{BunchH74} to get an algorithm that runs in $O(n^{\omega})$ time (see \cite{Mucha05}), where $\omega\leq 2.3727$ is the exponent of matrix multiplication \cite{CoppersmithW90, Vassilevska12}. Also, later on, Alt et al. \cite{AltBMP91}, as well as, Feder and Motwani \cite{FederM95} developed combinatorial algorithms that offer a slight improvement -- by a factor of, roughly, $\log_n \frac{n^2}{m}$ -- over the $O(m\sqrt{n})$ bound of Hopcroft and Karp whenever the graph is sufficiently dense. 
  
  Finally, it is worth mentioning that there was also a lot of developments on the (maximum-cardinality) matching problem in general, i.e., not necessarily bipartite, graphs. Starting with the pioneering work of Edmonds \cite{Edmonds65}, these developments led to bounds that essentially match the running time guarantees that were previously known only for bipartite case. More specifically, the running time bound of $O(m\sqrt{n})$ for the general-graph case was obtained by Micali and Vazirani \cite{MicaliV80,Vazirani94} (see also \cite{GabowT91} and \cite{GoldbergK04}). While, building on the algebraic characterization of the problem due to Rabin and Vazirani \cite{RabinV89}, Mucha and Sankowski \cite{MuchaS04} and then Harvey \cite{Harvey09} gave $O(n^{\omega})$-time algorithms for general graphs.
  
\subsection{Our Contribution}

In this paper, we develop a new algorithm for solving maximum $s$-$t$ flow and minimum $s$-$t$ cut problems in directed graphs. More precisely, we prove the following theorem.
  
\begin{theorem}\label{thm:main}
Let $G=(V,E)$ be a directed graph with $m$ arcs and unit capacities. For any two vertices $s$ and $t$, one can compute an integral maximum $s$-$t$ flow and minimum $s$-$t$ cut of $G$ in $\tO{m^{\frac{10}{7}}}$ time.
\end{theorem}

This improves over the long-standing $O(m\min\{\sqrt{m},n^{2/3}\})$ running time bound due to Even and Tarjan \cite{EvenT75} and, in particular, finally breaks the $\Omega(n^{\frac{3}{2}})$ running time barrier for sparse directed graphs.

Furthermore, by applying a well-known reduction (see \cite{CormenLRS09}), our new algorithm gives the first improvement on the sparse-graph case of the seminal $O(m\sqrt{n})$-time algorithms of Hopcroft-Karp \cite{HopcroftK73} and Karzanov \cite{Karzanov73} for the maximum-cardinality bipartite matching problem. 

\begin{theorem}\label{thm:main_matchings}
Let $G=(V,E)$ be an undirected bipartite graph with $m$ edges, one can solve the maximum-cardinality bipartite matching problem in $G$ in $\tO{m^{\frac{10}{7}}}$ time. 
\end{theorem}

\noindent This, again, breaks the 40-years-old running time barrier of $\Omega(n^{\frac{3}{2}})$ for this problem in sparse graphs. 

Additionally, we design a simple reduction of the maximum $s$-$t$ flow problem to perfect bipartite $\bb$-matching problem (see Theorem \ref{thm:flow_to_matchings}). (This reduction can be seen as an adaptation of the reduction of the maximum vertex-disjoint $s$-$t$-path problem to the bipartite matching problem due to Hoffman \cite{Hoffman60} -- cf. Section 16.7c in \cite{Schrijver03}.\footnote{We thank Lap Chi Lau \cite{Lau13} for pointing out this similarity.}) As the reduction in the other direction is well-known already, this establishes an algorithmic equivalence of these two problems. We also show (see Theorem \ref{thm:rounding_matchings} and Corollary \ref{col:rounding_flows}) how this reduction, together with the sub-linear-time algorithm for perfect matching problem in regular bipartite graphs of Goel et al. \cite{GoelKK10}, leads to an efficient, nearly-linear time, rounding procedure for $s$-$t$ flows.\footnote{Recently, it came to our attention that a very similar rounding result was independently obtained by Khanna et al. \cite{KhannaKL13}.}

Finally, our main technical contribution is a primal-dual algorithm for {(near-)}perfect bipartite $\bb$-matching problem (see Theorem \ref{thm:interior_point_matchings}). This iterative algorithm draws on ideas underlying interior-point methods and the electrical flow framework of Christiano et al. \cite{ChristianoKMST11}. It employs electrical flow computations to gradually improve the quality of maintained solution by advancing it toward (near-)optimality along so-called central path. 

We develop a way of analyzing this algorithm's rate of convergence by relating it to the structure of the corresponding electrical flows (see Theorem \ref{thm:main_interior_point}). This understanding enables us to devise a way of perturbing (see Section \ref{sec:heavy}) and preconditioning (see Section \ref{sec:preconditioning}) our intermediate solutions to ensure a convergence in only $\tO{m^{\frac{3}{7}}}$ iterations and thus improve over the well-known barrier of $\Omega(m^{\frac{1}{2}})$ iterations that all the previous interior-point-methods-based algorithms suffer from. (To the best of our knowledge, this is the first time that this barrier was broken for a natural optimization problem.) 

We also note that most of this understanding of convergence behavior of interior-point methods can be carried over to general LP setting. Therefore, we are hopeful that our techniques can be extended and will eventually lead to breaking the $\Omega(m^{\frac{1}{2}})$ iterations barrier for general interior-point methods.

\subsection{Our Approach}

The core of our approach comprises two components. One of them is combinatorial in nature and exploits an intimate connection between the maximum $s$-$t$ flow problem and  bipartite matching problem. The other one is more linear-algebraic and relies on interplay of interior-point methods and electrical flows. 

\paragraph{Maximum flows and bipartite matchings.} The combinatorial component shows that not only one can reduce bipartite matching problem to the maximum $s$-$t$ flow problem, but also that a reduction in the other direction exists. Namely, one can reduce, in a simple and purely combinatorial way, the maximum $s$-$t$ flow problem to a certain variant of bipartite matching problem (see Theorem \ref{thm:flow_to_matchings}). Once this reduction is established, it allows us to shift our attention to the matching problem. 

Also, as a byproduct, this reduction -- together with the algorithm of Goel et al. \cite{GoelKK10} -- yields a fast procedure for rounding fractional maximum flows (see Corollary \ref{col:rounding_flows}). This enables us to focus on obtaining solutions that are only nearly-optimal, instead of being optimal.

\paragraph{Bipartite Matchings and Electrical Flows.} The other component is based on using the interior-point method framework in conjunction with nearly-linear time electrical flow computations, to develop a faster algorithm for the bipartite matching problem. 

The point of start here is a realization that the recent approaches to approximating undirected maximum flow \cite{ChristianoKMST11,LeeRS13,Sherman13,KelnerLOS13}, despite achieving impressive progress, have fundamental limitations that make them unlikely to yield improvements for the exact undirected or (approximate) directed setting.\footnote{Note that it is known -- see, e.g., \cite{Madry11} -- that computing exact maximum $s$-$t$ flow in undirected graphs is algorithmically equivalent to computing the exact or approximate maximum $s$-$t$ flow in directed graph.}  Very roughly speaking, these limitations stem from the fact that, at their core, all these algorithms employ some version of gradient-descent method that relies on purely primal arguments, while almost completely neglecting the dual aspect of the problem. It is well-understood, however, that getting a running time guarantee that depends logarithmically, instead of polynomially, on $\eps^{-1}$ -- and such dependence is a prerequisite to making progress in directed setting -- one needs to also embrace  the dual side of the problem and take full advantage of it.  

\paragraph{Interior-point methods and fast algorithms.} The above realization motivates us to consider a more sophisticated approach, one that is inherently primal-dual and achieves logarithmic dependence on $\eps^{-1}$: interior-point methods. These methods constitute a powerful optimization paradigm that is a cornerstone of convex optimization (see, e.g., \cite{BoydV04,Wright97,Ye97}) and already led to development of polynomial-time exact algorithms for a variety of problems. Unfortunately, despite all its advantages and successes in tackling hard optimization tasks, this paradigm has certain shortcomings in the context of designing fast algorithms. The main reason for that is the fact that each iteration of interior-point method requires solving of a linear system, a task for which the current fastest general-purpose algorithm runs in $O(n^{\omega})=O(n^{2.3727})$ time \cite{AhoHU74,CoppersmithW90,Vassilevska12}. So, this bound becomes a bottleneck if one was aiming for, say, even sub-quadratic-time algorithm.   

Fortunately, it turns out that there is a way to circumvent this issue. Namely, even though the above bound is the best one known in general, one can get a better running time when dealing with some specific problem. This is achieved by exploiting the special structure of the corresponding linear systems. A prominent (and most important from our point of view) example here is the family of flow problems. Daitch and Spielman \cite{DaitchS08} showed that in the context of flow problems one can use the power of fast (approximate) Laplacian system solvers \cite{SpielmanTeng04,KoutisMP10,KoutisMP11,KelnerOSZ13} to solve the corresponding linear systems in nearly-linear time. This enabled \cite{DaitchS08} to develop a host of $\tO{m^{\frac{3}{2}}}$-time algorithms for a number of important generalizations of the maximum flow problem for which there was no such algorithms before. 

Unfortunately, this bound of $\tO{m^{\frac{3}{2}}}$ time turns out to also be  a barrier if one wants to obtain even faster algorithms. The new difficulty here is that the best worst-case bound on the number of iterations needed for an interior-point method to converge to near-optimal solution is $\Omega(m^{1/2})$. Although it is widely believed that this bound is far from optimal, it seems that our theoretical understanding of interior-point method convergence is still insufficient to make any progress on this front. In fact,  improving this state of affairs is a major and long-standing challenge in mathematical programing. 

\paragraph{Beyond the $\Omega(m^{\frac{1}{2}})$ barrier.}  Our approach to circumventing this $\Omega(m^{\frac{1}{2}})$ barrier and obtaining the desired $\tO{m^{\frac{10}{7}}}$-time algorithm for the bipartite $\bb$-matching problem consists of two stages.

First one -- presented in Section \ref{sec:simple} -- corresponds to setting up a primal-dual framework for solving the near-perfect $\bb$-matching problem. This framework is directly inspired by the principles underlying path-following interior-point methods and, in some sense, is equivalent to them. In it, we start with some initial sub-optimal solution (that is encoded as a minimum-cost flow problem instance) and gradually improve its quality up to near-optimality. These gradual improvements are guided by certain electrical flow computations -- the flows are used to update the primal solution and the corresponding voltages update the dual one -- and our solution ends up following a special trajectory in the feasible space: so-called central path. 

We analyze the performance of this optimization process by establishing a formal connection that ties the size of each improvement step to a certain characteristic of the corresponding electrical flow. Very roughly speaking, this size (and thus the resulting rate of convergence) is directly related to how much the electrical flow we compute resembles the current primal solution (which is also a flow). Once this connection is established, a simple energy-based argument immediately recovers the generic $O(m^{\frac{1}{2}})$ iterations bound known for interior-point methods. So, as each electrical flow computation can be performed in $\tO{m}$ time, this gives an overall $\tO{m^{\frac{3}{2}}}$-time algorithm.  

Finally, to improve upon the above $O(m^{\frac{1}{2}})$ iterations bound and deliver the desired $O(m^{\frac{10}{7}})$-time procedure, in Section \ref{sec:improved}, we devise two techniques: perturbation of arcs -- that can be seen as a refinement of the edge removal technique of  Christiano et al. \cite{ChristianoKMST11}; and solution preconditioning -- a way of adding auxiliary arcs to the solution to improve its conductance properties. We show that by a careful composition of these techniques, one is able to ensure that the guiding electrical flows align better with the primal solution -- thus allowing taking larger progress steps and guaranteeing faster convergence -- while keeping the unwanted impact of these modifications on the quality of final solution minimal. The analysis of this process constitutes the technical core of our result and is based on understanding of the interplay between the interior-point method and both the primal and dual structure of electrical flows.

We believe that this approach of understanding interior-point methods through the lens of electrical flows is a promising direction and our result is just a first step towards realizing its full potential.

\subsection{Organization}

We begin the technical part of the paper in Section \ref{sec:preliminaries} where we present some preliminaries on maximum flow problem, electrical flows, and bipartite ($\bb$-)matching problem, as well as, introduce some theorems we will need in the sequel. In Section \ref{sec:outline}, we provide a general outline of our results and the structure of our proof. 

In Section \ref{sec:reduction}, we describe the reduction of maximum $s$-$t$ flow problem to the bipartite $\bb$-matching problem. Next, in Sections \ref{sec:simple} and \ref{sec:improved}, we explain how our path-following algorithms and electrical flows can be used to get an improved algorithm for the bipartite $\bb$-matching problem, with Section \ref{sec:proof_main_interior_point} presenting the analysis of our path-following primitive. Finally, we conclude in Section \ref{sec:rounding} by showing how to round fractional $\bb$-matchings to integral ones. 

%% file: files/preliminaries.tex
\section{Preliminaries}\label{sec:preliminaries}

In this section, we introduce some basic notation and definitions we will need later. 

\subsection{\texorpdfstring{$\vsigma$-Flows}{Sigma-Flows} and the Maximum $s$-$t$ Flow Problem}

Throughout this paper, we denote by $G=(V,E,\uu)$ a directed graph with vertex set $V$, arc set $E$ (we allow parallel arcs), and (non-negative) integer capacities $u_e$, for each arc $e\in E$. We usually define $m=|E|$ to be the number of arcs of the graph in question and $n=|V|$ to be the number of its vertices. Each arc $e$ of $G$ is an ordered pair $(u,v)$, where $u$ is its {\em tail} and $v$ is its {\em head}.

The basic notion of this paper is the notion of a {\em $\vsigma$-flow} in $G$, where $\vsigma\in \bbR^n$, with $\sum_v \sigma_v =0$, is the {\em demand vector}. By a $\vsigma$-flow in $G$ we understand any vector $\ff\in \bbR^m$ that assigns values to arcs $G$ and satisfies the {\em flow conservation constraints}:
\begin{equation}\label{eq:conservation_constraints}
\sum_{e\in E^+(v)} f_{e} - \sum_{e\in E^-(v)} f_{e} = \sigma_v, \quad \text{for each vertex $v\in V$}.
\end{equation} 
Here, $E^+(v)$ (resp. $E^-(v)$) is the set of arcs of $G$ that are leaving (resp. entering) vertex $v$. Intuitively, these constraints enforce that the net balance of the total in-flow into vertex $v$ and  the total out-flow out of that vertex is equal to $\sigma_v$, for every $v\in V$. 

Furthermore, we say that a $\vsigma$-flow $\ff$ is {\em feasible} in $G$ iff $\ff$ obeys the {\em non-negativity and capacity constraints}:
\begin{equation}\label{eq:capacity_constraints}
0\leq f_e \leq u_e, \quad \text{for each arc $e\in E$}.
\end{equation} 

One type of $\vsigma$-flows that will be of special interest to us are $s$-$t$ flows, where $s$ (the {\em source}) and $t$ (the {\em sink}) are two distinguish vertices of $G$. Formally, a $\vsigma$-flow $\ff$ is an {\em $s$-$t$ flow} iff its demand vector $\vsigma$ is equal to $F\cdot \chi_{s,t}$ for some $F\geq 0$ -- we call $F$ the {\em value} of $\ff$ -- and the demand vector $\chi_{s,t}$ that has $-1$ (resp. $1$) at the coordinate corresponding to $s$ (resp. $t$) and zeros everywhere else. 

Now, the {\em maximum $s$-$t$ flow problem} corresponds to a task of finding for a given graph $G=(V,E,\uu)$, a source $s$, and a sink $t$, a feasible $s$-$t$ flow $\ff^*$ in $G$ of maximum value $F$. We call such a flow $\ff^*$ that maximizes $F$ {\em the maximum $s$-$t$ flow of $G$} and denote its value by $F^*$.

Sometimes, we will be also interested in (uncapacitated) {\em minimum-cost $\vsigma$-flow} problem (with non-negative costs). In this problem, we have a directed graph $G$ with infinite capacities on arcs (i.e., $u_e=+\infty$, for all $e$) and certain {\em (non-negative) length} (or {\em cost}) $l_e$ assigned to each arc $e$. Our goal is to find a feasible $\vsigma$-flow $\ff$ in $G$ whose {\em cost $l(\ff):=\sum_{e} l_e f_e$} is minimal. (Note that as we have infinite capacities here, the feasibility constraint \eqref{eq:capacity_constraints} just requires that $f_e\geq 0$ for all arcs $e$.)

Finally, one more problem that will be relevant in this context is the {\em minimum $s$-$t$ cut} problem. In this problem, we are given a directed graph $G=(V,E,\uu)$ with integer capacities, as well as, a source $s$ and sink $t$, and our task is to find an $s$-$t$ cut $C\subseteq V$ in $G$  minimizes the {\em capacity $\uu(C):=\sum_{E^{-}(C)} u_e$} among all $s$-$t$ cuts. Here, a cut $C\subseteq V$ is an {\em $s$-$t$ cut} iff $s\in C$ and $t\notin C$, and  $E^{-}(C)$ is the set of all arcs $(u,v)$ with $u\in C$ and $v\notin C$. It is well-known \cite{EliasFS56,FordF56} that the minimum $s$-$t$ cut problem is the dual of the maximum $s$-$t$ problem and, in particular, that the capacity of the minimum $s$-$t$ cut is equal to the value of the maximum $s$-$t$ flow, as well as, that given a maximum $s$-$t$ flow one can easily obtain the corresponding minimum $s$-$t$ cut. 

\subsection{Undirected Graphs}

Although the focus of our results is on directed graphs, it will be crucial for us to consider undirected graphs too. To this end, we view an undirected graph $G=(V,E,\uu)$ as a directed one in which the ordered pair $(u,v)\in E$ does not denote an arc anymore, but an (undirected) {\em edge} $(u,v)$ and the order just specifies an {\em orientation} of that edge from $u$ to $v$. (Even though we use the same notation for these two different types of graphs, we will always make sure that it is clear from the context whether we deal with directed graph that has arcs, or with undirected graph that has edges.) From this perspective, the definitions of $\vsigma$-flow $\ff$ that we introduced above for directed graphs transfer over to undirected setting almost immediately. The only (but very crucial) difference  is that in undirected graphs a feasible flow can have some of $f_e$s being negative - this corresponds to the flow flowing in the direction that is opposite to the edge orientation. As a result, the feasibility condition \eqref{eq:capacity_constraints} becomes 
\begin{equation}\label{eq:undir_capacity_constraints}
|f_e|\leq u_e, \quad \text{for each arc $e\in E$}.
\end{equation} 
 Also, the set $E^+(v)$ (resp. $E^-(v)$) denotes now the set of incident edges that are oriented towards (resp. away) from $v$, and $E(v):=E^+(v)\cup E^-(v)$ is just the set of all edges incident to $v$, regardless of their orientation. 

Finally, given a directed graph $G=(V,E,\uu)$, by its {\em projection} $\oG$ we understand an undirected graph that arises from treating each arc of $G$ as an edge with the corresponding orientation. Note that if $G$ had two arcs $(u,v)$ and $(v,u)$ then $\oG$ will have two parallel edges $(u,v)$ and $(v,u)$ that have opposite orientation and, possibly, different capacities.

\subsection{Electrical Flows and Potentials}

A notion that will play a fundamental role in this paper is the notion of electrical flows. Here, we just briefly review some of the key properties that we will need later. For an in-depth treatment we refer the reader to \cite{Bollobas98}. 

Consider an undirected graph $G$ and some vector of resistances $\rr\in \bbR^{m}$ that assigns to each edge $e$ its {\em resistance} $r_e>0$. For a given $\vsigma$-flow $\ff$ in $G$, let us define its {\em energy} (with respect to $\rr$) $\energy{\rr}{\ff}$ to be
\begin{equation}\label{eq:def_energy_flow}
\energy{\rr}{\ff}:= \sum_e r_e f_e^2 = \ff^T \RR \ff,
\end{equation}
where $\RR$ is an $m \times m$ diagonal matrix with $R_{e,e}=r_e$, for each edge $e$.

For a given undirected graph $G$, a demand vector $\vsigma$, and a vector of resistances $\rr$, we define an {\em electrical $\vsigma$-flow} in $G$ (that is {\em determined} by resistances $\rr$) to be the $\vsigma$-flow that minimizes the energy $\energy{\rr}{\ff}$ among all $\vsigma$-flows in $G$. As energy is a strictly convex function, one can easily see that such a flow is unique. Also, we emphasize that we do \emph{not} require here that this flow is feasible with respect to capacities of $G$ (cf. \eqref{eq:undir_capacity_constraints}). Furthermore, whenever we consider electrical flows in the context of a directed graph $G$, we will mean an electrical flow -- as defined above -- in the (undirected) projection $\oG$ of $G$.

One of very useful properties of electrical flows is that it can be characterized in terms of vertex potentials inducing it. Namely, one can show that a $\vsigma$-flow $\ff$ in $G$ is an electrical $\vsigma$-flow determined by resistances $\rr$ iff there exist {\em vertex potentials} $\phi_v$ (that we collect into a vector $\vphi\in \bbR^n$) such that, for any edge $e=(u,v)$ in $G$ that is oriented from $u$ to $v$,
\begin{equation}\label{eq:potential_flow_def}
f_e = \frac{\phi_v-\phi_u}{r_e}.
\end{equation}
In other words, a $\vsigma$-flow $\ff$ is an electrical $\vsigma$-flow iff it is {\em induced} via \eqref{eq:potential_flow_def} by some vertex potential $\vphi$. (Note that orientation of edges matters in this definition.)

Using vertex potentials, we are able to express the energy $\energy{\rr}{\ff}$ (see \eqref{eq:def_energy_flow}) of an electrical $\vsigma$-flow $\ff$ in terms of the potentials $\vphi$ inducing it as
\begin{equation}\label{eq:def_energy_potentials}
\energy{\rr}{\ff}= \sum_{e=(u,v)} \frac{(\phi_v-\phi_u)^2}{r_e}.
\end{equation}

One of the consequences of this characterization of electrical flows via vertex potentials is that one can view the energy of an electrical $\vsigma$-flow as being a result of optimization not over all the $\vsigma$-flows but rather over certain set of vertex potentials. Namely, we have the following lemma that, for completeness, we prove in the Appendix \ref{app:effective_conductance}. 

\begin{lemma}\label{lem:effective_conductance}
For any graph $G=(V,E)$, any vector of resistances $\rr$, and any demand vector $\vsigma$, 
\[
\frac{1}{\energy{\rr}{\ff^*}} = \min_{\vphi| \vsigma^T \vphi=1} \sum_{e=(u,v)\in E} \frac{(\phi_v-\phi_u)^2}{r_{e}},
\]
where $\ff^*$ is the electrical $\vsigma$-flow determined by $\rr$ in $G$. Furthermore, if $\vphi^*$ are the vertex potentials corresponding to $\ff^*$ then the minimum is attained by taking $\vphi$ to be equal to $\tvphi:=\vphi^*/\energy{\rr}{\ff^*}$.
\end{lemma}

Note that the above lemma provides a convenient way of lowerbounding the energy of an electrical $\vsigma$-flow. One just needs to expose any vertex potentials $\vphi$ such that $\vsigma^T \vphi=1$ and this will immediately constitute an energy lowerbound. 
Also, another basic but useful property of electrical $\vsigma$-flows is captured by the following fact.

\begin{fact}[Rayleigh Monotonicity]
\label{fa:rayleigh_monotonicity}
For any graph $G=(V,E)$, demand vector $\vsigma$ and any two vectors of resistances $\rr$ and $\rr'$ such that $r_e\geq r_e'$, for all $e\in E$, we have that if $\ff$ (resp. $\ff'$) is the electrical $\vsigma$-flow determined by $\rr$ (resp. $\rr'$) then
\[
\energy{\rr}{\ff}\geq \energy{\rr'}{\ff'}.
\]
\end{fact}

\subsection{Laplacian Solvers}

A very important algorithmic property of electrical flows is that one can compute very good approximations of them in nearly-linear time. Below, we briefly describe the tools enabling that.

To this end, let us recall that electrical $\vsigma$-flow is the (unique) $\vsigma$-flow induced by vertex potentials via \eqref{eq:potential_flow_def}. So, finding such a flow boils down to computing the corresponding vertex potentials $\vphi$. It turns out that computing these potentials can be cast as a task of solving certain type of linear system called {\em Laplacian} systems. To see that, let us define the \emph{edge-vertex incidence matrix} $\BB$ being an $n\times m$ matrix with rows indexed by vertices
  and columns indexed by edges such that
\[
  \BB_{v,e} = \begin{cases}
1  & \text{if $e\in E^+(v)$,}\\
-1  & \text{if $e\in E^{-}(v)$,}\\
0 & \text{otherwise.}
\end{cases}
\]

Now, we can compactly express the flow conservation constraints \eqref{eq:conservation_constraints} of a $\vsigma$-flow $\ff$ (that we view as a vector in $\bbR^m$) as
\[
\BB \ff =\vsigma.
\]

On the other hand, if $\vphi$ are some vertex potentials, the corresponding flow $\ff$ induced by $\vphi$ via \eqref{eq:potential_flow_def} (with respect to resistances $\rr$) can be written as
\[
\ff = \RR^{-1} \BB^T \vphi,
\]
where again $\RR$ is a diagonal $m\times m$ matrix with $R_{e,e}:=r_e$, for each edge $e$. 

Putting the two above equations together, we get that the vertex potentials $\vphi$ that induce the electrical $\vsigma$-flow determined by resistances $\rr$ are given by a  solution to the following linear system
\begin{equation}\label{eq:elec_flow_lin_system}
\BB \RR^{-1} \BB^T \vphi = \LL \vphi= \vsigma,
\end{equation}
where $\LL:=\BB \RR^{-1} \BB^T$ is the (weighted) \emph{Laplacian $\LL$} of $G$ (with respect to the resistances $\rr$). One can easily check that $\LL$ is an $n\times n$ matrix indexed by vertices of $G$ with entries given by 
\begin{equation}
\label{eq:def_of_laplacian}
  L_{u,v} = \begin{cases}
\sum_{e\in E(v)} 1/r_e & \text{if $u=v$,}\\
-1/r_{e}  & \text{if $e=(u,v)\in E$, and}\\
0 & \text{otherwise.}
\end{cases}
\end{equation}

One can see that the Laplacian $\LL$ is not invertible, but -- as long as, the underlying graph is connected -- it's null-space is one-dimensional and spanned by all-ones vector. As we require our demand vectors $\vsigma$ to have its entries sum up to zero (otherwise, no $\vsigma$-flow can exist), this means that they are always orthogonal to that null-space. Therefore, the linear system \eqref{eq:elec_flow_lin_system} has always a solution $\vphi$ and one of these solutions\footnote{Note that the linear system \eqref{eq:elec_flow_lin_system} will have many solutions, but each two of them are equivalent up to a translation. So, as the formula \eqref{eq:potential_flow_def} is translation-invariant, each of these solutions will yield the same unique electrical $\vsigma$-flow.} is given by
\[
\vphi = \pinv{\LL} \vsigma,
\]
where $\pinv{\LL}$ is the Moore-Penrose pseudo-inverse of $\LL$.

Now, from the algorithmic point of view, the crucial property of the Laplacian $\LL$ is that it is symmetric and {\em diagonally dominant}, i.e., for any $v\in V$, $\sum_{u\neq v} |L_{u,v}| \leq L_{v,v}$. This enables us to use fast approximate solvers for symmetric and diagonally dominant linear systems to compute an approximate electrical $\vsigma$-flow. Namely, building on the work of Spielman and Teng \cite{SpielmanT03,SpielmanTeng04}, Koutis et al. \cite{KoutisMP10,KoutisMP11} designed an SDD linear system solver that implies the following theorem. (See also recent work of Kelner et al. \cite{KelnerOSZ13} that presents an even simpler nearly-linear-time Laplacian solver.)

\begin{theorem}\label{thm:vanilla_SDD_solver}
For any $\eps>0$, any graph $G$ with $n$ vertices and $m$ edges, any demand vector $\vsigma$, and any resistances $\rr$, one can compute in $\tO{m \log m \log \eps^{-1}}$ time vertex potentials $\tvphi$ such that
$\|\tvphi-\vphi^*\|_{\LL}\leq \eps \|\vphi^*\|_{\LL}$, where $\LL$ is the Laplacian of $G$, $\vphi^*$ are potentials inducing the electrical $\vsigma$-flow determined by resistances $\rr$, and $\|\vphi\|_{\LL}:=\sqrt{\vphi^T \LL \vphi}$.
\end{theorem}

To understand the type of approximation offered by the above theorem, observe that $\|\vphi\|_{\LL}^2=\vphi^T\LL\vphi$ is just the energy of the flow induced by vertex potentials $\vphi$. Therefore, $\|\tvphi-\vphi^*\|_{\LL}$ is the energy of the electrical flow $\off$ that ``corrects'' the vertex demands of the electrical $\tvsigma$-flow induced by potentials $\tvphi$, to the ones that are dictated by $\vsigma$. So, in other words, the above theorem tells us that we can quickly find an electrical $\tvsigma$-flow $\tff$ in $G$ such that $\tvsigma$ is a slightly perturbed version of $\vsigma$ and $\tff$ can be corrected to the electrical $\vsigma$-flow $\ff^*$ that we are seeking, by adding to it some electrical flow $\off$ whose energy is at most $\eps$ fraction of the energy of the flow $\ff^*$. (Note that electrical flows are linear, so we indeed have that $\ff^*=\tff+\off$.) As we will see, this kind of approximation is completely sufficient for our purposes. 

\subsection{Bipartite \texorpdfstring{$\bb$-Matchings}{b-Matchings}}

A fundamental graph problem that constitutes both an application of our results, as well as, one of the tools we use to establish them, is the {\em (maximum-cardinality) bipartite $\bb$-matching problem}. In this problem, we are given an undirected bipartite graph $G=(V,E)$ with $V=P\cup Q$ -- where $P$ and $Q$ are the two sets of bipartition -- as well as, a \emph{demand vector} $\bb$ that assigns to every vertex $v$ an integral and positive demand $b_v$. Our goal is to find a maximum cardinality \emph{multiset} $M$ of the edges of $G$ that forms a {\em $\bb$-matching}. That is, we want to find a multi-set $M$ of edges of $G$ that is of maximum cardinality subject to a constraint that, for each vertex $v\in V$, the number of edges of $M$ that are incident to $v$ is at most $b_v$. (When $b_v=1$ for every vertex $v$, we will simply call such $M$ a \emph{matching}.)

 We say that a $\bb$-matching $M$ is \emph{perfect} iff every vertex in $V$ has exactly $b_v$ edges incident to it in $M$. Note that a perfect $\bb$-matching - if it exists in $G$ - has to necessarily be of maximum cardinality. Also, if a graph has a perfect $\bb$-matching then it must be that $\sum_{v\in P} b_v=\sum_{v\in Q} b_v$. Now, by the {\em perfect bipartite $\bb$-matching problem} we mean a task in which we need to either find the perfect $\bb$-matching in $G$ or conclude that it does not exist. 
 
 Finally, by a {\em fractional} solution to a $\bb$-matching problem, we understand an $|E|$-dimensional vector $\xx$ that allocates non-negative value of $x_e$ to each edge $e$ and is such that for every vertex $v$ of $G$, the sum $\sum_{e\in E(v)} x_e$ of (fractional) incident edges in $\xx$ is at most $b_v$. Also, we define the {\em size} of a fractional $\bb$-matching $\xx$  to be $\onorm{\xx}$.

An interesting class of graphs that is guaranteed to always have a perfect matching are bipartite graphs that are \emph{$d$-regular}, i.e., that have the degree of each vertex equal to $d$. A remarkable algorithm of Goel et al. \cite{GoelKK10} shows that one can find a perfect matching in such graphs in time that is proportional only to number of its vertices and not edges. (Note that a $d$-regular bipartite graph has exactly $\frac{dn}{2}$ edges and thus this number can be much higher than $n$ when $d$ is large.) In particular, they prove the following theorem that we will use later. 

\begin{theorem}[see Theorem 4 in \cite{GoelKK10}]\label{thm:regular_bipartite_matchings}
Given an $n\times n$ doubly-stochastic matrix $\MM$ with $m$ non-zero entries, one can find a perfect matching in the support of $M$ in $O(n\log^2 n)$ expected time with $O(m)$ preprocessing time.
\end{theorem}

%% file: files/outline.tex
\section{From Flows to Matchings, and Back}\label{sec:outline}

As we already mentioned, our results stem from exploiting the interplay between the maximum $s$-$t$ flow and bipartite $\bb$-matching problem, as well as, from understanding the performance of interior-point methods -- when applied to these two problems -- via the structure of corresponding electrical flows. To highlight these elements, we decompose the proof of our main theorem (Theorem \ref{thm:main}) into three natural parts.

\subsubsection*{Reducing Maximum Flow to $\bb$-Matching}

First, we focus on analyzing the relationship between the maximum $s$-$t$ flow and the (maximum-cardinality) bipartite $\bb$-matching problem. It is well-known that the latter  can be reduced to the former in a simple way. As it turns out, however, one can also go the other way -- there is a simple, combinatorial reduction from the maximum flow problem to the task of finding a perfect bipartite $\bb$-matching.\footnote{One can view this as one possible explanation of why the techniques used in the context of bipartite matchings and maximum flows are so similar.} 

Before making this precise, let us introduce one definition. Consider a $\bb$-matching problem instance corresponding to a bipartite graph $G=(V,E)$ with $P$ and $Q$ ($V=P\cup Q$) being two sides of the bipartition. For any edge $e=(p,q)\in E$, let us define the {\em thickness} $d(e)$ of that edge to be $d(e):=\min\{b_p,b_q\}$. (So, $d(e)$ is an upper bound on the value of $x_e$ in any feasible $\bb$-matching $\xx$.) We say that a $\bb$-matching instance is {\em balanced} iff 
\begin{equation}
\label{eq:def_balanced}
\sum_{e\in E} d(e)\leq 4\onorm{\bb}.  
\end{equation}

Now, in Section \ref{sec:reduction}, we establish the following result.

\begin{theorem}\label{thm:flow_to_matchings}
If one can solve a balanced instance of a perfect bipartite $\bb$-matching problem in a (bipartite) graph with  $\on$ vertices and $\om$ edges in $T(\on,\om,\onorm{\bb})$ time, then one can solve the maximum $s$-$t$ flow problem in a graph $G=(V,E,\uu)$ with $m$ arcs and capacity vector $\uu$ in $\tO{(m+T(\Theta(m),4m,4\onorm{\uu}))\log \onorm{\uu}}$ time. 
\end{theorem}

This connection between maximum flows and bipartite matchings is useful in two ways. Firstly, it enables us to reduce the main problem we want to solve -- the maximum $s$-$t$ flow problem with unit capacities -- to a seemingly simpler one: the perfect bipartite $\bb$-matching problem. Secondly, the fact that this reduction works also for fractional instances provides us with an ability to lift our $\bb$-matching rounding procedure that we develop later (see Theorem \ref{thm:rounding_matchings}) to the maximum flow setting (see Corollary \ref{col:rounding_flows}).

\subsubsection*{The Algorithm for Near-Perfect $\bb$-Matching Problem}

Once the above reduction is established, we can proceed to designing an improved algorithm for the perfect bipartite $\bb$-matching problem. This algorithm consists of two parts. 

The first one -- constituting the technical core of our paper -- is related to the (fractional) near-perfect bipartite $\bb$-matching problem, a certain relaxation of the perfect bipartite $\bb$-matching problem. To describe this task formally, let us call a $\bb$-matching $\xx$  {\em near-perfect} if its size $\onorm{\xx}$ is at least $\frac{\onorm{\bb}}{2}-\tO{m^{\frac{3}{7}}}$, i.e., it is within $\tO{m^{\frac{3}{7}}}$ additive factor of the size of a perfect $\bb$-matching. Now, given a bipartite graph $G=(P\cup Q,E)$ and demand vector $\bb$, the {\em near-perfect $\bb$-matching problem} is a task of either finding a {\em near-perfect} $\bb$-matching in $G$ or concluding that no {\em perfect} $\bb$-matching exists in that graph.

Our goal is to design an algorithm that solves this near-perfect $\bb$-matching problem in $\tO{m^{\frac{10}{7}}}$ time. To this end, 
in Sections \ref{sec:simple} and \ref{sec:improved} we prove the following theorem. 

\begin{theorem}\label{thm:interior_point_matchings}
Let $G=(V,E)$ with $V=P\cup Q$ be an undirected bipartite graph with $n$ vertices and $m$ edges and let $\bb$ be a demand vector that corresponds to a balanced $\bb$-matching instance with $\onorm{\bb}=O(m)$. In $\tO{m^{\frac{10}{7}}}$ time, one can either find a fractional near-perfect $\bb$-matching $\xx$ or conclude that no perfect $\bb$-matching exists in $G$. 
\end{theorem}

(Observe that whenever we have an instance of maximum $s$-$t$ flow problem that has $\om$ arcs and unit capacities, $\onorm{\uu}$ is exactly $\om$. So, if we apply the reduction from Theorem \ref{thm:flow_to_matchings} to that instance then the resulting $\bb$-matching problem instance will be balanced, have $m\leq 4 \om$ edges, as well as, $\onorm{\bb}\leq 4\onorm{\uu}= 4\om \leq 2m$. Therefore, we will be able to apply the above Theorem \ref{thm:interior_point_matchings} to it.)

At a very high level, our algorithm for the near-perfect $\bb$-matching problem is inspired by the way the existing interior-point method path-following algorithms  (see, e.g., \cite{Ye97,Wright97,BoydV04}) can be used to solve it. Basically, our algorithm is an iterative method that starts with some initial, far-from-optimal solution and then gradually improves this maintained solution to near-optimality (pushing it along so-called central path)  using appropriate electrical flows as a guidance. We then show how to tie the convergence rate of this process to the structure of the guiding electrical flows. At that point, one can use a simple energy-bounding argument to establish a generic convergence bound that yields an (unsatisfactory) $\tO{m^{\frac{3}{2}}}$-time algorithm.

To improve upon this bound and deliver the desired $\tO{m^{\frac{10}{7}}}$-time algorithm, we show how one can appropriately ``shape'' these guiding electrical flows to make their guidance more effective and thus guarantee faster convergence. Very roughly speaking, it turns out there is a way of changing the maintained solution to make it essentially the same from the point of view of our $\bb$-matching instance, while dramatically improving the quality of corresponding electrical flows that guide it. 

Our way of executing this idea is based on a careful composition of two techniques. One of them  corresponds to perturbing, in a certain way, the arcs that are most significantly distorting the structure of electrical flow -- this technique can be viewed as a refinement of edge removal technique of Christiano et al. \cite{ChristianoKMST11}. The other technique corresponds to preconditioning the whole solution by adding additional, auxiliary, arcs to it. These arcs are chosen so to significantly improve the conductance properties of the solution (when viewed as a graph with resistances) while not leading to too significant deformation of the final obtained solution.

\subsubsection*{Rounding Near-Perfect $\bb$-Matchings}

Finally, our final step on our way towards solving the perfect $\bb$-matching problem (and thus the maximum $s$-$t$ flow problem) is related to turning the approximate and fractional answer returned by the algorithm from Theorem \ref{thm:interior_point_matchings} into an exact and integral one. To this end, note that if that algorithm returned a near-perfect $\bb$-matching that was integral, there would be a standard way to either turn it into a perfect $\bb$-matching or conclude that no such perfect $\bb$-matching exists. Namely, one could just use repeated augmenting path computations. It is well-known that given an integral $\bb$-matching, one can perform,  in $O(m)$ time, an augmenting path computation that either results in increasing the size of our $\bb$-matching by one, or concludes that no further augmentation is possible (and thus no perfect $\bb$-matching exists). So, as our initial near-perfect $\bb$-matching has size at least $\frac{\onorm{\bb}}{2}-\tO{m^{\frac{3}{7}}}$, after at most $\tO{m^{\frac{3}{7}}}$ iterations, i.e., in time $\tO{m^{\frac{10}{7}}}$, we would get the desired answer. 

Unfortunately, the above approach can fail completely once our near-perfect $\bb$-matching is fractional. This is so, as in this case we do not have any meaningful lowerbound on the progress on the size of the $\bb$-matching brought by the augmenting path computation. 

Therefore, to deal with this issue, we develop the last ingredient of our algorithm: a nearly-linear time procedure that allows one to round fractional $\bb$-matchings.  More precisely, in Section \ref{sec:rounding}, building on the work of Goel et al. \cite{GoelKK10} (see Theorem \ref{thm:regular_bipartite_matchings}), we establish the following theorem.

\begin{theorem}\label{thm:rounding_matchings}
Let $G=(V,E)$ be an undirected bipartite graph with $m$ edges and let $\bb$ be a demand vector, if $\xx$ is a fractional $\bb$-matching in $G$ of size $k=\onorm{\xx}$ then one can find in $\tO{m}$ time an integral $\bb$-matching in $G$ of size $\floor{k}$.
\end{theorem}

Clearly, if we apply the above rounding method to the fractional near-perfect matching $\xx$ computed by the algorithm from Theorem \ref{thm:interior_point_matchings}, it will give us an integral $\bb$-matching $\xx^*$ whose size is still at least $\frac{\onorm{\bb}}{2}-\tO{m^{\frac{3}{7}}}$. So, the augmenting path-based approach we outlined above will let us obtain the desired integral and exact answer to the perfect $\bb$-matching problem within the desired time bound. 

In the light of all the above, we see that combining all the above pieces indeed yields an $\tO{m^{\frac{10}{7}}}$-time algorithm for the perfect bipartite $\bb$-matching problem in graphs with $\onorm{\bb}=O(m)$. Now, using the reduction from Theorem \ref{thm:flow_to_matchings}, this gives us the analogous algorithm for the maximum $s$-$t$ flow problem in unit-capacity graphs and that, in turn, results in an algorithm for the bipartite matching problem. So, both Theorem \ref{thm:main} and Theorem \ref{thm:main_matchings} hold.

\subsubsection*{Rounding $s$-$t$ Flows}

Finally, we mention the other byproduct of our techniques -- the fast rounding procedure for flows. Namely, using the reduction described in Theorem \ref{thm:flow_to_matchings} and the rounding from Theorem \ref{thm:rounding_matchings} we can obtain a fast rounding procedure not only for fractional $\bb$-matchings but also for fractional $s$-$t$ flows. Specifically, the proof of the following corollary appears in Appendix \ref{app:col_rounding_flows}. 

\begin{corollary}\label{col:rounding_flows}
Let $G=(V,E,\uu)$ be a directed graph with capacities and let $\ff$ be some feasible fractional $s$-$t$ flow in $G$ of value $F$. In $\tO{m}$ time, we can obtain out of $\ff$ an integral $s$-$t$ flow $\ff^*$ of value $\floor{F}$ that is feasible in $G$. 
\end{corollary} 

\noindent Again, we note that a very similar rounding result was independently obtained by Khanna et al. \cite{KhannaKL13}. 

%% file: files/reduction.tex
\section{From Maximum Flows to Perfect Matchings}\label{sec:reduction}

In this section, we show how to reduce the maximum $s$-$t$ flow problem in a directed capacitated graph $G=(V,E,\uu)$ to solving $O(\log \onorm{\uu})$ balanced instances of the perfect bipartite $\bb$-matching problem, i.e., we prove Theorem \ref{thm:flow_to_matchings}. We note that our reduction can be seen as an adaptation of the reduction of the maximum vertex-disjoint $s$-$t$ path problem to the bipartite matching problem due to Hoffman \cite{Hoffman60} -- cf. Section 16.7c in \cite{Schrijver03}.

To this end, let $G=(V,E,\uu)$ with $n=|V|$ vertices and $m=|E|$ arcs, as well as, the source $s$ and sink $t$ be our input instance of the maximum $s$-$t$ flow problem. Without loss of generality, we can assume that there is no arcs entering $s$ and no arcs leaving $t$, as these arcs do not affect the maximum $s$-$t$ flow. Also, let $F^*$ be the value of the maximum $s$-$t$ flow in $G$.  

\subsection{The Reduction}

\begin{figure}[ht]
\centering
\vspace{8pt}
\includegraphics[width=\textwidth]{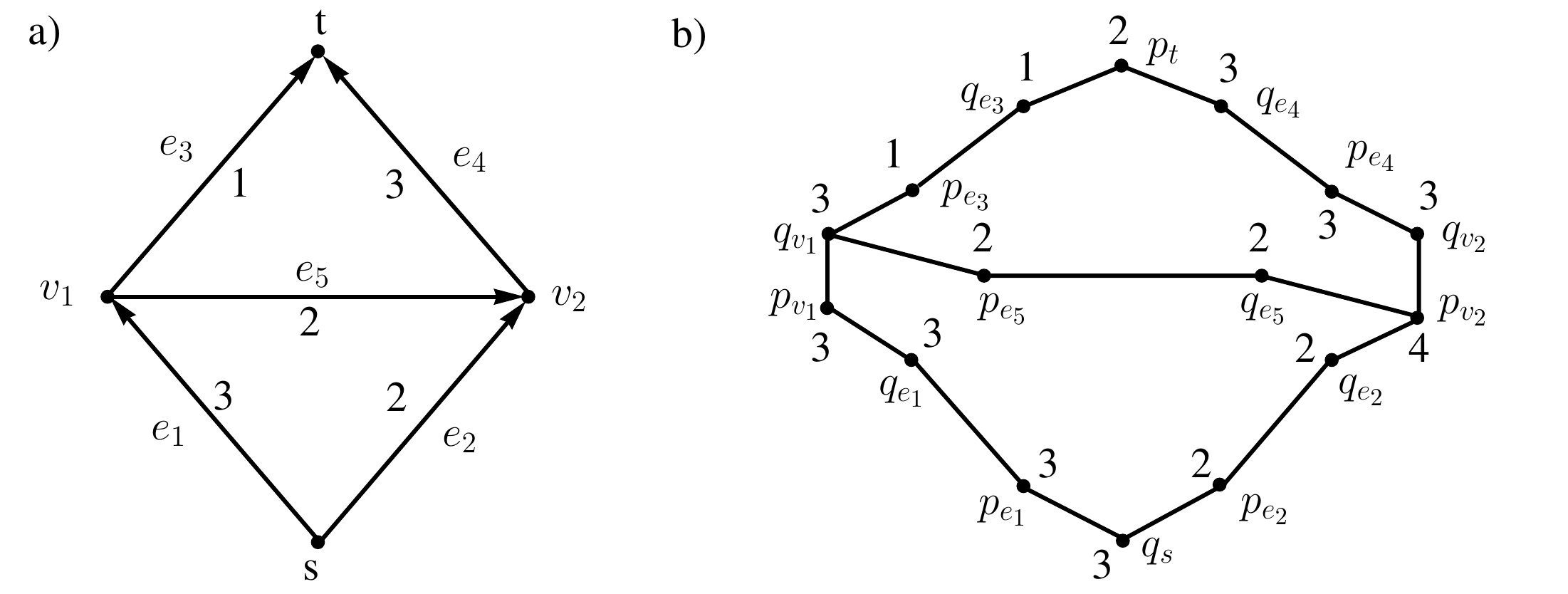}
\vspace{8pt}
\caption{{\bf a)} An example directed $s$-$t$ flow instance $G$. Numbers next to arcs denote their capacities.  {\bf b)} The $\bb$-matching instance corresponding to the example from a) in case of $F=2$. Here, numbers next to vertices denote their demands. }
\label{fig:reduction_example}
\end{figure}

We show that for any integral value of $F$, we can setup, in $\tO{m}$ time, a balanced bipartite $\bb$-matching problem instance, for some demands $\bb$ and bipartite graph $\oG=(P\cup Q,\oE)$, such that: (1) there will be a perfect $\bb$-matching in $\oG$ if there is a feasible $s$-$t$ flow of value $F$ in $G$; and (2) given a perfect $\bb$-matching in $\oG$ one can recover in $\tO{m}$ time an $s$-$t$ flow of value $F$ that is feasible in $G$. Observe that once such a reduction is designed, Theorem \ref{thm:flow_to_matchings} will follow by noticing that $1\leq F^* \leq \onorm{\uu}$ and applying a simple binary search strategy to find the value of $F^*$ and extract the corresponding maximum $s$-$t$-flow. 

Given the input graph $G=(V,E,\uu)$, source $s$, sink $t$ and the value of $F$, the construction of our desired balanced bipartite $\bb$-matching instance $\oG=(P\cup Q,\oE)$ is as follows. First, for each arc $e\in E$, we create two vertices $p_e\in P$ and $q_e\in Q$ and an edge $(p_e,q_e)$ between them, as well as, we set the demand $b_{p_e}$ and $b_{q_e}$ of these vertices to $u_e$. Next, for every vertex $v$ of $G$ other than $s$ and $t$, we add a vertex $p_v$ to $P$ and a vertex $q_v$ to $Q$. Also, we create an edge $(p_v,q_v)$, as well as, an edge $(p_v,q_{e})$ (resp. $(q_v,p_e)$) for every arc $e$ that is incoming to (resp. outgoing of) $v$ in $G$. We set the demands $b_{p_v}$ (resp. $b_{q_v}$) to be equal to $\sum_{e\in E^+(v)} u_e$ (resp. $\sum_{e\in E^-(v)} u_e$). Finally, we create a vertex $q_s\in Q$ (resp. $p_t\in P$) and add an edge $(q_s,p_e)$ (resp. $(q_e,p_t)$ for each arc $e$ that is leaving $s$ (resp. incoming to $t$) in $G$. We put the demand $b_{q_s}$ (resp. $b_{p_t}$) to be $(\sum_{e\in E^-(s)} u_e) - F$ (resp. $(\sum_{e\in E^+(t)} u_e) - F$). (Note that we can assume here that both these quantities are non-negative as both $\sum_{e\in E^-(s)} u_e$ and $\sum_{e\in E^+(t)} u_e$ are obvious upperbounds on the value of $F^*$.) 

An example $s$-$t$ flow instance and the corresponding instance of the bipartite $\bb$-matching can be found in Figure \ref{fig:reduction_example}.

To see that this instance is balanced, note that every edge $h$ of $\oG$ that is incident to some vertex $p_e$ or $q_e$ has its thickness $d(h)$ equal to $u_e=b_{p_e}=b_{q_e}$. So, the contribution of these edges to the total thickness $\sum_{h\in \oE} d(h)$ of edges of $\oG$ is at most $3 \sum_{e\in E} u_e \leq \frac{3}{2} \onorm{\bb}$. On the other hand, the only edges that are not incident to some $p_e$ or $q_e$ are the ones of the form $(p_v,q_v)$. However, the total contribution of these edges to the total thickness is at most 
\[
\sum_{v\neq s,t} \min\{\sum_{e\in E^+(v)} u_e,\sum_{e\in E^-(v)} u_e\}\leq \sum_{v\neq s,t} \frac{\sum_{e\in E^+(v)} u_e +\sum_{e\in E^-(v)} u_e}{2} \leq \onorm{\uu} \leq \onorm{\bb},  
\]
as needed.

Now, the proof of correctness of this reduction appears in Appendix \ref{app:reduction}. 

%% file: files/basic_algorithm.tex
\section{Basic \texorpdfstring{$\tO{m^{\frac{3}{2}}}$-Time}{\~O(m\textasciicircum(3/2))-Time} Algorithm for Bipartite $\bb$-Matching Problem}
\label{sec:simple}

Over the next two sections, we prove Theorem \ref{thm:interior_point_matchings}. That is, we present an algorithm for the {near-perfect} bipartite $\bb$-matching problem in the setting where the input instance is balanced (see \eqref{eq:def_balanced}) and $\onorm{\bb}$ is $O(\cm)$. In what follows we assume, for convenience, that $\onorm{\bb}$ is at most $2\cm$ and that the graph $\cG$ is sparse, i.e., $\cm=O(\cn)$.\footnote{It is easy to see that these assumptions are made without loss of generality. Whenever $\onorm{\bb}$ is $O(\cm)$, one can ensure that $\onorm{\bb}\leq 2 \cm$ and $\cm=O(\cn)$ by adding an appropriate -- but still $O(\cm)$ -- number of dummy copies of complete bipartite $K_{6,6}$ graph with uniform demands. Adding each such dummy isolated copy brings the ratio of $\onorm{\bb}$ and $\cm$, as well as, of $\cm$ to $n$ down towards $\frac{18}{12}$, while never leading to violation of the balance condition \eqref{eq:def_balanced} and preserving the $\bb$-matching structure of the original input graph.}

In this section, we show a basic algorithm that runs in $\tO{\cm^{\frac{3}{2}}}$ time. Later, in Section \ref{sec:improved}, we refine this algorithm to obtain the desired running time of $\tO{\cm^{\frac{10}{7}}}$.

For the sake of clarity, in our description and analysis we assume that the nearly-linear time Laplacian system solver (see Theorem \ref{thm:vanilla_SDD_solver}) always returns an exact solution, i.e., all the electrical $\vsigma$-flows we compute are exact. We discuss how to handle the approximate nature of the solver's output in Appendix \ref{app:inexact_elec_flow_disc}. 

\subsubsection*{From \texorpdfstring{$\bb$-Matching}{b-Matching} to Min-Cost \texorpdfstring{$\vsigma$-flow}{Sigma-flow}}
Let us fix our instance of the bipartite $\bb$-matching problem in bipartite graph $G=(V,E)$ with $V=P\cup Q$. We will solve our $\bb$-matching instance by reducing it to a task of finding a minimum-cost $\hvsigma$-flow in a certain related directed graph $\hG=(\hV,\hE,\hll)$ with $\hll$ being a length vector. 

\begin{figure}[ht]
\centering
\vspace{8pt}
\includegraphics[width=\textwidth]{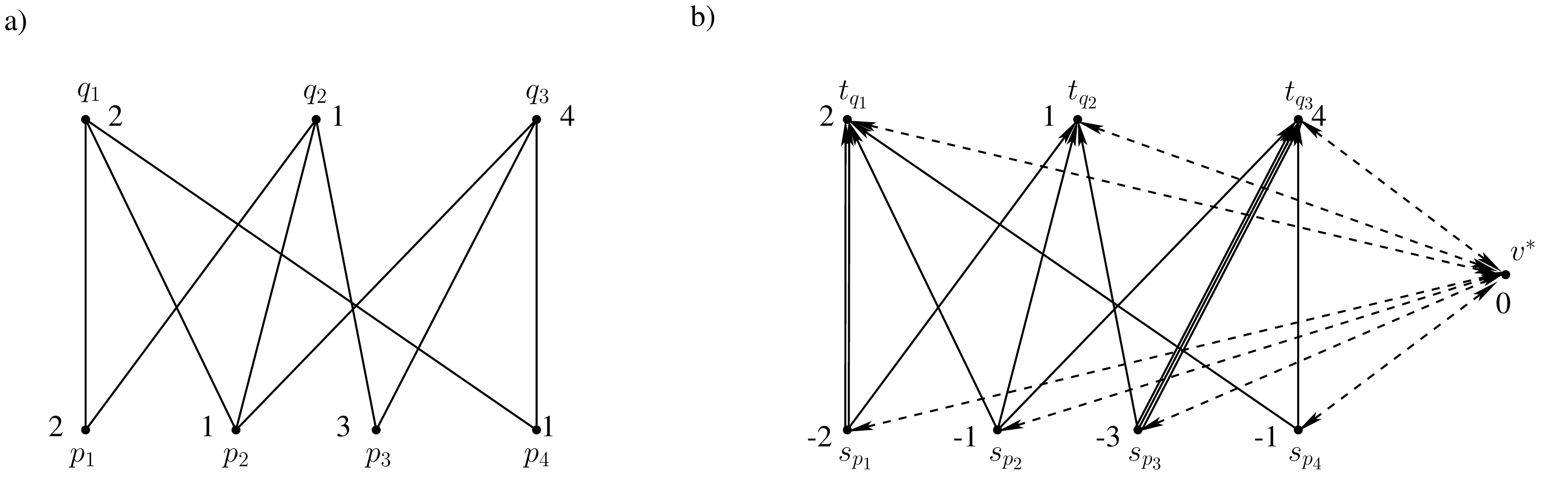}
\vspace{8pt}
\caption{{\bf a)} An example instance of bipartite $\bb$-matching problem. Numbers next to vertices represent their demands.  {\bf b)} The minimum-cost $\hvsigma$-flow problem instance corresponding to the example from a). All arcs have cost $\hl_e$ equal to $1$ and the numbers next to vertices denote their demands in $\hvsigma$. There are two parallel copies of the arc $(s_{p_1},t_{q_1})$ and three parallel copies of the arc $(s_{p_3},t_{q_3})$. Also, each dashed arc represents two arcs that have the same endpoints but opposite orientation. }
\label{fig:min_cost_example}
\end{figure}

The reduction is performed as follows (see Figure \ref{fig:min_cost_example} for an example). The vertex set $\hV$ of the graph $\hG$ consist of a special vertex $v^*$, as well as, vertices $s_p$ (resp. $t_q$), for every vertex $p\in P$ (resp. $q\in Q$) of the graph $\cG$. Next, for every edge $e=(p,q)$ in $\cG$, we add to $\hG$ $d(e)$ copies of an arc $(s_p,t_q)$, where we recall that $d(e):=\min\{b_p,b_q\}$ is the thickness of $e$. Finally, for each vertex $p\in P$ (resp. $q\in Q$) of $\cG$, we add to $\hG$ arcs $(s_p,v^*)$ and $(v^*,s_p)$ (resp. $(v^*,t_q)$ and $(t_q,v^*)$). We set the lengths $\hl_e$ of all arcs $e$ to $1$. 

To gain some intuition on this reduction, note that if a perfect $\bb$-matching indeed exists in $\cG$ then the flow that encodes it in $\hG$ is fully supported on the arcs $(s_p,t_q)$ and does not send more than one unit of flow on any of these arcs. So, the purpose of including the extra vertex $v^*$ and the arcs incident to it is to support (and appropriately penalize) the initial and intermediate solutions as they approach optimality.  

Also, observe that this new graph has $\hn:=n+1$ vertices and, due to our $\bb$-matching instance being balanced, we have that the total number $\hm$ of arcs is at most
\[
2n+\sum_{e=(p,q)\in \cG} d(e)\leq 2n+O(\cm) = O(\cm).
\]
So, bounding our running time in terms of $\hm$ provides a bound in terms of the number of edges $\cm$ of our original $\bb$-matching instance that is asymptotically the same. 

Now, consider a demand vector $\hvsigma$ that has surplus of $b_p$ at each vertex $s_p$, a deficit of $b_q$ at each vertex $t_q$ and a zero demand at vertex $v^*$. (Note that such a demand vector will be valid, i.e., $\sum_v \hsigma_v=0$, as we can assume that $\sum_p b_p=\sum_q b_q$ -- otherwise it would be impossible to have a perfect $\bb$-matching in $\cG$.) We claim that any near-optimal $\hvsigma$-flow gives us a solution to our near-perfect $\bb$-matching instance. (Recall from Section \ref{sec:outline} that a $\bb$-matching is near-perfect if its size is at least $\frac{\onorm{\bb}}{2}-\tO{\hm^{\frac{3}{7}}}$. Although, in the lemma below it suffices that we have a slack of only $\frac{1}{2}$ instead of $\tO{\hm^{\frac{3}{7}}}$.)

\begin{lemma}\label{lem:mincost_to_matchings}
Given any feasible $\hvsigma$-flow $\ff$ in $\hG$ whose cost $\hll(\ff)$ is within additive $\frac{1}{2}$ of the optimum, in $\tO{\hm}$ time, we can either compute a (fractional) near-perfect $\bb$-matching $\xx$ in $\cG$ or conclude that no perfect $\bb$-matching exists in $\cG$. 
\end{lemma} 

\begin{proof}
First, observe that if there exists a perfect $\bb$-matching $\xx^*$ in $\cG$ then a flow $\ff^*$ that just puts, for each $e=(p,q)$ of $\cG$, $\frac{x^*_e}{d(e)}\leq 1$ units of flow on each (of $d(e)$) copies of the arc $(s_p,t_q)$ in $\hG$, is a feasible $\hvsigma$-flow with cost $\frac{\onorm{\bb}}{2}$. (Recall that in the minimum-cost problem we assume that arc capacities are infinite, thus feasibility condition \eqref{eq:capacity_constraints} boils down to non-negativity of all $f_e^*$s.) So, we can assume that our $\hvsigma$-flow $\ff$ has its cost $\hl(\ff)$ at most $\frac{\onorm{\bb}}{2}+\frac{1}{2}$. (Otherwise, we know that there is no perfect $\bb$-matching in $\cG$.)

Now, given any feasible $\hvsigma$-flow in $\hG$, we can decompose it into a collection of flow-paths and flow-cycles, where each of these flow-paths transports some amount of flow from some vertex $s_p$ to some vertex $t_q$. By our construction of the graph $\hG$, each such flow-path has to have a length at least $1$. On the other hand, if this flow-path is indeed of length exactly $1$ then it has to correspond to a single arc $(s_p,t_q)$ that reflects the existence of edge $(p,q)$ in $\cG$. As a result, our feasible $\hvsigma$-flow $\ff$ in $\hG$ has to have its cost $\hl(\ff)$ to be at least $\frac{\onorm{\bb}}{2}$ and, furthermore, $\hl(\ff)-\frac{\onorm{\bb}}{2}$ is an upper bound on the total amount of flow in $\ff$ that is not transported over the direct one-arc flow paths (and thus passes through the vertex $v^*$).  

So, as we argued that the cost of $\ff$ has to be at most $\frac{\onorm{\bb}}{2}+\frac{1}{2}$, there is only at most $\frac{1}{2}$ units of flow in $\ff$ that passes through the vertex $v^*$. Now, to extract the desired (fractional) near-perfect $\bb$-matching $\xx$, we just take $x_e=f_{(s_p,t_q)}$, for each edge $e=(p,q)$ in $\cG$. Clearly, the size of such fractional matching is at least $\frac{\onorm{\bb}}{2}-\frac{1}{2}$, which is well above our lowerbound of $\frac{\onorm{\bb}}{2}-\tO{m^{\frac{3}{7}}}$ for a near-perfect matching. Also, our construction works in $\tO{\hm}$ time, as desired. 
\end{proof}

\subsubsection*{Slack Variables}

In the light of the above, our goal now is to compute the near-optimal solution to our minimum-cost $\hvsigma$-flow problem instance in the graph $\hG$. Our approach to this task is inspired by so-called path-following interior-point methods \cite{Ye97,Wright97,BoydV04}. At a very high level, we will start with certain initial solution that is far from being optimal, and then we will gradually improve -- in an iterative manner -- its quality until close-to-optimal solution is obtained. This gradual improvement will be performed in a very specific way. It will always try to push the current solution further down so-called central path. 

Before we can define the central path, let us first mention that, in general, there are two natural ways of tracking the progress of a current solution towards optimality. One of them is purely primal and relies on just maintaining a feasible solution $\ff$ and comparing its cost against some estimate of the cost of the optimal solution. The second one -- and the one that we will actually use here -- is based on primal-dual paradigm. Namely, in addition to maintaining a feasible primal solution $\ff$, we will also keep a dual feasible solution $\yy$. This dual solution provides an embedding of all the vertices in $\hG$ into a line, i.e., $\yy$ just assigns a real number $y_v$ to each vertex $v$ of $\hG$. Its feasibility condition is that for any arc $e=(v,w)$ of $\hG$ it should be the case that its {\em slack variable} $s_e:=\hl_e-y_w+y_v$ is always non-negative, i.e., that the length of the arc $e$ in this embedding is never larger than its length according to the length vector $\hll$. 

Before we proceed further, we note that the dual solution $\yy$ is uniquely determined -- up to a translation -- by the vector $\ss$ (given the length vector $\hll$). So, for notational convenience, from now on, we will describe the dual solutions in terms of the vector $\ss$ instead of $\yy$.

\subsubsection*{Duality Gap}

It is not hard to see that any feasible dual solution $\ss$ provides a lower-bound on the cost of the optimal solution (after all, this is just a consequence of weak duality). In particular, one has that for any pair $(\ff,\ss)$ of feasible primal and dual solutions, the so-called {\em duality gap}, i.e., the difference between the upper bound on the value of optimal solution that is provided by the primal solution $\ff$ and the lower bound provided by the dual solution $\ss$ is exactly 
\[
\ff^T \ss = \vmu^T \onev = \sum_{e} \mu_e,
\]
where $\mu_e:=f_es_e$, for each arc $e$, and $\onev$ is all-ones vector (of dimension $\hm$). 

This means that one can obtain a close-to-optimal solution by devising a procedure that (quickly) converges to a pair of primal and dual solutions $(\ff,\ss)$ whose duality gap $\onorm{\vmu}$ is small (in our case, at most $\frac{1}{2}$).

\subsubsection*{$\gamma$-Centered Solutions and the Central Path}

To describe in more detail the convergence process we will employ, let us associate with each arc $e$ a {\em measure} $\nu_e\geq 1$. One can view $\nu_e$ as a certain notion of importance of a given arc. (The motivation behind introducing this notion will be clear later.) We will always make sure that the measures of arcs are not smaller than $1$ and also that their total sum is never too large. That is, we will make sure to maintain the following invariant.

\begin{invariant}
\label{inv:measure_upperbound}
We have that $\vnu^T \onev = \sum_{e} \nu_e \leq 4\hm$ and for each arc $e$, $\nu_e\geq 1$.
\end{invariant}

We want to note that when discussing the preservation of the above invariant we will only focus on ensuring that the upperbound is not violated. The fact that $\nu_e\geq 1$ for all arcs $e$ will be automatically enforced as we will make sure that the initial measure of all the arcs is always at least $1$ and our algorithm will never decrease any measures -- they only might increase.

\paragraph{$\gamma$-centered solutions.} Now, let us define, for each arc $e$, $\hmu_e:=\frac{\mu_e}{\nu_e}=\frac{f_es_e}{\nu_e}$ to be the normalized value of $\mu_e$ and let 
\begin{equation}
\label{eq:def_hmu}
\hmu(\ff,\ss,\vnu):=\frac{\sum_{e} f_e s_e}{\sum_{e} \nu_e}=\frac{\sum_{e} \mu_e}{\sum_{e} \nu_e}=\frac{\sum_{e} \nu_e \hmu_e}{\sum_{e} \nu_e}
\end{equation}
be the weighted average value of $\hmu_e$ with weights given by the measures $\vnu$.

We will call a solution $(\ff,\ss, \vnu)$ (where $\vnu$ represents the associated measures) {\em $\gamma$-centered}, for some $\gamma\geq 0$, if 
\begin{equation}
\label{eq:def_centrality}
\norm{\hvmu-\hmu(\ff,\ss,\vnu)\onev}{\vnu,2}=\sqrt{\sum_e \nu_e(\hmu_e-\hmu(\ff,\ss,\vnu))^2} \leq \gamma\hmu(\ff,\ss,\vnu),
\end{equation}
where, for a given vector $\xx\in \RR^{\hm}$,
\begin{equation}
\label{eq:def_norm_nu}
\norm{\xx}{\vnu,p}:=\left({\sum_e \nu_e x_e^p}\right)^{\frac{1}{p}},
\end{equation} 
i.e., $\norm{\xx}{\vnu,p}$ is the $\ell_p$-norm of the vector $\xx$ reweighed by the measures $\vnu$. 

Note that in a $0$-centered solution $(\ff,\ss,\vnu)$ we have all $\hmu_e$ equal to $\hmu(\ff,\ss,\vnu)$. More generally, a simple but very useful observation is that
\begin{fact}
\label{fa:central_vs_max_min}
For any $\gamma$-centered solution $(\ff,\ss,\vnu)$ we have that
\[
(1-\gamma) \hmu(\ff,\ss,\vnu) \leq \hmu_e = \frac{f_e s_e}{\nu_e} \leq (1+\gamma) \hmu(\ff,\ss,\vnu),
\]
for each arc $e$. 
\end{fact}

\paragraph{$\hmu(\ff,\ss,\vnu)$ as a measure of progress.} The quantity $\hmu(\ff,\ss,\vnu)$ will be important to us for one more reason. It will constitute our measure of progress on the quality of our maintained solution. To see why it indeed can serve this role, recall that by Invariant \ref{inv:measure_upperbound} we have that 
\begin{equation}
\label{eq:duality_gap_bound}
\ff^T \ss = \sum_e \mu_e = \hmu(\ff,\ss,\vnu) (\sum_e \nu_e) \leq 4 \hmu(\ff,\ss,\vnu) \hm.
\end{equation}
So, if our goal is to obtain a solution whose duality gap is at most $\frac{1}{2}$ we just need to make sure that the corresponding value of $\hmu(\ff,\ss,\vnu)$ is at most $\frac{1}{8\hm}$. 

The main reason why we choose to measure our progress in terms of $\hmu(\ff,\ss,\vnu)$ instead of the actual duality gap $\ff^T\ss$ is that in our algorithm we will sometime end up increasing measures of arcs. Such increases lead to an increase of the duality gap, so measuring our progress in terms of $\ff^T\ss$ would require dealing with such local non-monotonicity of this quantity. Continently, once we focus on keeping track of  $\hmu(\ff,\ss,\vnu)$ (and ensure that Invariant \ref{inv:measure_upperbound} is never violated), these issues will be avoided.

\paragraph{The central path.} Finally, after introducing the above definitions, we can define the central path to be the set of all the $0$-centered solutions.\footnote{Strictly speaking, in the literature, the central path corresponds to $0$-centered solutions with the measures of all arcs being one.} One can show that this set constitutes an actual path in feasible space that spans all the $0$-centered solutions and passes arbitrarily close to (but never reaches) an optimal solution to our minimum cost flow problem. This explains the name of ``path-following'' interior-point methods, as they start with some initial $0$-centered solution and gradually advance along the central path to get increasingly more optimal $\gamma$-centered solution for some small fixed $\gamma$.

 \subsubsection*{Traversing the Central Path with Electrical Flows}

Motivated by this path-following approach, our algorithm for computing near-optimal solution to the minimum-cost $\hvsigma$-flow problem will start with some $0$-centered solution $(\ff^0,\ss^0, \vnu^0)$ that has fairly large value of $\hmu(\ff^0,\ss^0,\vnu^0)$ (and thus is far from being optimal). Then, we will devise a sequence of solutions $(\ff^t,\ss^t, \vnu^t)$, where $t$ is the step index, that have increasingly smaller value of $\hmu(\ff^t,\ss^t,\vnu^t)$ (and thus, indirectly, the duality gap) while making sure that they always are $\hgamma$-centered for some small constant $\hgamma:= \frac{1}{400}$. This way, our algorithm will eventually converge to the desired close-to-optimal solution.

To implement this approach, we start with the following lemma that shows we can get the initial $0$-centered solution $(\ff^0,\ss^0,\vnu^0)$ -- its proof appears in Appendix \ref{app:initial_solution}. 
\begin{lemma}
\label{lem:initial_solution}
There exists an explicit $0$-centered primal-dual feasible solution $(\ff^0,\ss^0,\vnu^0)$ with $\sum_e \nu^0_e \leq 3\hm$ and $\hmu(\ff^0,\ss^0,\vnu^0)=1$. 
\end{lemma}

Note that the bound on the total measure of the arcs ensures that the Invariant \ref{inv:measure_upperbound} is preserved. Furthermore, there is a slack of at least $\hm$ remaining between $\sum_e \vnu^0$ and the upperbound of $4\hm$ from Invariant \ref{inv:measure_upperbound}. It will be used to accommodate future measure increases in our improved algorithm (see Section \ref{sec:improved}).  

We now proceed to explaining how given some $\hgamma$-centered solution $(\ff^t,\ss^t,\vnu^t)$, we can modify it to obtain a $\hgamma$-centered solution $(\ff^{t+1},\ss^{t+1},\vnu^{t+1})$ that has a smaller value of $\hmu(\ff^t,\ss^t,\vnu^t)$. 

\paragraph{The associated flow $\hff^t$.} For a given solution $(\ff,\ss,\vnu)$ let us call it {\em $\vsigma$-feasible}, for some demand vector $\vsigma$, if it is dual feasible (i.e., $\ss\geq 0$) and if $\ff$ is a feasible $\vsigma$-flow. (So, a $\hvsigma$-feasible solution is a solution that is primal-dual feasible for our minimum-cost $\hvsigma$-flow problem.) Next, given a $\vsigma$-feasible solution $(\ff,\ss,\vnu)$, let us define an {\em associated electrical flow $\hff$} to be the electrical $\vsigma$-flow in (the undirected projection of) $\hG$ determined by resistances $\rr$ that are given as
\begin{equation}
\label{eq:hf_resistances}
r_e := \frac{s_e}{f_e} = \frac{\mu_e}{(f_e)^2},
\end{equation}
for arc $e$. (Whenever we use this definition, it will be always the case that all $f_e$s are positive and thus the resistances $r_e$ are well-defined.)

 \paragraph{Making an improvement step.} The central object in our procedure for taking an improvement step will be the electrical flow $\hff^t$ that is associated with the solution $(\ff^t,\ss^t,\vnu^t)$. The fundamental property of this flow is that it allows us to simultaneously update our solution $(\ff^t,\ss^t,\vnu^t)$ {both} in the primal (flow) space -- via the flow $\hff^t$ itself -- and in the dual (line embedding) space -- via the vertex potentials $\hvphi^t$ that induced $\hff^t$ (see \eqref{eq:potential_flow_def}). (In Section \ref{sec:proof_main_interior_point}, we provide a detailed description of the whole improvement step.)

 As we will see, such a guided update not only decreases the duality gap of our solution, but also perfectly maintains its centering when only first-order terms (i.e., terms linear in the updates) are accounted for. Unfortunately, the second-order terms (i.e., the ones depending on the product of primal and dual updates) can disturb the centering. So, to be able to control this deficiency, we need to ensure that the step size $\delta^t$ that governs the ``aggressiveness'' of the improvement step is sufficiently small. 
 
Of course, on the other hand, it is important for us to have this step be as large as possible. After all, the extent of our duality gap improvement -- and thus overall convergence rate of our algorithm -- is directly proportional to this size. So, it is crucial for us to develop a good grasp on how the size of that step relates to the properties of the flow $\hff^t$. 

To this end, let us define, for some -- not necessarily feasible -- flow $\ff$ and a positive vector $\xx>0$, $\vrho(\ff,\xx)$ to be the vector of congestions inflicted in $\hG$ by $\ff$ with respect to capacities given by $\xx$.  That is, 
\begin{equation}
\label{eq:def_of_rho}
\rho(\ff,\xx)_e:=\frac{|f_e|}{x_e},
\end{equation}
for each arc $e$ in $\hG$. 

Now, in Section \ref{sec:proof_main_interior_point}, we present a precise implementation and analysis of our update step. (This implementation can be viewed as a direct analogue of the update steps of path-following interior-point methods.) The result of this analysis is presented in the following theorem, which, in particular, ties the congestion vector $\vrho(\hff^t,\ff^t)$ inflicted by the electrical flow $\hff^t$ with respect to the primal solution $\ff^t$, to an upperbound on the size $\delta^t$ of the improvement step. 

\begin{theorem}
\label{thm:main_interior_point}
Let $(\ff^t,\ss^t,\vnu^t)$ be a solution that is $\hgamma$-centered and $\hvsigma$-feasible, and let $\hff^t$ be the associated electrical flow. We can compute in $\tO{\hm}$ time a $\hgamma$-centered and $\hvsigma$-feasible solution $(\ff^{t+1},\ss^{t+1},\vnu^{t+1})$ with $\hvmu(\ff^{t+1},\ss^{t+1},\vnu^{t+1})\leq (1-\delta^t) \hvmu(\ff^{t},\ss^{t},\vnu^{t})$, as long as, 
\[
0<\delta^t\leq \min\left\lbrace \frac{\sqrt{\hgamma}}{\norm{\vrho(\hff^t,\ff^t)}{\vnu^t,4}},\frac{1}{2}\right\rbrace.
\]
Furthermore, we have that the measures do not change, i.e., $\vnu^{t+1}=\vnu^t$, and if for each arc $e$, we define $(1+\kappa^t_e):=\frac{(1-\delta^t)s_e^{t+1}f_e^{t}}{f_e^{t+1}s_e^t }=\frac{(1-\delta^t)r_e^{t+1}}{r_e^t}$ and $(1+\okappa^t_e):=\frac{(1-\delta^t)f_e^{t}}{f_e^{t+1}}$ to make $\vkappa^t$ (resp. $\ovkappa^t$) reflect the relative change (scaled by $(1-\delta^t)$) of resistances $\rr^t$ (resp. flows $\ff^t$) then $\inorm{\vkappa^t},\inorm{\ovkappa^t}\leq \frac{1}{2}$ and
\[
|\kappa_e^t|,|\okappa_e^t| \leq 4 (\delta^t\rho(\hff^t,\ff^t)_e + \hkappa^t_e),
\]
for some vector $\hvkappa^t$ with $\norm{\hvkappa^t}{\vnu^t,2}\leq \frac{1}{16}$.
\end{theorem} 

So, we see that the allowed size $\delta^t$ of the improvement steps is proportional to how much the guiding flow $\hff^t$ resembles the current primal solution $\ff^t$. Thus, for example, if there is some arc $e$ that flows much larger flow in $\hff^t$ than in $\ff^t$, i.e., an arc $e$ with large value of $\rho(\hff^t,\ff^t)_e$, this arc will be severely penalized by the $\ell_4$-norm measuring the quality of the resemblance. 

Also, it is worth pointing out that it is very important that the above bound is based on $\ell_4$ instead, say $\ell_2$ norm. In fact, one can show (see Lemma \ref{lem:energy_bound}) that in case of our problem the $\ell_2$ norm of congestion vector is always $\Omega(\hm^{\frac{1}{2}})$. So, using $\ell_2$ norm would not lead to any improvement over the $\Omega(\hm^{\frac{1}{2}})$ iteration bound.

\subsection{Bounding the Running Time}

At this point, we want to present a fairly elementary proof of $\delta:=\Omega(\hm^{-\frac{1}{2}})$ lowerbound on our allowed improvement step size $\delta^t$. Note that once we achieve that then, by Lemma \ref{lem:initial_solution} and Theorem \ref{thm:main_interior_point}, we will have that the value of our measure of progress $\hmu(\ff^t,\ss^t,\vnu^t)$ after $T$ steps is at most
\[
\hmu(\ff^T,\ss^T,\vnu^T)\leq \prod_{t=1}^{T} (1-\delta^t) \leq (1-\delta)^T.
\]
So, by setting $T:=\delta^{-1} \log 8\hm=\tO{\hm^{\frac{1}{2}}}$, we recover the $O(\hm^{\frac{1}{2}})$ iterations convergence bound of interior-point methods. This leads to a simple $\tO{\hm \delta^{-1}}=\tO{\cm^{\frac{3}{2}}}$-time procedure that produces a solution with duality gap at most 
\[
4\hm \hmu(\ff^T,\ss^T,\vnu^T) \leq 4\hm (1-\delta)^T \leq \frac{1}{2},
\]
where we used Invariant \ref{inv:measure_upperbound} (see \eqref{eq:duality_gap_bound}). This, in turn, by Lemma \ref{lem:mincost_to_matchings} provides us with a solution to our instance of near-perfect $\bb$-matching problem. 

Therefore, to conclude the analysis of the simple $\tO{\cm^{\frac{3}{2}}}$-time algorithm for the near-perfect $\bb$-matching problem, it remains to establish the claimed lowerbound on $\delta^t$. 

\paragraph{Congestion and energy.} By Theorem \ref{thm:main_interior_point}, performing such lowerbounding of $\delta^t$ boils down to upperbounding $\norm{\vrho(\hff^t,\ff^t)}{\vnu^t,4}$. To understand how the latter can be done, one should observe the following simple but crucial fact. (This fact follows from Fact \ref{fa:central_vs_max_min} and definition of the resistances $\rr^t$ \eqref{eq:hf_resistances}.)

\begin{fact}
\label{fa:rho_vs_rt}
For any $\gamma$-centered solution $(\ff^t,\ss^t,\vnu^t)$ and any flow $\hff$ in $\hG$ we have that
\[
r^t_e \hf^2_e = \frac{s_e^t}{f^t_e} \hf_e^2 \geq (1-\gamma) \nu_e^t \frac{\hmu(\ff^t,\ss^t,\vnu^t)}{(f^t_e)^2} \hf_e^2 = (1-\gamma) \nu_e^t \hmu(\ff^t,\ss^t,\vnu^t) \rho(\hff,\ff^t)_e^2,
\]
and, similarly,
\[
r^t_e \hf^2_e \leq (1+\gamma) \nu_e^t \hmu(\ff^t,\ss^t,\vnu^t) \rho(\hff,\ff^t)_e^2,
\]
for any arc $e$ in $\hE$. 
\end{fact}

Observe that the above inequalities state that -- up to a $(1\pm\gamma)$ factor -- the square of the congestion $\rho(\hff,\ff^t)_e$ incurred by an arc $e$ in the flow $\hff$ is upperbounded by
\[
\frac{r^t_e\hf^2_e}{\nu_e^t \hmu(\ff^t,\ss^t,\vnu^t)},
\]
which corresponds to normalized (by $\nu_e^t \hmu(\ff^t,\ss^t,\vnu^t)$)  contribution of the arc $e$ to the energy $\energy{\rr^t}{\hff}$ of the flow $\hff$ with respect to resistances $\rr^t$.

This simple connection between the congestion of an arc in $\hff$ and its contribution to the energy of that flow that is provided by Fact \ref{fa:rho_vs_rt} will be fundamental to the rest of our discussion. In particular, it gives us an intuition on why we even expect the guiding electrical flows $\hff^t$ to inflict small congestion with respect to $\ff^t$ and thus allow us to take a larger step size. This intuition is based on an observation that the main goal of electrical flows is to minimize energy. So, by choosing the resistances appropriately, we in some sense align this goal with our goal of making as large step sizes as possible. Roughly speaking, we are employing here the $\ell_2$ norm minimization offered by electrical flows to achieve the desired $\ell_4$-minimization corresponding to larger step sizes. 

Now, an immediate consequence of the above connection is an elementary way of upperbounding $\norm{\vrho(\hff^t,\ff^t)}{\vnu^t,4}$: we just bound the $\ell_2$-energy of the guiding electrical flow $\hff^t$ that is associated with our solution $(\ff^t,\ss^t,\vnu^t)$ and exploit the generic relationship between $\ell_2$ and $\ell_4$ norms.  

To implement this approach, let us start with the following lemma that gives us a bound on the $\ell_2$-energy of the electrical flow $\hff^t$. 

\begin{lemma}
\label{lem:hfft_energy_bound}
For any $\vsigma$-feasible solution $(\ff,\ss,\vnu)$ and associated electrical flow $\hff$, we have that
\[
\energy{\rr}{\hff}\leq \energy{\rr}{\ff} \leq 4 \hm \hmu(\ff,\ss,\vnu).
\]
\end{lemma}

\begin{proof}
The fact that $\energy{\rr}{\hff}\leq \energy{\rr}{\ff}$ follows directly from the definition of $\hff$ and the fact that electrical $\vsigma$-flow minimizes energy among all the $\vsigma$-flows (which includes $\ff$). 

Now, by definition \eqref{eq:hf_resistances}  of $\rr$ and that of $\hmu(\ff,\ss,\vnu)$ \eqref{eq:def_hmu} we have that
\[
\energy{\rr}{\ff} = \sum_e \frac{s_e}{f_e} f_e^2 = \sum_e s_e f_e = \hmu(\ff,\ss,\vnu) (\sum_e \nu_e) \leq 4 \hm \hmu(\ff,\ss,\vnu),
\]
where the last line follows by Invariant \ref{inv:measure_upperbound}.
\end{proof}

Once we establish this upperbound on $\ell_2$-energy, we simply use it to upperbound the $\ell_4$-energy of the congestion vector. Specifically, by applying Cauchy-Schwarz inequality and the fact that $\inorm{\xx}\leq \onorm{\xx}$, for any vector $\xx$, we get that
\[
\norm{\vrho(\hff^t,\ff^t)}{\vnu^t,4}^4=\sum_e \nu_e^t \rho(\hff^t,\ff^t)_e^4 \leq \left(\sum_e \sqrt{\nu_e^t} \rho(\hff^t,\ff^t)_e^2\right)^2 \leq  \left(\sum_e \nu_e^t \rho(\hff^t,\ff^t)_e^2\right)^2 = \norm{\vrho(\hff^t,\ff^t)}{\vnu^t,2}^4,
\]
where we also used the fact that $\nu_e^t\geq 1$. 

Now, to bound the $\ell_2$ norm (instead of $\ell_4$ norm) of the congestion vector, we just note that by Fact \ref{fa:rho_vs_rt} and Lemma \ref{lem:hfft_energy_bound}
\begin{equation}
\label{eq:worst_case_rho}
\left(\sum_e \nu_e^t \rho(\hff^t,\ff^t)_e^2\right)^2 \leq \left(\sum_e  \frac{\nu_e^t r^t_ef^2_e}{(1-\hgamma)\nu_e^t \hmu(\ff^t,\ss^t,\vnu^t)}\right)^2 = \left(\frac{\energy{\rr^t}{\hff^t}}{(1-\hgamma)\hmu(\ff^t,\ss^t,\vnu^t)}\right)^2\leq O(\hm^2).
\end{equation} 

Therefore, we can conclude with the following lowerbound on $\delta^t$. 

\begin{fact}
\label{fa:basic_delta_lowerbound}
For any $t$, $\delta^t\geq \frac{1}{O(\sqrt{\hm})}$.
\end{fact}

It is worth pointing out that, as we discussed before, the fact that we settled here for an $\ell_2$-norm-based (instead of an $\ell_4$-norm-based) dependence of $\delta^t$ on the congestion vector $\vrho(\hff^t,\ff^t)$, this $\frac{1}{O(\sqrt{\hm})}$ lowerbound is the best possible to achieve with this approach. Therefore, to have any hope of obtaining an improvement that goes beyond this bound (as we will do in the next section), we crucially require to be working with $\ell_4$-norm-based (instead of only $\ell_2$-norm-based) arguments. 
 

%% file: files/improved_algorithm.tex
\section{An Improved Algorithm for Bipartite $\bb$-Matching Problem}
\label{sec:improved}

After setting up our primal-dual framework and presenting the $\tO{\cm^{\frac{3}{2}}}$-time algorithm in the previous section, we can now proceed to developing our improved algorithm with the running time of $\tO{\cm^{\frac{10}{7}}}=\tO{\cm^{\frac{3}{2}-\eta}}$, for $\eta:=\frac{1}{14}-o(1)$.

Given our analysis and discussion in the previous section, a tempting approach to obtaining such an improved bound could be trying to simply tighten our analysis performed there (e.g., by taking advantage of $\ell_4$-norm-based instead of only $\ell_2$-norm-based arguments) and thus improve the worst-case lowerbound on $\delta^{t}$ that we established (cf. Fact \ref{fa:basic_delta_lowerbound}). 

It turns out, however, that getting an improved bound is not merely a matter of performing a better analysis. In the worst-case, our $O(\hm^{-\frac{1}{2}})$ bound is actually tight. After all, if there is an arc that incurs $\Omega(\hm^{\frac{1}{2}})$ congestion in the associated electrical flow, the resulting $\ell_4$-norm of the congestion vector will be $\Omega(\hm^{\frac{1}{2}})$. So, even though the connection between congestion and energy we established before (see Fact \ref{fa:rho_vs_rt} and Lemma \ref{lem:hfft_energy_bound}) tells us that there cannot be too many such arcs (as each one of them would need to contribute a very significant, $\Omega(1)$, fraction of the total energy of the electrical flow), having just one such arc is already enough to prevent us from taking larger than $O(\hm^{-\frac{1}{2}})$ improvement step. 

Therefore, as we cannot rule out that such worst-case situation arises in each iteration of our algorithm\footnote{One would suspect, however, that such situations are indeed rare. This might be one explanation of why in practice interior-point methods are able to take most of its step sizes to be very large and thus converge much faster than indicated by the worst-case bound of $O(\hm^{\frac{3}{2}})$.}, getting our desired improvement requires developing a strategy that explicitly ensures that this is not the case (or, at least, not too often). 

At a high level, our general approach to accomplishing this goal is based on ``massaging'' the solution that we maintain. That is, we devise and carefully combine two methods of altering our solutions. These methods, on one hand, significantly improve the behavior of the associated electrical flow while, on other hand, only slightly perturb the characteristics of that solution that are vital to recovering the desired near-perfect $\bb$-matching at the end.

The first of these two methods is related to edge removal technique of Christiano et al. \cite{ChristianoKMST11}. Their technique is based on repeated removal from the graph of the edges that suffer too much congestion. As \cite{ChristianoKMST11} showed (via a simple energy-based argument), when such edge removal is applied to electrical flows that guide multiplicative-weight-update-based optimization routine, one obtains a significantly faster convergence to approximately optimal solution.

Unfortunately, as our primal-dual framework has much more delicate nature than the multiplicative-weight-update method, such removal of ``bottlenecking'' arcs would be too drastic and, in particular, could destroy the structure of our dual solution. Therefore, we apply a more careful approach. 

Instead of removing arcs, we only perturb them by moderately increasing their lengths (and thus their slack variables). (Note that, by \eqref{eq:hf_resistances}, increasing arc's slack variable increases its resistance.) Furthermore, to avoid significant distortion of the dual solution, we do not apply this perturbation to all ``bottlenecking'' arcs, but only the ones that are ``heavy'' in the primal solution (see Definition \ref{def:heavy} below). 

We then use a certain refinement of the original energy-based argument of Christiano et al. \cite{ChristianoKMST11} (that needs, in particular, to deal with the fact that -- in contrast to the multiplicative-weight-update-based framework of \cite{ChristianoKMST11} -- in our framework the arc's resistances can change in a completely non-monotonic fashion) to show that the behavior of our guiding electrical flows on such ``heavy'' arcs is indeed improved. 

Now, our second method -- that is somewhat complementary to the first one and aims at accommodating the ``light'' arcs -- is based on an appropriate preconditioning of our solution by augmenting it with auxiliary arcs. The purpose of adding these arcs is to improve the connectivity (and thus electrical conductance) of the underlying solution (when treated as a graph with resistances) while changing the structure of our original solution in only minimal way (that can be fixed later). We then show via a certain dual-based argument that existence of these auxiliary arcs ensures that ``light'' arcs are never the bottlenecking ones (and thus do need to be dealt with anymore). 

We proceed now to detailed description and analysis of our improved algorithm.

\subsubsection*{The Sets $\Cset{l}{\hff}$ and $\theta$-Smoothness}

We start by specifying the behavior of associated electrical flows that is ``good'' from our perspective. To this end, for a flow $\hff$ in $\hG$, a solution $(\ff,\ss,\vnu)$, and integer $l$, let us define $\Cset{l}{\hff}$ to be the set of all the arcs $e$ such that
\begin{equation}
\label{eq:def_of_C}
\frac{\sqrt{\hm}}{2^{l+1}}< \rho(\hff,\ff)_e \leq \frac{\sqrt{\hm}}{2^{l}},
\end{equation}
i.e., the collection of all the arcs whose congestion in the flow $\hff$ (with respect to capacities given by $\ff$) is between $\frac{\sqrt{\hm}}{2^{l+1}}$ and $\frac{\sqrt{\hm}}{2^{l}}$.

Now, we introduce a definition that will be fundamental to the rest of our discussion.
\begin{definition}
\label{def:smoothness}
For some $0\leq \theta \leq 1$, a flow $\hff$, and solution $(\ff,\ss,\vnu)$ (that will be always clear from the context), we say that $\hff$ is {\em $\theta$-smooth on some of arcs $S\subseteq \hE$} iff, for any integer $l\leq \log \theta^{-3}$, we have that
\[
\vnu(\Cset{l}{\hff}\cap S)\leq \floor{\theta^3 2^{3l}},
\]
where $\vnu(S'):=\sum_{e\in S'} \nu_e$. 
Furthermore, we simply say that $\hff$ is {\em $\theta$-smooth} if $S=\hE$, i.e., $S$ contains all the arcs. 
\end{definition} 
Clearly, the $\theta$-smoothness constraints the distribution of the arcs that suffer high congestion in $\hff$. In particular, it implies that there is no arcs whose congestion $\rho(\hff,\ff)_e$ is larger than $\theta \sqrt{\hm}$. 

Observe that the tight worst-case example for the lowerbound on $\delta^t$ (cf. Fact \ref{fa:basic_delta_lowerbound}) corresponds to situation where the electrical flow $\hff^t$ associated with the maintained $\hgamma$-centered solution solution $(\ff^t,\ss^t,\vnu^t)$ makes some arcs highly-congested, i.e., makes them suffer congestion of $\Omega(\sqrt{\hm})$. However, the above definition of $\theta$-smoothness, forbids existence of such arcs. Therefore, the hope is that once our electrical flows $\hff^t$ are $\theta$-smooth, a better lowerbound on $\delta^t$ (and thus faster convergence) is possible. As the following lemma shows, this hope is indeed well-founded.

\begin{lemma}
\label{lem:better_delta_lowerbound}
Let $(\ff^t,\ss^t,\vnu^t)$ be a $\vsigma$-feasible and $\hgamma$-centered solution and let $\hff^t$ be the associated electrical flow that is $\theta$-smooth, for some $0\leq \theta\leq 1$. We have that
\[
\delta^t \geq \frac{1}{\cdelta\theta \sqrt{\hm}},
\]
for some sufficiently large constant $\cdelta\geq 1$.
\end{lemma}

\begin{proof}
By Theorem \ref{thm:main_interior_point}, in order to lowerbound $\delta^t$ we need to upperbound the quantity
\[
\norm{\vrho(\hff^t,\ff^t)}{\vnu^t,4}^4 = \sum_e \nu_e^t \rho(\hff^t,\ff^t)_e^4.
\] 

To this end, note that 
\begin{eqnarray*}
\sum_e \nu_e^t \rho(\hff^t,\ff^t)_e^4 & \leq & \sum_l \vnu^t(\Cset{l}{\hff^t}) \frac{\hm^2}{2^{4l}} \leq \sum_l  \floor{\theta^3 2^{3l}} \frac{\hm^2}{2^{4l}}\leq \sum_{l\geq \floor{\log \theta^{-1}}}  \frac{\theta^{3}\hm^2}{2^{l}}\leq 4\theta^{4}\hm^2,
\end{eqnarray*}
where we used the $\theta$-smoothness of $\hff^t$ (cf. Definition \ref{def:smoothness}) and the fact that $\floor{\theta^3 2^{3l}}=0$ whenever $l<\floor{\log \theta^{-1}}$. So, the lemma follows once $\cdelta>0$ is chosen to be an appropriately large constant.
\end{proof}

In the light of the above lemma, if we somehow knew that all -- or, at least, most of -- the flows $\hff^t$ that we compute are indeed $\theta$-smooth for some  $\theta=O(\hm^{-\eta})$, we would immediately get the desired faster algorithm. Unfortunately, as we already discussed, it seems to be hard to argue that this is what happens in the worst-case.  Therefore, to address this problem we develop a perturbation approach that we carefully apply to our maintained solutions to ensure that most of the flows $\hff^t$ is indeed $\theta$-smooth for some small enough value of $\theta$.

\subsubsection*{$\alpha$-Stretching}

One of the main operations that we will use to implement our perturbations is called $\alpha$-stretching. To describe it, consider a solution $(\ff,\ss,\vnu)$ that is $\gamma$-centered and a parameter $\alpha\geq 0$. We define an {\em $\alpha$-stretching} of an arc $e$ to be an operation that returns a solution $(\ff',\ss',\vnu')$ obtained from $(\ff,\ss,\vnu)$ by, first, increasing the length $\hl_e$ of the arc $e$ (and thus the value of $s_e$) by $\alpha s_e$ and, then, increasing the measure $\nu_e$ of $e$ by a factor of $(1+\beta)$, where
\begin{equation}
\label{eq:def_of_beta}
\beta:=\frac{\alpha f_e s_e}{\nu_e\hmu(\ff,\ss,\vnu)}.
\end{equation}
The remaining part of the solution remains the same.

The property of  $\alpha$-stretching that is key from our point of view, is that after applying it to some arc $e$ its resistance $r_e:=\frac{s_e}{f_e}$ increases by a factor of exactly $(1+\alpha)$. Furthermore, our choice of value of $\beta$ is justified by the lemma below -- its proof appears in Appendix \ref{app:choosing_beta}.

\begin{lemma}
\label{lem:choosing_beta}
If $(\ff,\ss,\vnu)$ was a $\gamma$-centered solution with $\gamma\leq \frac{1}{2}$ then so will be $(\ff',\ss',\vnu')$ and $\hmu(\ff',\ss',\vnu')= \hmu(\ff,\ss,\vnu)$. Furthermore, $(1-\gamma)\alpha \leq \beta \leq (1+\gamma)\alpha$.
\end{lemma}

So, we see that applying $\alpha$-stretching with this setting of $\beta$ does not perturb our measure of progress $\hmu(\ff,\ss,\vnu)$, even though the duality gap $\ff^T\ss$ changes due to corresponding increase in measure. (Again, this is one reason why we chose $\hmu(\ff,\ss,\vnu)$ to measure our progress.)

On the other hand, besides the increase in measure, another undesirable byproduct of $\alpha$-stretching is the increase of arc's length. To mitigate the effect of this process on the validity of our final solution, we will ensure that the following invariant is maintained throughout the algorithm.

\begin{invariant}
\label{inv:length_upper}
The overall increase of arcs' length due to $\alpha$-stretching is at most $\tO{\hm^{\frac{1}{2}-\eta}}$ and no individual arc has its total increase of length larger than $1$.
\end{invariant} 

\noindent Maintaining this invariant will be important for two reasons. One is captured by the following simple lemma, whose proof appears in Appendix \ref{app:bound_on_s}.

\begin{lemma}
\label{lem:bound_on_s}
If Invariant \ref{inv:length_upper} is preserved then for any $\vsigma$-feasible solution $(\ff,\ss,\vnu)$, we have that $s_{e}$ is at most $6$, for any arc $e$.
\end{lemma}

The other, and even more important one, is that as long as this invariant is preserved the final close-to-optimal solution $(\ff^t,\ss^t,\vnu^t)$ to our perturbed problem still allows us to recover the desired near-perfect $\bb$-matching in the (original) graph $G$ (or conclude that no perfect $\bb$-matching exists). More precisely, in Appendix \ref{app:perturbed_solution} we prove the following lemma.

\begin{lemma}
\label{lem:perturbed_solution}
Provided Invariant \ref{inv:length_upper} holds, given any feasible $\hvsigma$-flow $\ff$ in $\hG$ whose cost is within additive $\frac{1}{2}$ of the optimum, we can recover in $\tO{\hm}$ time a (fractional) near-perfect $\bb$-matching in $G$, or conclude that no perfect $\bb$-matching exists in $G$. 
\end{lemma}

\subsubsection*{Heavy Arcs}

As we already mentioned, an important role in our improved algorithm is played by a classification of arcs into two classes, ``heavy'' and ``light'', depending on their current flow in the primal solution. We make this classification precise below.  

\begin{definition}\label{def:heavy}
Given a $\gamma$-centered solution $(\ff,\ss,\vnu)$ with $\gamma\leq \frac{1}{2}$, we call an arc $e$ {\em heavy} if $f_e\geq \nu_e \fheavy$, where
\[
\fheavy:=\cheavy^{-1} \hm^{\frac{1}{2}-3\eta} \hmu(\ff,\ss,\vnu),  
\]
for some sufficiently large constant $\cheavy>1$ that we will fix later (see Lemma \ref{lem:separated_sets}). We say that an arc is {\em light} if it is not heavy.
\end{definition}

The motivation for the above classification stems from a desire to control the increase in arc's length due to an application of $\alpha$-stretching.  Namely, observe that if we $\alpha$-stretch an heavy arc $e$ then the increase in this arc's length is by at most 
\[
\alpha s_e \leq \frac{\alpha (1+\gamma) \nu_e \hmu(\ff,\ss,\vnu)}{f_e} \leq \cheavy (1+\gamma) \alpha \frac{\hm^{3\eta}}{\sqrt{\hm}}, 
\] 
where we used Fact \ref{fa:central_vs_max_min}. So, as long as we apply $\alpha$-stretching operations only to heavy arcs -- which essentially will be the case in our algorithm -- we can guarantee that the resulting change in arc length is relatively small. This will be important to ensuring that Invariant \ref{inv:length_upper} is never violated.

Having introduced the above concepts, we are ready to proceed to presenting our improved algorithm. In this presentation, we fix for the rest of this section $\htheta:=\hm^{-\eta}$, where 
\begin{equation}\label{eq:def_eta}
\eta:=\frac{1}{14}-\ceta \frac{\log \log n}{\log n}
\end{equation}
 and $\ceta$ is a sufficiently large constant to be fixed later. 

We  describe our algorithm in two stages. First, in Section \ref{sec:heavy}, we present a variant of the algorithm that works under an ad-hoc assumption that all the electrical flows $\hff^t$ that we compute are always $\htheta$-smooth on the set of light arcs. (So, we need to deal there only with its possible non-$\htheta$-smoothness on the set of heavy arcs.) Then, in Section \ref{sec:preconditioning}, we show how to apply a preconditioning technique to obtain an augmented version of our graph such that when we run our algorithm it is indeed true that the above ad-hoc assumption holds. 

%% file: files/improved_algorithm2.tex
\subsection{Perturbing Heavy Arcs}\label{sec:heavy}

In this section, we work under an ad-hoc assumption that the electrical flows $\hff^t$ that are associated with the maintained solutions $(\ff^t,\ss^t,\vnu^t)$ are always $\htheta$-smooth -- with $\htheta:=\hm^{-\eta}$ -- on the set of arcs that are light with respect to that solution. We present an $\tO{\hm^{\frac{3}{2}-\eta}}$-time algorithm for this setting. 

\subsubsection*{$\htheta$-Improvement Phase}

The core of this algorithm is an implementation of a primitive we call a {\em $\htheta$-improvement phase}. This primitive, given a $\hvsigma$-feasible and $\hgamma$-centered solution $(\ff^{t_0},\ss^{t_0},\vnu^{t_0})$, returns in $\tO{\hm^{1+2\eta}}$ time a $\hvsigma$-feasible and $\hgamma$-centered solution $(\ff^{t_f},\ss^{t_f},\vnu^{t_f})$ such that 
\begin{equation}\label{eq:theta_improvement_progress}
\hmu(\ff^{t_f},\ss^{t_f},\vnu^{t_f})\leq \hlambda\hmu(\ff^{t_0},\ss^{t_0},\vnu^{t_0}),
\end{equation}
where $\hlambda:=\left(1-\frac{1}{2\cdelta \htheta \sqrt{\hm}}\right)^{\htheta^{-2}}$ and $\cdelta$ is the constant from Lemma \ref{lem:better_delta_lowerbound}. 

Observe that once we obtain an implementation of such a $\htheta$-improvement phase, we can get the desired improved algorithm as follows. We start with a $\hvsigma$-feasible and $\hgamma$-centered solution $(\ff^0,\ss^0,\vnu^0)$ as in Lemma \ref{lem:initial_solution}. Next, we apply $\hT$ iterations of $\htheta$-improvement phase to it, with
\begin{equation}\label{eq:theta_improvement_overall_hT}
\hT:=2\cdelta \htheta^{3}\sqrt{\hm} \ln 8\hm = O(\hm^{\frac{1}{2}-3\eta}\log \hm).
\end{equation}

Note that after doing this, we know that if  $(\ff^F,\ss^F,\vnu^F)$ is the final $\hvsigma$-feasible $\hgamma$-centered solution we compute then
\[
\hmu(\ff^F,\ss^F,\vnu^F)\leq \hlambda^{\hT} \hmu(\ff^0,\ss^0,\vnu^0) = \left(1-\frac{1}{2\cdelta \htheta \sqrt{\hm}}\right)^{2\cdelta \htheta\sqrt{\hm}\ln 8\hm}\leq \frac{1}{8\hm}.
\] 

So, as long as we can show that $(\ff^F,\ss^F,\vnu^F)$ satisfies Invariants \ref{inv:measure_upperbound} and \ref{inv:length_upper}, we can use Lemma \ref{lem:perturbed_solution} to recover the desired near-perfect $\bb$-matching in $G$ or conclude that no perfect $\bb$-matching exists in $G$. 
Also, the overall running time of this algorithm will indeed be $\tO{\hT \hm^{1+2\eta}}=\tO{\hm^{\frac{3}{2}-\eta}}$, as desired. 

\subsubsection*{Implementation of $\htheta$-Improvement Phase via Stretch-boosts}

In the light of the above discussion, we just need to focus on implementing the $\htheta$-improvement phase, as well as, ensuring that running it for $\hT$ iterations will not violate Invariants \ref{inv:measure_upperbound} and \ref{inv:length_upper}.

\IncMargin{1em}
\begin{algorithm}
\DontPrintSemicolon
\SetAlFnt{\small \sf}
\SetKwComment{tcc}{$(*$ }{ $*)$}

\AlFnt
\SetKw{Return}{return}
\SetKwInOut{Input}{Input}\SetKwInOut{Output}{Output}
\Input{A $\hvsigma$-feasible and $\hgamma$-centered solution $(\ff^{t_0},\ss^{t_0},\vnu^{t_0})$}
\Output{A $\hvsigma$-feasible and $\hgamma$-centered solution $(\ff^{t_f},\ss^{t_f},\vnu^{t_f})$ with $\hmu(\ff^t,\ss^t,\vnu^t)\leq \hlambda \hmu(\ff^{t_0},\ss^{t_0},\vnu^{t_0})$}
\BlankLine
{Initialize $t\leftarrow t_0$}\;
\While{\AlFnt{$\hmu(\ff^t,\ss^t,\vnu^t)> \hlambda \hmu(\ff^{t_0},\ss^{t_0},\vnu^{t_0})$}}{
Compute the electrical $\hvsigma$-flow $\hff^t$ associated with $(\ff^t,\ss^t,\vnu^t)$\;
\eIf{\AlFnt{$\hff^t$ is $\htheta$-smooth on heavy arcs}}{
Apply interior-point method step from Theorem \ref{thm:main_interior_point} to $(\ff^t,\ss^t,\vnu^t)$  with $\delta^t:=\frac{1}{2\cdelta \htheta \sqrt{\hm}}$ \;
Save the resulting solution as $(\ff^{t+1},\ss^{t+1},\vnu^{t+1})$\tcc*[r]{progress step}
}{
Let $l^*\leq \log \htheta^{-3}$ be such that $\vnu^t(\Cset{l^*}{\hff^t}\cap E_{H}^t)> \max\{\htheta^3 2^{3l^*},1\}$\;
\lForEach{\AlFnt{arc $e$ in $\Cset{l^*}{\hff^t}\cap E_{H}^t$}}{
apply $1$-stretching to $e$\tcc*[r]{stretch-boost}
}
Save the resulting solution as $(\ff^{t+1},\ss^{t+1},\vnu^{t+1})$
}
$t\leftarrow t+1$\;
}
Output $(\ff^t,\ss^t,\vnu^t)$ as the solution $(\ff^{t_f},\ss^{t_f},\vnu^{t_f})$\;
\caption{Implementation of $\htheta$-improvement phase via stretch-boosts}\label{fig:improve_phase_nonlight}
\end{algorithm}\DecMargin{1em}

Our implementation -- presented in Figure \ref{fig:improve_phase_nonlight} -- is an iterative procedure. We maintain a $\hvsigma$-feasible $\hgamma$-centered solution $(\ff^t,\ss^t,\vnu^t)$ -- initially, $(\ff^t,\ss^t,\vnu^t)$ is equal to $(\ff^{t_0},\ss^{t_0},\vnu^{t_0})$. Next, as long as $\hmu(\ff^t,\ss^t,\vnu^t)> \hlambda \hmu(\ff^{t_0},\ss^{t_0},\vnu^{t_0})$ we repeat the following iterative step. 

We first check if the electrical flow $\hff^t$ associated with $(\ff^t,\ss^t,\vnu^t)$ is $\htheta$-smooth on the set of heavy arcs. 

If it is indeed the case then one can easily see that such $\hff^t$ needs to be $2\htheta$-smooth (on the set of all the arcs). (This uses our ad-hoc assumption that all $\hff^t$ we compute are always $\htheta$-smooth on the set of light arcs.) So, in this situation, we can just apply an interior-point method step -- as described in Theorem \ref{thm:main_interior_point} -- to $(\ff^t,\ss^t,\vnu^t)$ with setting $\delta^t:=\frac{1}{2\cdelta \htheta \sqrt{\hm}}$. (Note that by Lemma \ref{lem:better_delta_lowerbound} this setting of $\delta^t$ is valid.)  For future reference, we call this step a {\em progress step}. After executing it, we set the resulting $\hvsigma$-feasible and $\hgamma$-centered solution $(\ff^{t+1},\ss^{t+1},\vnu^{t+1})$ as our current solution and proceed to next iterative step. 

Otherwise, that is, if $\hff^t$ is not $\htheta$-smooth on the set of heavy arcs, then -- by Definition \ref{def:smoothness} -- there is an $l^*\leq \log \htheta^{-3}$ such that 
\begin{equation}
\label{eq:boosting_condition}
\vnu^t(\Cset{l^*}{\hff^t}\cap E_{H}^t)> \max\{\htheta^3 2^{3l^*},1\},
\end{equation}
where we used the fact that all measures of arcs are always at least one and $E_{H}^t$ denotes the set of heavy arcs with respect to the solution $(\ff^t,\ss^t,\vnu^t)$. 

To cope with this situation, we perform $1$-stretching of all the arcs in $\Cset{l^*}{\hff^t}\cap E_{H}^tt$. (Note that, by Lemma \ref{lem:choosing_beta}, this operation does not change the value of $\hmu(\ff^t,\ss^t,\vnu^t)$ and our solution remains $\hvsigma$-feasible and $\hgamma$-centered.) We call this operation {\em stretch-boosting} and $l^*$ will be referred to as the {\em index} of this stretch-boosting. After performing stretch-boosting, we proceed to the next iterative step. 

This finishes the description of our implementation. 

\subsubsection*{Analysis}

To analyze the above procedure, let us note that due to our stopping condition, once this procedure terminates the resulting solution $(\ff^{t_f},\ss^{t_f},\vnu^{t_f})$ satisfies our requirements. Also, there will be at most $\htheta^{-2}$ progress steps executed. This is so, as $1$-stretching does not affect the value of $\hmu(\ff^t,\ss^t,\vnu^t)$ and, by Theorem \ref{thm:main_interior_point}, each progress step decreases $\hmu(\ff^t,\ss^t,\vnu^t)$ by a factor of at least $(1-\delta^t)=\hlambda^{\frac{1}{\htheta^{-2}}}$. Thus, as each of these steps runs in $\tO{\hm}$ time, the resulting total time of progress steps is $\tO{\hm \htheta^{-2}}=\tO{\hm^{1+2\eta}}$, as desired.

Therefore, we can just focus on bounding the number of stretch-boost operations executed, as well as, on showing that calling our implementation of $\htheta$-improvement phase $\hT$ times does not violate Invariants \ref{inv:measure_upperbound} and \ref{inv:length_upper}. 

We start with the former task. In this bounding of the number of stretch-boost operations, we assume that Invariant \ref{inv:measure_upperbound} holds. We will justify this assumption later when proving that our two desired invariants are indeed preserved by our algorithm. 

To do the bounding, we consider the energy $\energy{\rr^t}{\hff^t}$ of the electrical $\hvsigma$-flow $\hff^t$ (determined by resistances $\rr^t$ given by \eqref{eq:hf_resistances}) that is associated with our current solution $(\ff^t,\ss^t,\vnu^t)$. We treat this quantity as a potential function and show the following facts:

\begin{enumerate}[(a)]\addtolength{\itemsep}{-.4\baselineskip}
\item $\energy{\rr^t}{\hff^t}$ is always at least $\cenergy^{-1}\hm \hmu(\ff^t,\ss^t,\vnu^t)$ and at most $\cenergy\hm \hmu(\ff^t,\ss^t,\vnu^t)$, for some sufficiently large constant $\cenergy>1$ -- see Lemma \ref{lem:energy_bound};\label{cond:stretch_tot_energy}
\item $\energy{\rr^t}{\hff^t}$ increases by a factor of at least $(1+\cincrease\htheta^2)$, for some constant $\cincrease>0$, whenever a stretch-boosting step is applied -- see Lemma \ref{lem:stretch_boosting_energy_increase};\label{cond:stretch_boosting_energy}
\item $\energy{\rr^t}{\hff^t}$ decreases by a factor of at most $(1+\cendecrease\htheta^2\ln \hm)$ each time a progress step is executed, where $\cendecrease>0$ is some sufficiently large constant -- see Lemma \ref{lem:int_step_energy_decrease}.\label{cond:stretch_progress_energy}
\end{enumerate}

Note that once the above statements are established, it must be the case that there is at most $T_s:=\cendecrease \cincrease^{-1} \htheta^{-2} \ln \cenergy^2 \hm=\tO{\htheta^{-2}}$ stretch-boost operation overall. To see that, assume this was not the case, i.e., that there was more than $T_s$ stretch-boosts. Then, by the above statements and the fact that there is at most $\htheta^{-2}$ progress steps we would have that
\[
\energy{\rr^{t_f}}{\hff^{t_f}} > \frac{(1+\cincrease\htheta^2)^{T_s} \energy{\rr^{t_0}}{\hff^{t_0}}}{(1+\cendecrease\htheta^2\ln \hm)^{\htheta^{-2}}}\geq \frac{(1+\cincrease\htheta^2)^{T_s} \hm \hmu(\ff^{t_0},\ss^{t_0},\vnu^{t_0})}{\cenergy\hm^{\cendecrease}}\geq \cenergy \hm \hmu(\ff^{t_f},\ss^{t_f},\vnu^{t_f}),
\]
which would violate the upperbound on energy $\energy{\rr^{t_f}}{\hff^{t_f}}$ established by statement \eqref{cond:stretch_tot_energy}.

 So, it must be then indeed the case that there is at most $T_s=\tO{\htheta^{-2}}$ stretch-boosts, which gives the desired $\tO{\htheta^{-2}\hm}=\tO{\hm^{1+2\eta}}$ total running time bound. 

In the light of the above, we can turn our attention to proving statements \eqref{cond:stretch_tot_energy}-\eqref{cond:stretch_progress_energy}. We start with statement \eqref{cond:stretch_tot_energy}. This statement essentially follow from Lemma \ref{lem:hfft_energy_bound} and some simple energy-lowerbounding argument. The prove of the following lemma appears in Appendix \ref{app:energy_bound}.

\begin{lemma}
\label{lem:energy_bound}
Let $(\ff^t,\ss^t,\vnu^t)$ be a $\hvsigma$-feasible and $\hgamma$-centered solution. Provided that Invariant \ref{inv:measure_upperbound} holds, we have that
\[
\cenergy^{-1}\hm \hmu(\ff^t,\ss^t,\vnu^t)\leq \energy{\rr^t}{\hff^t}\leq \cenergy \hm \hmu(\ff^t,\ss^t,\vnu^t),
\]
where $\hff^t$ is the electrical $\hvsigma$-flow associated with the solution $(\ff^t,\ss^t,\vnu^t)$ and $\cenergy>1$ is a sufficiently large constant. 
\end{lemma}

Next, we proceed to analyzing the effect of stretch-boosting on the energy $\energy{\rr^t}{\hff^t}$. Intuitively, by the connection between congestion and energy hinted  by Fact \ref{fa:rho_vs_rt}, we know that the arcs with large congestion have to have unusually high contribution to the energy $\energy{\rr^t}{\hff^t}$. So, as $1$-stretching effectively doubles the resistances of such arcs, it is not surprising that it ends up significantly increasing that energy. We make this formal -- and thus establish statement \eqref{cond:stretch_boosting_energy} -- in the lemma below. Its proof appears in Appendix \ref{app:stretch_boosting_energy_increase}.

\begin{lemma}
\label{lem:stretch_boosting_energy_increase}
Each stretch-boost increases $\energy{\rr^t}{\hff^t}$ by a factor of at least 
\[
\left(1+\cincrease\htheta^{2}\left(\vnu^t(\Cset{l^*}{\hff^t}\cap E_{H}^t)\right)^{\frac{1}{3}}\right)\geq \left(1+\cincrease\htheta^{2}\right),
\]
for some constant $\cincrease>0$.
\end{lemma}

To complete our analysis, it remains to show that our potential $\energy{\rr^t}{\hff^t}$ does not decrease too much during the progress steps. In other words, we prove statement \eqref{cond:stretch_progress_energy}. 

Note that the difficulty here stems from the fact that, in principle, the resistances of arcs can change pretty arbitrarily during a progress step. They can either increase or decrease and even by a constant multiplicative factor, thus possibly leading to severe and very non-monotone energy fluctuations. 

The key reason that enables us to control that energy change after all, is that we perform the progress step only if the flow $\hff^t$ is $\htheta$-smooth. This is helpful in two ways. Firstly, because we can use it together with the connection between the change of the resistance of an arc and its congestion that we established in Theorem \ref{thm:main_interior_point}, to show that there is not too many arcs that significantly change their resistance (see Lemma \ref{lem:bounding_smooth} below). Secondly (and even more importantly), our connection between congestion and energy, allows us to conclude that $\htheta$-smoothness implies that there is no small (measure-wise) set of arcs that contributes unusually high portion of the energy. So, even if some small set of arcs changes its resistances significantly, it is not able to influence the overall energy by too much (see Lemma \ref{lem:int_step_energy_decrease}). (In a sense, this intuition is one of the main motivations for introducing the notion of $\htheta$-smoothness.) We, again, formalize this intuition below. 

First, for a given vector $\vlambda$ and an integer $l$, let us define $T_l^{\vlambda}$ to be the set of all arcs such that
\begin{equation}\label{eq:def_t_lambda}
\frac{1}{2^{l+1}}\leq |\lambda_e|\leq \frac{1}{2^l}.
\end{equation}
Now, we say that $\vlambda$ is {\em $\tau$-restricted}, for some measure $\vnu$ and $\tau\geq 0$ if for any $l\geq 0$, 
\begin{equation}\label{eq:def_tau_restricted}
\vnu(T_l^{\vlambda}) \leq \tau 2^{3l}.
\end{equation}
Now, the lemma below bounds the change of resistances during any of our progress steps. 

\begin{lemma}
\label{lem:bounding_smooth}
Let $\hff^t$ be a $2\htheta$-smooth electrical flow associated with $(\ff^t,\ss^t,\vnu^t)$. Let $(\ff^{t+1},\ss^{t+1},\vnu^{t+1})$ be the solution obtained by applying an interior-point method step -- as in Theorem \ref{thm:main_interior_point} -- to $(\ff^t,\ss^t,\vnu^t)$ with $\delta^t:=(2\cdelta\htheta\sqrt{\hm})^{-1}$. Then the vectors $\vkappa^t$ and $\okappa^t$ are all $\crestrict$-restricted (with respect to $\vnu^t$) for some constant $\crestrict>0$.
\end{lemma}

\begin{proof}
We will prove that both the vector $\delta^t\vrho(\hff^t,\ff^t)$ and the vector $\hvkappa^t$ are $O(1)$-restricted. It is easy to see that then the bound from Theorem \ref{thm:main_interior_point} will imply that $\ovkappa^t$ and $\vkappa^t$ are $O(1)$-restricted too. So, choosing large enough constant $\crestrict$ will prove the lemma.

To this end, observe $\hvkappa^t$ is $O(1)$-restricted as $\norm{\hvkappa^t}{\vnu^t,2}$ by Theorem \ref{thm:main_interior_point}. On the other hand, note that if
\[
\delta^t\rho(\hff^t,\ff^t)_e \geq \frac{1}{2^{l}},
\]
 for some $l\geq 1$ and arc $e$, then $e\in \Cset{l'}{\hff^t}$ for some 
 \[
 l'\leq l+\log \delta^t\sqrt{\hm}+2 \leq l+ \log \htheta^{-1} + O(1).
\]
But by $2\htheta$-smoothness of $\hff^t$, this means that the total measure of such arcs is at most
\[
\floor{\htheta^{-3}2^{3l'}} \leq O(2^{3l}),
\]
which establishes that $\delta^t \vrho(\hff^t,\ff^t)$ is indeed $O(1)$-restricted. The lemma follows.
\end{proof}

Using the above observation, we can now finish establishing property \eqref{cond:stretch_progress_energy} by proving the following lemma whose appears in Appendix \ref{app:int_step_energy_decrease}.
\begin{lemma}
\label{lem:int_step_energy_decrease}
Let $\hff^t$ be a $2\htheta$-smooth electrical flow associated with the $\hgamma$-centered solution $(\ff^t,\ss^t,\vnu^t)$. Let $(\ff^{t+1},\ss^{t+1},\vnu^{t+1})$ be the solution obtained by applying an interior-point methods step -- as in Theorem \ref{thm:main_interior_point} -- to $(\ff^t,\ss^t,\vnu^t)$ with $\delta^t:=(2\cdelta \htheta \sqrt{\hm})^{-1}$. Then, 
\[
\energy{\rr^{t+1}}{\hff^{t+1}}\geq (1+\cendecrease\theta^2 \ln \hm)^{-1}\energy{\rr^t}{\hff^t},
\]
where $\hff^{t+1}$ is the electrical flow associated with $(\ff^{t+1},\ss^{t+1},\vnu^{t+1})$ and $\cendecrease>1$ is a sufficiently large constant. 
\end{lemma}

\subsubsection*{Preservation of Invariants \ref{inv:measure_upperbound} and \ref{inv:length_upper}}

Now, as we completed the analysis of the running time of our $\htheta$-improvement phase implementation, we establish the remaining claim, i.e., we prove that executing the above procedure $\hT$ times does not lead to violation of Invariants \ref{inv:measure_upperbound} and \ref{inv:length_upper}. 

\paragraph{Bounding measure increase.}

To this end, let us first focus on bounding the measure increases. By Lemma \ref{lem:choosing_beta}, we know that whenever we $1$-stretch an arc $e$, its measure increases by at most $(1+\hgamma)\nu_e^t$. So, to bound the total measure increase it suffices to bound the total measure of arcs that are affected by $1$-stretches across all the stretch-boost operations. (Here, if the same arc becomes $1$-stretched multiple times, in different stretch-boosts, we account for its measure multiple times.) 

In order to do that, note that by Lemma \ref{lem:stretch_boosting_energy_increase}, if $\nu_i$ is the measure of the set of arcs that are $1$-stretched in $i$-th stretch-boost, we have that the total increase of energy resulting from that is at least
\[
\prod_{i=1}^{k} (1+\cincrease\htheta^{2}\nu_i^{\frac{1}{3}}).
\]
Also, by Lemma \ref{lem:energy_bound}, we know that we have to have that in any single stretch-boost, the energy cannot increase by more than $\cenergy^2$ factor. So, we have that
\[
(1+\cincrease\htheta^{2}\nu_i^{\frac{1}{3}})\leq \cenergy^2
\]
and thus $\nu_i\leq \nu_{\max}:=\cenergy^6 \htheta^6=O(\hm^{6\eta})$, for each $i$. 

As a result, we can lowerbound the total increase of energy due to stretch-boosts by
\[
\prod_{i=1}^{k} (1+\cincrease\htheta^{2}\nu_i^{\frac{1}{3}})\geq (1+\cincrease\htheta^{2}\nu_{\max}^{\frac{1}{3}})^{\frac{\nu}{\nu_{\max}}},
\]
where $\nu:=\sum_i \nu_i$. 

Finally, by Lemma \ref{lem:energy_bound} and Lemma \ref{lem:int_step_energy_decrease}, as well as, the fact that we have at most $\htheta^{-2}$ progress steps, we know that the overall (multiplicative) increase of energy resulting from all the stretch-boosts can be at most
\[
(1+\cincrease\htheta^{2}\nu_{\max}^{\frac{1}{3}})^{\frac{\nu}{\nu_{\max}}}\leq \prod_{i=1}^{k} (1+\cincrease\htheta^{2}\nu_i^{\frac{1}{3}}) \leq \cenergy^2 (1+\cendecrease \htheta^2 \log \hm)^{\htheta^{-2}}.
\]
Therefore, as $\htheta^{2}\nu_{\max}^{\frac{1}{3}}$ is $\Omega(1)$, we have that the total measure increase $\nu$ is at most
\begin{equation}
\label{eq:measure_increase_phase}
\nu \leq \tO{\nu_{\max}}=\tO{\hm^{6\eta}}.
\end{equation} 
As a result, after executing at most $\hT$ $\htheta$-improvement phases, the overall increase of measure can be bounded by
\begin{equation}\label{eq:total_measure}
\hT \cdot \tO{\hm^{6\eta}} =\tO{\hm^{\frac{1}{2}+3\eta}} < \hm.
\end{equation}
Now, given that by Lemma \ref{lem:initial_solution}, we start with our measure being at most $3\hm$ and thus have a slack of at least $\hm$ measure left before Invariant \ref{inv:measure_upperbound} becomes violated, this overall increase will indeed not lead to violation of this invariant.

\paragraph{Bounding arc length increase.} To show that Invariant \ref{inv:length_upper} is preserved as well, let us first note that the only way for length of arcs to increase is due to $1$-stretching occurring during stretch-boosts. Furthermore, we only $1$-stretch an arc if it is heavy. So, if a given (heavy) arc $e$ gets $1$-stretched at some step $t$ then its length increases by at most
\[
s_e^t \leq \frac{(1+\hgamma) \nu_e^t\hmu(\ff^t,\ss^t,\vnu^t)}{f_e^t} \leq \frac{(1+\hgamma)\cheavy \hm^{3\eta}}{\sqrt{\hm}}.
\] 
On the other hand, by Lemma \ref{lem:choosing_beta}, the increase of measure of such arc is at least $(1-\hgamma)\nu_e^t$. So, as $\nu_e^t\geq 1$, the increase of measure of an arc is within a factor of $O(\frac{\hm^{3\eta}}{\sqrt{\hm}})=O(\hm^{-4\eta})$ of increase of the length. So, as we just proved that the total measure increase is at most $\tO{\hm^{\frac{1}{2}+3\eta}}$ (cf. \eqref{eq:total_measure}), the desired bound of $\tO{\hm^{\frac{1}{2}-\eta}}$ on the total length increase follows. 

Finally, as each $1$-stretch increases the measure by a factor of at least $(2-\hgamma)\geq \frac{3}{2}$ and -- as we discussed above -- we never $1$-stretch anymore an arc whose measure is bigger than $\nu_{\max}=O(\hm^{6\eta})$, no single arc will get $1$-stretched more than $O(\log \nu_{\max})=O(\log \hm)$ times. As a result, no single arc has its length increased by more that $O(\hm^{-4\eta}\log \hm)$ that is much smaller than $1$. Therefore, the Invariant \ref{inv:length_upper} is also preserved. This concludes our analysis.

%% file: files/improved_algorithm3.tex
\subsection{Preconditioning the Graph \texorpdfstring{$\hG$}{G}}\label{sec:preconditioning}

Our analysis from the previous section was crucially relying on the assumption that all the flows $\hff^t$ are always $\htheta$-smooth on the set of light arcs. Unfortunately, this assumption is not always valid. 

To cope with this problem, we develop a modification of our algorithm that ensures that this $\htheta$-smoothness assumption holds after all. Roughly speaking, we achieve that by an appropriate preconditioning our solution at the beginning of each $\htheta$-improvement phase. This preconditioning is based on augmenting the graph $\hG$ with additional, auxiliary arcs and correspondingly extending our solution on them. These arcs are very light (i.e., have small value $f_e^t$ of flow flowing through them in augmented solution), while providing good connectivity (and thus relatively low effective resistance) between different vertices of the augmented graph. 

The underlying intuition here is that the over-congestion of a light arc $e$ is caused by amounts of flow that are at most $\sqrt{\hm}\fheavy$ (cf. Definition \eqref{def:heavy}) and thus are relatively small compared to the whole duality gap. So, by deploying these very light auxiliary arcs we encourage the electrical flow $\hff^t$ to reroute such over-congesting flow from $e$ and send it along auxiliary arcs. On the other hand. as the small value of this rerouted flow is small, the perturbation of our desired (non-augmented) solution introduced by these rerouting is relatively minor. Thus, we are able to deal with it relatively easily at the end of the whole $\htheta$-improvement phase, while still ending up making overall progress on the quality of our solution. 

\subsubsection*{Augmenting the Graph and the Solution}

The exact implementation of our preconditioning is based on modifying the execution of $\htheta$-improvement phase that was presented in the previous section in the following way. Let $(\ff^{t_0},\ss^{t_0},\vnu^{t_0})$ be the  $\hgamma$-centered and $\hvsigma$-feasible solution at the beginning of some $\htheta$-improvement phase. 

We start with augmenting the graph $\hG$ by adding to it a new vertex $\ov$, as well as, $a_v$ copies of an arc $(v,\ov)$ and $a_v$ copies of an arc $(\ov,v)$, for each vertex $v$ of $\hG$ other than $v^*$, where
\begin{equation}
\label{eq:def_of_a}
a_v:=\sum_{e\in E(v)} \nu_{e}^{t_0}
\end{equation}
is the total measure (with respect to $\vnu^{t_0}$) of all the arcs adjacent to $v$ in $\hG$. We will call these newly added arcs {\em auxiliary} and denote the augmented graph as $\oG$. 

Next, we extend the solution $(\ff^{t_0},\ss^{t_0},\vnu^{t_0})$ to that augmented graph $\oG$ by assigning $f_e^{t_0}:=\fauxiliary$, $s_e^{t_0}:=\frac{\hmu(\ff^{t_0},\ss^{t_0},\vnu^{t_0})}{f_e}$ and $\vnu^{t_0}:=1$ to each auxiliary arc $e$, where
\begin{equation}
\label{eq:def_of_fauxliary}
\fauxiliary:=\frac{\cdelta \hmu(\ff^{t_0},\ss^{t_0},\vnu^{t_0})\hm^{\frac{1}{2}-3\eta}}{\cauxiliary \hT}
\end{equation}
with $\cauxiliary$ being some sufficiently large constant to be fixed later, and the lengths of the auxiliary arcs being chosen so that the extended solution is still dual feasible. (As we will soon see, the actual lengths of auxiliary arcs are irrelevant.) Note that after this extension, the solution $(\ff^{t_0},\ss^{t_0},\vnu^{t_0})$ remains $\hgamma$-centered, $\hvsigma$-feasible and the value of $\hmu(\ff^{t_0},\ss^{t_0},\vnu^{t_0})$ is unchanged. Also, observe that by Invariant \ref{inv:measure_upperbound}, the number $\om$ of arcs of the augmented graph $\oG$ is still only $O(\hm)$. So, relating various quantities -- in particular, the running times of our procedures -- to either $\hm$ or $\om$ results in only a constant-factor discrepancy (that we will ignore in what follows).

Now, after the above preprocessing, we run the $\htheta$-improvement phase implementation, as described in the previous section, on the extended solution in the augmented graph $\oG$. (In Section \ref{sec:smoothness}, we will prove that the assumption that underlies the analysis from the previous section, i.e., that all the flows $\hff^t$ are $\htheta$-smooth on light arcs, is indeed valid.) The only further modification here is that after each progress step we $\hm^{2\eta}|\okappa_e^t|$-stretch each auxiliary arc $e$ with $|\okappa_e^t|\geq \htheta^{2}$ (cf. Theorem \ref{thm:main_interior_point}). We will call this stretch operation {\em freezing}. (Note that as $\alpha$-stretching only increases the resistances of arcs, this modification is compatible with the energy-based potential argument we employed in the previous section.) This freezing ensures that the flows on auxiliary arcs do not change to significantly in our solution and thus the impact of preconditioning provided by auxiliary arcs on the quality of the final solution is minimized. We make this more precise in the following lemma whose proof appears in Appendix \ref{app:auxiliary_flow_growth}.

\begin{lemma}
\label{lem:auxiliary_flow_growth}
During the whole $\htheta$-improvement phase, we have that for each auxiliary arc $e$, $\cfreeze^{-1} \fauxiliary\leq f_e^t \leq \cfreeze \fauxiliary$, for some constant $\cfreeze>0$. Also, the total increase of measure of auxiliary arcs in that phase is at most $\tO{\hm^{8\eta}}$.
\end{lemma}

Finally, once the execution of the above $\htheta$-improvement phase finishes, we end up with a $\hgamma$-centered and $\hvsigma$-feasible solution $(\ff^{t},\ss^{t},\vnu^{t})$ such that $\hmu(\ff^{t},\ss^{t},\vnu^{t})\leq \hlambda \hmu(\ff^{t_0},\ss^{t_0},\vnu^{t_0})$, as desired. However, this solution corresponds to the augmented graph $\oG$ instead of to the original graph $\hG$. 

To deal with this deficiency, we first simply discard all the auxiliary arcs and correspondingly truncate the solution $(\ff^{t},\ss^{t},\vnu^{t})$ to non-auxiliary arcs. Unfortunately, doing that might, in particular, render that solution not $\hvsigma$-feasible. So, to alleviate this problem, in Section \ref{sec:fixing} below, we describe a fixing procedure that, given such a truncated solution, produces the intended solution $(\ff^{t_f},\ss^{t_f},\vnu^{t_f})$ that corresponds to the original graph $\hG$, is $\hgamma$-centered, $\hvsigma$-feasible and
\[
\hmu(\ff^{t_f},\ss^{t_f},\vnu^{t_f})\leq \hlambda (1+O(\hm^{-\frac{1}{2}}))\hmu(\ff^{t_0},\ss^{t_0},\vnu^{t_0}).
\]
(Note that in our algorithm we are executing only $\hT=\tO{\hm^{\frac{1}{2}-3\eta}}$ $\htheta$-improvement phases overall. So, this additional $(1+O(\hm^{-\frac{1}{2}}))$ factor above is inconsequential.) 

As we will see, a byproduct of this fixing procedure is an increase in the measure of (non-auxiliary) arcs. However, we will show that this increase is bounded by $O(\frac{\cfreeze\hm}{\cauxiliary \hT})$. Thus, taking $\cauxiliary$ to be sufficiently large ensures that the resulting measure increases do not lead to violation of Invariant \ref{inv:measure_upperbound}. (Note that the auxiliary arcs are always discarded at the end, so from the point of view of Invariant \ref{inv:measure_upperbound}, it suffices that by Lemma \ref{lem:auxiliary_flow_growth} the measure of these arcs is always $o(\hm)$.)

In the light of the above discussion, all that remains is to describe and analyze the fixing procedure and to show that one can indeed assume that all the electrical flows $\hff^t$ computed during such $\htheta$-improvement phase are $\htheta$-smooth on the set of light arcs.

\subsubsection*{Fixing Procedure}\label{sec:fixing}
We start by describing and analyzing the fixing procedure that we employ at the end of each $\htheta$-improvement phase. Recall that in this procedure we are given as input a $\hgamma$-centered and $\hvsigma$-feasible solution $(\ff^{t},\ss^{t},\vnu^{t})$ in the augmented graph $\oG$ such that
$\hmu(\ff^{t},\ss^{t},\vnu^{t})\leq \hlambda \hmu(\ff^{t_0},\ss^{t_0},\vnu^{t_0})$. Our goal is to obtain 
a $\hgamma$-centered $\hvsigma$-feasible solution $(\ff^{t_f},\ss^{t_f},\vnu^{t_f})$ in the original graph $\hG$ that satisfies $\hmu(\ff^{t_f},\ss^{t_f},\vnu^{t_f})\leq \hlambda (1+O(\hm^{-\frac{1}{2}}))\hmu(\ff^{t_0},\ss^{t_0},\vnu^{t_0})$.

We do this in two steps. First, we simply truncate the solution $(\ff^{t},\ss^{t},\vnu^{t})$ to the original graph $\hG$ by discarding all the auxiliary arcs and flow on them. Let us denote the resulting solution as $(\ff',\ss',\vnu')$. It is not hard to see that this solution is still $\hgamma$-centered. In the following lemma -- whose proof appears in Appendix \ref{app:fixing_duality_increase} -- we argue that also the value of $\hmu(\ff',\ss',\vnu')$ has not increased by much. 

\begin{lemma}
\label{lem:fixing_duality_increase}
$\hmu(\ff',\ss',\vnu')\leq \hlambda (1+O(\hm^{-\frac{1}{2}}))\hmu(\ff^{t_0},\ss^{t_0},\vnu^{t_0})$.
\end{lemma}

At this point, we know that the solution $(\ff',\ss',\vnu')$ is $\hgamma$-centered and $\hmu(\ff',\ss',\vnu')$ is as small as needed. Unfortunately, this solution can still be not $\hvsigma$-feasible. 

Therefore, in the second step of our procedure, we address this last shortcoming. Our approach here requires introducing a certain simple operation. For a given some solution $(\ff,\ss,\vnu)$, as well as, some $\alpha\geq 0$ and an arc $e$, let us define {\em $\alpha$-widening} of $e$ (in $(\ff,\ss,\vnu)$) as an operation in which we increase the value of $f_e$ by a factor of $(1+\alpha)$ and increase $\nu_e$ by a factor of $(1+\beta)$, where $\beta$ is given via \eqref{eq:def_of_beta}.

We can view the $\alpha$-widening operation as a counterpart of the $\alpha$-stretching operation. In fact, one can see that due to symmetric nature of $f_e$ and $s_e$ and our choice of $\beta$, Lemma \ref{lem:choosing_beta} also holds for $\alpha$-widening operation. (Note that in the proof of Lemma \ref{lem:choosing_beta} the roles of $f_e$ and $s_e$ are completely interchangeable.)

Now, our way of obtaining the desired solution $(\ff^{t_f},\ss^{t_f},\vnu^{t_f})$ is very simple. Let us denote by $\ovsigma$ the actual demand vector of $\ff'$ and let $\tvsigma:=\hvsigma-\ovsigma$ be the vector of demand differences. We start with $(\ff',\ss',\vnu')$ and for each vertex $v$ of $\hG$ other than $v^*$, we do the following. If $\tsigma_v\geq 0$ (resp. $\tsigma_v<0$), we apply $\alpha_v$-widening to the arc $e(v):=(v,v^*)$ (resp. $e(v):=(v^*,v)$) with $\alpha_v:=\frac{|\tsigma_v|}{f_e'}$. 

We take $(\ff^{t_f},\ss^{t_f},\vnu^{t_f})$ to be the resulting solution. It is easy to see that this solution is $\hvsigma$-feasible now. Also, by Lemma \ref{lem:choosing_beta}, we know that this solution remains $\hgamma$-centered and that $\hmu(\ff^{t_f},\ss^{t_f},\vnu^{t_f})=\hmu(\ff',\ss',\vnu')\leq\hlambda (1+O(\hm^{-\frac{1}{2}}))\hmu(\ff^{t_0},\ss^{t_0},\vnu^{t_0})$, as needed.

So, we just need to establish the claimed bound of $O(\frac{\cfreeze\hm}{\cauxiliary \hT})$ on total measure increase resulting from this procedure. To this end, note that by Lemma \ref{lem:choosing_beta} this increase is at most
\begin{eqnarray*}
(1+\hgamma) \sum_{v\neq v^*} \alpha_v \nu_{e(v)}^t & = &  (1+\hgamma) \sum_{v\neq v^*} |\tsigma_v| \frac{\nu_{e(v)}^t}{f_{e(v)}'} \leq (1+\hgamma)^2 \sum_{v\neq v^*} |\tsigma_v| \frac{s_{e(v)}'}{\hmu(\ff',\ss',\vnu')} \\
&=& O\left(\sum_{v\neq v^*}  \frac{|\tsigma_v|}{\hmu(\ff',\ss',\vnu')}\right)=O\left(\frac{\onorm{\hvsigma-\ovsigma}}{\hmu(\ff',\ss',\vnu')}\right),
\end{eqnarray*}
where we used the Fact \ref{fa:central_vs_max_min} and Invariant \ref{inv:length_upper}, as well as, we applied Lemma \ref{lem:bound_on_s} to conclude that each $s_{e(v)}'$ is $O(1)$. 

Thus, in the light of the above, it only remains to bound $\onorm{\hvsigma-\ovsigma}$. 

\begin{lemma}
\label{lem:fixing_value_flow}
$\onorm{\hvsigma-\ovsigma}=O(\frac{\cfreeze\hm \hmu(\ff',\ss',\vnu')}{\cauxiliary \hT})$
\end{lemma}

\begin{proof}
One can see that we can bound $\onorm{\hvsigma-\ovsigma}$ by bounding the total (additive) change of the flow $\ff^t$ on all auxiliary arcs during the whole execution of $\htheta$-improvement procedure. Furthermore, as the flow $\ff^t$ changes only during progress steps, and there is at most $\htheta^{-2}=\hm^{2\eta}$ of them, it suffices to prove that in each progress step this change is at most $O(\cfreeze\hmu(\ff',\ss',\vnu')\frac{\hm^{1-2\eta} }{\cauxiliary \hT}))$. 

Now, by Theorem \ref{thm:main_interior_point} and Lemma \ref{lem:auxiliary_flow_growth}, this (additive) change at step $t$ can be bounded as
\begin{equation}\label{eq:fixing_inter_flow_bound}
\sum_{e\in S} |\okappa_e^t| f_e^t \leq \cfreeze \fauxiliary \sum_{e\in S} |\okappa_e^t| \leq 4 \cfreeze \fauxiliary \left(\sum_{e\in S} \delta^t \rho(\hff^t,\ff^t)_e + \sum_{e\in S} \hkappa_e^t\right),
\end{equation}
where $S$ is the set of auxiliary arcs. By Cauchy-Schwarz inequality, we get that 
\[
\sum_{e\in S}  \rho(\hff^t,\ff^t)_e \leq \sum_{e\in S} \nu_e^t \rho(\hff^t,\ff^t)_e \leq \sqrt{\norm{\vrho(\hff^t,\ff^t)}{\vnu^t,2}^2 \vnu^t(S)}\leq O(\hm),
\] 
where we used \eqref{eq:worst_case_rho} and the fact that by Lemma \ref{lem:auxiliary_flow_growth} $\vnu^t(S)$ is $O(\hm)$. Similarly, we obtain that
\[
\sum_{e\in S} \hkappa_e^t \leq \sum_{e\in S} \nu_e^t \hkappa_e^t \leq \sqrt{\norm{\hkappa^t}{\vnu^t,2}^2 \vnu^t(S)}\leq O(\sqrt{\hm}),
\]
where we used that fact that $\norm{\hkappa^t}{\vnu^t,2}\leq \frac{1}{16}$. 

Plugging the above to bounds back into \eqref{eq:fixing_inter_flow_bound} and recalling that we always set $\delta^t:=(2\cdelta\htheta \sqrt{\hm})^{-1}$, we obtain that
\begin{eqnarray*}
\onorm{\hvsigma-\ovsigma}&\leq& 4 \cfreeze \fauxiliary \left(\sum_{e\in S} \delta^t \rho(\hff^t,\ff^t)_e + \sum_{e\in S} \hkappa_e^t\right)\leq  O\left(\frac{\cfreeze \fauxiliary\sqrt{\hm}}{\cdelta \htheta}\right)\\
&\leq & O\left(\frac{\cfreeze \hmu(\ff^{t_0},\ss^{t_0},\vnu^{t_0})\hm^{1-2\eta} }{\cauxiliary \hT}\right) \leq O\left(\frac{\cfreeze \hmu(\ff',\ss',\vnu')\hm^{1-2\eta} }{\cauxiliary \hT}\right),
\end{eqnarray*}

where we utilized \eqref{eq:def_of_fauxliary}, as well as, the fact that, due to our stopping condition for $\htheta$-improvement phase, we can always assume that $\hmu(\ff^{t_0},\ss^{t_0},\vnu^{t_0})$ is $O(\hmu(\ff',\ss',\vnu'))$. The lemma follows.
\end{proof}

Clearly, by setting $\cauxiliary$ to be a sufficiently large constant, we can ensure that the total measure increase due to fixing procedure will not lead to violation of Invariant \ref{inv:measure_upperbound}.

\subsubsection*{$\htheta$-smoothness on Light Arcs}\label{sec:smoothness}

As the final step of our analysis, we prove now that in the course of our algorithm -- after the modifications described above -- all the electrical flows $\hff^t$ that we compute are indeed $\htheta$-smooth on the set of light arcs. That is, the assumption underlying the analysis performed in Section \ref{sec:heavy} is indeed justified.

To this end, let us fix some $\hvsigma$-feasible and $\hgamma$-centered solution $(\ff^t,\ss^t,\vnu^t)$ in our augmented graph $\oG$ and let $\hff^t$ be the associated electrical $\hvsigma$-flow. For convenience, we drop from now on all the references to $t$ in our notation. 

Our proof will take advantage of the dual nature of electrical flows. In particular, it will be instrumental for us to consider the vertex potentials $\vphi$ that induce the electrical flow $\hff$ via \eqref{eq:potential_flow_def}. The crucial property of these potentials is that they provide an embedding of all the vertices of $\oG$ into a line. To make it precise, for a given arc $e=(u,v)$, let us denote by $\phi_e^-$ (resp. $\phi_e^+$): the value of $\phi_u$ (resp. $\phi_v$), if $\phi_u\leq \phi_v$; and the value of $\phi_v$ (resp. $\phi_u$), otherwise. In other words, $\phi^-_e$ (resp. $\phi_e^+$) is the coordinate of the left-most (resp. right-most) endpoint of $e$ in this line embedding. 

Observe that by \eqref{eq:potential_flow_def} and definition of resistances $\rr$ (cf. \eqref{eq:hf_resistances}), we have that for a given arc $e=(u,v)$, the distance $\Delta_e$ between the embeddings of its endpoints is 
\begin{equation*} \label{eq:def_delta_embedding}
\Delta_e:=\phi_e^+-\phi_e^-=|\phi_u-\phi_v|=|\hf_e|r_e=\frac{\mu_e|\hf_e|}{f_e^2}=\frac{\mu_e\rho(\hff,\ff)_e}{f_e}, 
\end{equation*}
and thus by Fact \ref{fa:central_vs_max_min}
\begin{equation}\label{eq:delta_embedding_estimation}
(1-\hgamma)\hmu(\ff,\ss,\vnu)\frac{\rho(\hff,\ff)_e}{f_e}\leq \frac{\Delta_e}{\nu_e}\leq (1+\hgamma)\hmu(\ff,\ss,\vnu)\frac{\rho(\hff,\ff)_e}{f_e}. 
\end{equation}
Furthermore, for two subsets $T,U\subseteq \oV$ of vertices of $\oG$, let us define the distance $\dist{T}{U}$ between these sets to be
\begin{equation}
\label{eq:distance_def_local}
\dist{T}{U}:=\min_{v\in T,u\in U} |\phi_v-\phi_u|.
\end{equation}
Also, let us call two such subsets $T\subseteq \oV$ and $U\subseteq \oV$, {\em $(\Delta,k)$-separated}, for some $\Delta>0$ and integer $k\geq 0$, if $\dist{T}{U}\geq \Delta$ and $\min\{\aa(T),\aa(U)\}\geq k$, where $\aa(U'):=\sum_{v\in U'} a_v$ and $a_v$ is defined in \eqref{eq:def_of_a}.

Now, assume for the sake of contradiction that $\hff$ is not $\htheta$-smooth on the set of light arcs, i.e., there exists an $l^*\leq \log \htheta^{-3}$ such that
\begin{equation}\label{eq:embedd_local_non_smooth}
\vnu(\Cset{l^*}{\hff}\setminus E_H^t)=\vnu(S^*)>\floor{\htheta^{3} 2^{3l^*}},
\end{equation}
where $S^*$ denotes $\Cset{l^*}{\hff}\setminus E_H^t$ with the set $\Cset{l^*}{\hff}$ defined by \eqref{eq:def_of_C} and $E_H^t$ denotes the set of heavy arcs. Our main goal is to show that in this case there exist two subsets $T,U\subseteq \oV$ of vertices that are $(\Delta^*,k^*)$-separated with 
\begin{equation}
\label{eq:def_delt_star_k_star}
\Delta^*:= \frac{\cheavy \hm^{3\eta}}{14\cdot 2^{l^*}} \ \ \mathrm{ and } \ \ k^*:= \frac{2^{2l^*}\hm^{1-6\eta}(\hmu(\ff,\ss,\vnu))^2}{\ckstar \fauxiliary^2},
\end{equation}
where $\cheavy$ is the constant from Definition \ref{def:heavy}, $\fauxiliary$ is given by \eqref{eq:def_of_fauxliary}, and  $\ckstar$ is a sufficiently large constant that does not depend on $\cheavy$ and will be set later.

To motivate this goal, we prove the following lemma.

\begin{lemma}
\label{lem:separated_sets}
If there exist $T,U\subseteq \oV$ that are $(\Delta^*,k^*)$-separated then 
\[
\energy{\rr}{\hff}>\cenergy\hm \hmu(\ff,\ss,\vnu),
\]
provided $\cheavy$ is chosen to be large enough.
\end{lemma}

Observe that the conclusion of this lemma violates the bound from Lemma \ref{lem:energy_bound}. Thus, the resulting contradiction would allows us to conclude that $\hff$ indeed needs to be $\htheta$-smooth on the set of light arcs, as we wanted to prove. 

\begin{proof}
Note that as $\dist{T}{U}\geq \Delta^*$, it must be the case that either $\dist{\{\ov\}}{U}\geq \frac{\Delta^*}{2}$ or $\dist{\{\ov\}}{T}\geq \frac{\Delta^*}{2}$. (Recall that $\ov$ is the special vertex of $\oG$ that is adjacent to all the auxiliary arcs.) Let us assume -- without loss of generality -- that the first case holds.

Now, as $\min\{\aa(T),\aa(U)\}\geq k^*$, we know that, in particular, $\aa(U)\geq k^*$. This, in turn, means that at least $k^*$ of auxiliary arcs $e$ must have $\Delta_e\geq \frac{\Delta^*}{2}$. Furthermore, by Lemma \ref{lem:auxiliary_flow_growth}, we know that all but $O(\hm^{8\eta})$ of these arcs have measure $1$. So, as $k^*$ is $\tOm{\hm^{1-6\eta}}$, by ensuring that the constant $\ceta$ in the definition of $\eta$ (\eqref{eq:def_eta}) is big enough, we can conclude that the set $\hS$ of auxiliary arcs with $\Delta_e\geq \frac{\Delta^*}{2}$ and $\nu_e=1$ has size of at least $\frac{k^*}{2}$.

So, by \eqref{eq:delta_embedding_estimation} and Lemma \ref{lem:auxiliary_flow_growth}, we have that, for any such arc $e$ in $\hS$,
\[
\rho(\hff,\ff)_e \geq \frac{\Delta_e f_e}{(1+\hgamma)\nu_e \hmu(\ff,\ss,\vnu)} \geq \frac{\Delta_e \fauxiliary}{(1+\hgamma)\cfreeze \hmu(\ff,\ss,\vnu)}\geq \frac{\Delta^* \fauxiliary}{4\cfreeze \hmu(\ff,\ss,\vnu)},
\]
where we used the fact that $\Delta_e\geq \frac{\Delta^*}{2}$ and $\nu_e=1$, for all $e$ in $\hS$.

Now, the above inequality enables us to lowerbound the energy $\energy{\rr}{\hff}$ of the flow $\hff$ using sole contribution of arcs in $\hS$. We get that \begin{eqnarray*}
\energy{\rr}{\hff} &\geq& \sum_{e\in \hS} r_e \hf_e^2 \geq \sum_{e\in \hS} (1-\hgamma) \hmu(\ff,\ss,\vnu) \rho(\hff,\ff)_e^2 \geq  (1-\hgamma) \hmu(\ff,\ss,\vnu) |\hS| \left(\frac{\Delta^* \fauxiliary}{4\cfreeze \hmu(\ff,\ss,\vnu)}\right)^2\\
&\geq & \hmu(\ff,\ss,\vnu) k^* \left(\frac{\Delta^* \fauxiliary}{8\cfreeze \hmu(\ff,\ss,\vnu)}\right)^2 \geq \Omega\left(\hmu(\ff,\ss,\vnu) \hm^{1-6\eta} 2^{2l^*} \left(\frac{\cheavy^2 \hm^{6\eta}}{\cfreeze^2\ckstar 2^{2l^*}}\right)\right)\\
& \geq & \Omega\left(\hmu(\ff,\ss,\vnu) \hm \left(\frac{\cheavy^2}{\cfreeze^2\ckstar}\right)\right),
\end{eqnarray*}
where we used the definition of $\rr$ \eqref{eq:hf_resistances} and Fact \ref{fa:central_vs_max_min}.

So, once $\cheavy$ is chosen to be large enough constant -- which we can always ensure to be the case -- the lemma follows. (Note that at this point the constant $\cfreeze$ is fixed already and we will make sure that when we later set the constant $\ckstar$, it does not depend on the value of $\cheavy$.) 
\end{proof}

\subsubsection*{Finding the $(\Delta^*,k^*)$-separated Sets}

In the light of the above, it remains to establish how condition \eqref{eq:embedd_local_non_smooth} implies the existence of such $(\Delta^*,k^*)$-separated sets $T^*$ and $U^*$.  To this end, for a given $x\in \bbR$, let us define $V_x^-$ (resp. $V_x^+$) to be the set of vertices $v$ with $\phi_v\leq x$ (resp. $\phi_v\geq x$). Also, let $E_x$ denote the set of arcs $e$ of $\oG$ such that $\phi_e^-\leq x \leq \phi_e^+$. 

Now, let $x^*$ be the smallest $x$ such that $\aa(V_x^-)\geq k^*$. If $\aa(V_{x^*+\Delta^*})\geq k^*$ then taking $T^*:=V_{x^*}^-$ and $U^*=V_{x^*+\Delta^*}^+$ will clearly constitute the $(\Delta^*,k^*)$-separated sets we are looking for. 

So, we can focus on the case that $\aa(V_{x^*+\Delta^*})< k^*$. Let us then take $T^*:=V_{x^*}^+\cap V_{x^*+\Delta^*}^-$. Note that, as $3k^*$ is smaller than the number of all auxiliary arcs, we need to have $\aa(T^*)\geq k^*$. Next, let us take $U^*:=V_{x^*-\Delta^*}^-\cup V_{x^*+2\Delta^*}^+$. Clearly, $\dist{T^*}{U^*}\geq \Delta^*$. Therefore, once we show that $\aa(U^*)\geq k^*$, $T^*$ and $U^*$ will constitute the desired $(\Delta^*,k^*)$-separated sets.

We proceed now to showing that indeed $\aa(U^*)\geq k^*$. Let us define $F(x):=\sum_{e\in E_x} f_e$ (resp. $\hF(x):=\sum_{e\in E_x} |\hf_e|$) to be the total flow of $\ff$ (resp. $\hff$) flowing through the arcs in $E_x$. We will be interested in two quantities
\[
A^*:=\int_{\bbR\setminus I^*} F(x) dx \ \ \mathrm{ and } \ \ \hA:= \int_{\bbR\setminus I^*} \hF(x) dx,
\]
where $I^*$ is an interval $[x^*-\Delta^*,x^*+2\Delta^*]$. (Observe that if the interval $I^*$ was not excluded, $\hA$ would be equal to the energy $\energy{\rr}{\hff}$ of the flow $\hff$.)

\paragraph{Lowerbounding $\hA$} First, we want to lowerbound $\hA$. To this end, we note that by \eqref{eq:delta_embedding_estimation}, for any $e\in S^*$ (recall that $S^*:=\Cset{l^*}{\hff}\setminus E_H^t$ and thus, in particular, is contain only light arcs), we have that
\begin{eqnarray*}
\Delta_e &\geq & (1-\hgamma)\hmu(\ff,\ss,\vnu)\nu_e\frac{\rho(\hff,\ff)_e}{f_e} \geq (1-\hgamma)\hmu(\ff,\ss,\vnu)\frac{\rho(\hff,\ff)_e}{\fheavy}\\ 
&\geq & (1-\hgamma)\cheavy \frac{\rho(\hff,\ff)_e}{\hm^{\frac{1}{2}-3\eta}} \geq  (1-\hgamma)\cheavy \frac{\hm^{3\eta}}{2^{l^*+1}} \geq 6 \Delta^*,
\end{eqnarray*}
where we also used \eqref{eq:def_of_C}, \eqref{eq:def_delt_star_k_star}, and Definition \ref{def:heavy}.

As the interval $I^*$ has length $3\Delta^*$, this means that for any arc $e\in S^*$, the interval $[\phi_e^-,\phi_e^+]\setminus I^*$ has length of at least 
\[
\Delta_e-3\Delta^* \geq \frac{\Delta_e}{2}.
\]

This, in turn, implies that even if we account  for contributions of the arcs from $S^*$ only, we have that
\begin{eqnarray}\label{eq:precond_hA_lowerbound}
\hA &=& \int_{\bbR\setminus I^*} \hF(x) dx \geq \frac{1}{2} \sum_{e\in S^*} \Delta_e |\hf_e| \geq \frac{(1-\hgamma)\hmu(\ff,\ss,\vnu)}{2} \sum_{e\in S^*} \frac{\rho(\hff,\ff)_e\nu_e|\hf_e|}{f_e}\nonumber\\
 &\geq& \frac{(1-\hgamma)\hmu(\ff,\ss,\vnu)}{2} \sum_{e\in S^*} \rho(\hff,\ff)_e^2\nu_e \geq \frac{\hmu(\ff,\ss,\vnu)}{5} \sum_{e\in S^*} \frac{\nu_e\hm}{2^{2l^*}}\\
 &=& \frac{\hmu(\ff,\ss,\vnu)}{5\cdot 2^{2l^*}} \vnu(S^*) \hm \geq \frac{\hmu(\ff,\ss,\vnu)}{5} \htheta^3 2^{l^*} \hm =\frac{\hmu(\ff,\ss,\vnu)}{5} 2^{l^*} \hm^{1-3\eta}, \nonumber
\end{eqnarray}
where we used \eqref{eq:def_of_C}, \eqref{eq:delta_embedding_estimation}, and \eqref{eq:embedd_local_non_smooth}. 

\paragraph{Upperbounding $A^*$} Now, we want to upperbound the value of $A^*$. To do that, let us define $\oS$ to be the set of arcs that have at least one endpoint outside of the interval $I^*$. Note that by our way of setting up the auxiliary arcs and the fact that by Lemma \ref{lem:auxiliary_flow_growth} and \eqref{eq:measure_increase_phase}, the total increase of measure of arcs during the $\htheta$-improvement phase is $\tO{\hm^{8\eta}}$, we have that
\begin{equation}
\label{eq:precond_os_vs_ustar}
\aa(U^*)\geq \frac{\vnu(\oS)}{3}-\tO{\hm^{8\eta}}.
\end{equation} 
So, if we are able to show that $\vnu(\oS)\geq 4k^* = \tOm{\hm^{1-6\eta}}$ and ensure again that the constant $\ceta$ in definition of $\eta$ \eqref{eq:def_eta} is large enough, we will prove that $\aa(U^*)\geq k^*$, as desired.

To establish such lowerbound on $\vnu(\oS)$, we use \eqref{eq:delta_embedding_estimation} and observe that
\begin{equation*}
A^*=\int_{\bbR\setminus I^*} F(x) dx \leq \sum_{e\in \oS} \Delta_e f_e \leq (1+\hgamma) \hmu(\ff,\ss,\vnu) \sum_{e\in \oS} \nu_e \rho(\hff,\ff)_e,
\end{equation*}
where we noted that the only arcs that can contribute to $A^*$ are all in the set $\oS$. 
Therefore, by Fact \ref{fa:central_vs_max_min} and Cauchy-Schwarz inequality, we have that
\[
(1+\hgamma) \hmu(\ff,\ss,\vnu) \sum_{e\in \oS} \nu_e \rho(\hff,\ff)_e \leq (1+\hgamma) \hmu(\ff,\ss,\vnu) \sqrt{(\sum_{e\in \oS} \nu_e \rho(\hff,\ff)_e^2) \vnu(\oS)}.
\]
So, putting the above two bounds together, we get that
\begin{equation}\label{eq:upperbound_on_A_star}
A^*\leq (1+\hgamma) \hmu(\ff,\ss,\vnu) \sum_{e\in \oS} \nu_e \rho(\hff,\ff)_e \leq 5\hmu(\ff,\ss,\vnu) \sqrt{\hm\vnu(\oS)},
\end{equation}
where we also used \eqref{eq:worst_case_rho} and Lemma \ref{lem:hfft_energy_bound}.

At this point, our last needed observation is captured by the following lemma.
\begin{lemma}
\label{lem:Astar_vs_Ahat}
For any $x\in \bbR$, we have that $F(x)\geq\hF(x)$.
\end{lemma} 

Notice that once the above lemma is established, we have that 
\[
A^*=\int_{\bbR\setminus I^*} F(x) dx\geq \int_{\bbR\setminus I^*} \hF(x) dx = \hA,
\]
and, as a result, we can put \eqref{eq:precond_hA_lowerbound} and \eqref{eq:upperbound_on_A_star} together to obtain
\begin{eqnarray*}
\vnu(\oS) &\geq& \frac{1}{\hm}\left( \frac{A^*}{5\hmu(\ff,\ss,\vnu)}\right)^2 \geq  \frac{1}{\hm}\left( \frac{\hA}{5\hmu(\ff,\ss,\vnu)}\right)^2\geq \frac{1}{\hm}\left( \frac{2^{l^*}\hm^{1-3\eta}}{25}\right)^2\\
& \geq & \Omega(2^{2l^*} \hm^{1-6\eta}) \geq 4k^*, 
\end{eqnarray*}
once the constant $\ckstar$ in the definition \eqref{eq:def_delt_star_k_star} of $k^*$ is taken to be large enough. (Note that the term $\frac{\hmu(\ff,\ss,\vnu)}{\fauxiliary}$ in the definition of $k^*$ \eqref{eq:def_delt_star_k_star} is bounded by a constant that is independent of $\cheavy$. So, indeed $\ckstar$ does not depend on $\cheavy$, as we wanted to ensure.)

Therefore, by \eqref{eq:precond_os_vs_ustar}, the above bounds shows that indeed $\aa(U^*)\geq k^*$, as needed. 

At this point, we just need to perform the remaining proof of the lemma and the analysis of our improved algorithm will be concluded.

\begin{proof}
The simple, but fundamental, observation we need to make here is that the flow $\hff$ -- being an electrical $\hvsigma$-flow induced by vertex potentials $\vphi$ via relationship \eqref{eq:potential_flow_def} -- is always flowing in one direction, i.e., from left to right, with respect to the line embedding given by $\vphi$.  This, together with the fact that $\hff$ is a $\hvsigma$-flow, implies that
\[
\sum_{v\in V_x^+}\hsigma_v = \hF(x). 
\]

On the other hand, $\ff$ is also a feasible $\hvsigma$-flow, which means that the net inflow into $V_x^+$ of $\ff$ has to be at least $\sum_{v\in V_x^+}\hsigma_v$. This gives us that
\[
F(x)\geq \sum_{v\in V_x^+}\hsigma_v = \hF(x),
\]
as we wanted to establish.
\end{proof}

%% file: files/interior_point.tex
\section{Electrical Flows and the Central Path}\label{sec:proof_main_interior_point}

In this section, we describe how we can use electrical flows to advance our solution along the central path. In other words, we describe and analyze the implementation of the improvement step and thus prove Theorem \ref{thm:main_interior_point}. This implementation is directly inspired by -- and, in fact, can be seen as a reinterpretation of -- the improvement steps used in path-following method. 

Recall that in the improvement step, we are given a $\hgamma$-centered $\hvsigma$-feasible solution $(\ff^t,\ss^t,\vnu^t)$ and our goal is to compute, in $\tO{\hm}$ time, a $\hgamma$-centered $\hvsigma$-feasible solution $(\ff^{t+1},\ss^{t+1},\vnu^{t+1})$ with 
\begin{equation}
\label{eq:inter_progress}
\hmu(\ff^{t+1},\ss^{t+1},\vnu^{t+1})\leq (1-\delta^t) \hmu(\ff^t,\ss^t,\vnu^t).
\end{equation}

We perform this improvement in two main steps. The first one -- the descent step -- uses the electrical flow $\hff^t$ associated with  $(\ff^{t},\ss^{t},\vnu^{t})$ and the corresponding vertex potentials $\hphi^t$ that induce it, to perform a primal and dual update that results in a new, intermediate, solution $(\off^t,\oss^t,\ovnu^t)$. This intermediate solution is $\hvsigma$-feasible and has $\hmu(\off^t,\oss^t,\ovnu^t)\leq (1-\delta^t)\hmu(\ff^{t},\ss^{t},\vnu^{t})$ as desired, but it might be not $\hgamma$-centered anymore. To fix that, in the second -- centering -- step, we compute the desired solution $(\ff^{t+1},\ss^{t+1},\vnu^{t+1})$ out of $(\off^t,\oss^t,\ovnu^t)$ by using another electrical flow computation that again provides a primal and dual update.

 We describe and analyze both of these steps below. Note that as each of these two steps requires only one computation of electrical flow, it can be easily implemented to run in $\tO{\hm}$ time, as needed.

\paragraph{Descent Step}

Let $\hff^t$ be the electrical $\hvsigma$-flow associated with the solution $(\ff^{t},\ss^{t},\vnu^{t})$ and let $\hvphi^t$ be the vertex potentials that induce $\hff^t$. Consider a new primal-dual solution $(\off^t,\oss^t,\ovnu^t)$ given by
\begin{eqnarray}
\of_e^t & := & (1-\delta^t) f_e^t + \delta^t \hf_e^t\\
\os_e^t & := & s_e^t - \frac{\delta^t}{(1-\delta^t)} (\hphi^t_u-\hphi^t_v)=s_e^t - \delta^t \frac{s_e^t}{(1-\delta^t)f_e^t} \hf_e^t\\
\onu_e^t& := & \nu^t_e,
\end{eqnarray}
for each arc $e=(v,u)$ in $\hG$, where $\delta^t$ satisfies conditions of the theorem and we also used the definition \eqref{eq:hf_resistances} of the resistances that determine $\hff^t$, as well as, the relationship \eqref{eq:potential_flow_def} between electrical flow and the vertex potentials that induce it.

Observe that as $\off^t$ is a convex combination of two $\hvsigma$-flows -- the flows $\ff^t$ and $\hff^t$ -- it also is an $\hvsigma$-flow.  Furthermore, as all $\nu_e^t\geq 1$, we have $\norm{\vrho(\hff^t,\ff^t)}{\vnu^t,4}\geq \inorm{\vrho(\hff^t,\ff^t)}$ and thus, for each arc $e$,
\begin{equation}
\label{eq:interior_bound_2}
\of_e^t  =  (1-\delta^t) f_e^t + \delta^t \hf_e^t \geq (1-\delta^t) f_e^t - \delta^t |\hf_e^t| =  (1-\delta^t - \delta^t \rho(\hff^t,\ff^t)_e) f_e^t \geq (1-\frac{1}{2}-\sqrt{\hgamma})f_e^t>0,
\end{equation}
and similarly
\[
\os_e^t = s_e^t - \delta^t \frac{s_e^t}{(1-\delta^t)f_e^t} \hf_e^t \geq s_e^t - \delta^t \frac{s_e^t}{(1-\delta^t)f_e^t} |\hf_e^t| = s_e^t - \delta^t \frac{s_e^t}{(1-\delta^t)} \rho(\hff^t,\ff^t)_e \geq (1-2\sqrt{\hgamma})s_e^t>0.
\]
So, $(\off^t,\oss^t,\ovnu^t)$ is $\hvsigma$-feasible, as desired.

Let us now analyze the value of $\hmu(\off^t,\oss^t,\ovnu^t)$. To this end, observe that, for any arc $e$,
\begin{eqnarray}
\ohmu_e^t =\frac{\of^t_e \os_e^t}{\onu^t_e} & = & (\nu_e^t)^{-1}((1-\delta^t) f_e^t + \delta^t \hf_e^t)(s_e^t - \delta^t \frac{s_e^t}{(1-\delta^t)f_e^t} \hf_e^t)\notag\\
&=& (\nu_e^t)^{-1}\left((1-\delta^t)f_e^ts_e^t +  \delta^t \hf_e^t s_e^t - \delta^t \frac{s_e^t}{(1-\delta^t)f_e^t} \hf_e^t (1-\delta^t) f_e^t - (\delta^t)^2 \frac{s_e^t}{(1-\delta^t)f_e^t}(\hf_e^t)^2\right) \notag\\
&=& (\nu_e^t)^{-1}\left((1-\delta^t) \mu_e^t + \delta^t \hf_e^t s_e^t - \delta^t \hf_e^t s_e^t - (\delta^t)^2 \frac{s_e^t}{(1-\delta^t)} f_e^t \rho(\hff^t,\ff^t)_e^2\right)\label{eq:duality_change}\\
&=& \left(1-\delta^t- \frac{(\delta^t\rho(\hff^t,\ff^t)_e)^2}{(1-\delta^t)}\right) \hmu_e^t.\notag
\end{eqnarray}
So, by definition \eqref{eq:def_hmu} and the fact that $\frac{(\delta^t\rho(\hff^t,\ff^t)_e)^2}{(1-\delta^t)}\geq 0$ for all $e$, we see that 
\[
\hmu(\off^t,\oss^t,\ovnu^t)\leq (1-\delta^t) \hmu(\ff^{t},\ss^{t},\vnu^{t})
\]
and this inequality would be an equality if the second-order terms (i.e., terms quadratic in $\delta^t$) were ignored. (Also, if these terms were not present, the centrality of the solution would be preserved too.)

Finally, let us focus on analyzing the centrality of $(\off^t,\oss^t,\ovnu^t)$. To this end,  note that by definition \eqref{eq:def_centrality} and by \eqref{eq:duality_change} above we have
\begin{eqnarray}
\norm{\ohvmu^t-\hmu(\off^t,\oss^t,\ovnu^t)\onev}{\ovnu^t,2}^2 & \leq & \norm{\ohvmu^t-(1-\delta^t)\hmu(\ff^t,\ss^t,\vnu^t)\onev}{\ovnu^t,2}^2 \leq \sum_e \onu_e^t (\ohmu_e^t-(1-\delta^t)\hmu(\ff^t,\ss^t,\vnu^t))^2\notag\\
&=& \sum_e \nu_e^t \left( (1-\delta^t)(\hmu_e^t-\hmu(\ff^t,\ss^t,\vnu^t))-\frac{(\delta^t\rho(\hff^t,\ff^t)_e)^2}{(1-\delta^t)}\hmu^t_e\right)^2\notag\\
&\leq & 2 \left((1-\delta^t)^2\sum_e \nu_e^t (\hmu^t_e-\hmu(\ff^t,\ss^t,\vnu^t))^2 + \sum_e \nu_e^t \frac{(\delta^t\rho(\hff^t,\ff^t)_e)^4}{(1-\delta^t)^2}(\hmu_e^t)^2 \right)\label{eq:centrality_increase}\\
&\leq & 2 \left((1-\delta^t)^2\norm{\hvmu^t-\hmu(\ff^t,\ss^t,\vnu^t)\onev}{\vnu^t,2}^2 + \frac{(1+\hgamma)^2\hmu(\ff^t,\ss^t,\vnu^t)^2}{(1-\delta^t)^2} (\delta^t)^4  \sum_e \nu_e^t \rho(\hff^t,\ff^t)_e^4 \right)\notag\\
&\leq & 2 \left((1-\delta^t)^2 \hgamma^2 \hmu(\ff^t,\ss^t,\vnu^t)^2 + \frac{(1+\hgamma)^2\hmu(\ff^t,\ss^t,\vnu^t)^2}{(1-\delta^t)^2} (\delta^t)^4 \sum_e \nu_e^t \rho(\hff^t,\ff^t)_e^4 \right)\notag\\
&\leq & 2 \left((1-\delta^t)^2 \hgamma^2 + \frac{(1+\hgamma)^2}{(1-\delta^t)^2} (\delta^t)^4 \norm{\vrho(\hff^t,\ff^t)}{\vnu^t,4}^4 \right)\hmu(\ff^t,\ss^t,\vnu^t)^2\notag\\
&\leq& 10 \hgamma^2\hmu(\off^t,\oss^t,\ovnu^t)^2. \notag
\end{eqnarray}
In the above derivation, the first inequality follows as the $\norm{\vmu-t\onev}{\vnu,2}$ is always minimized by taking $t=\hmu(\ff,\ss,\vnu)$. We also used the fact that $(\ff^t,\ss^t,\vnu^t)$ is $\hgamma$-centered, Fact \ref{fa:central_vs_max_min} and the upperbound on $\delta^t$.

Therefore, we see that the price of making progress on the duality gap is that the centrality of our solution could deteriorate by a factor of at most three. 

\paragraph{Centering Step}

To alleviate this possible increase of centrality, we apply a second step that restores the centrality back within the desired bounds while not increasing the duality gap (so to not to counter the progress on the duality gap that we just achieved).

To this end, consider a flow $\off^*$ in $\hG$ defined as
\begin{equation}
\label{eq:def_of_offstar}
\of^*_e := \frac{\ohmu_e^t-\hmu(\off^t,\oss^t,\ovnu^t)}{\ohmu_e^t}\of^t_e,
\end{equation}
for every arc $e$ of $\hG$. Note that the flow $\off^*$ might (and actually will) not be feasible in $\hG$, as some of $\of^*_e$ can be negative.

Now, consider a flow $\off'$ given by
\begin{equation}
\label{eq:offprime_def}
\of_e' := \of^t_e - \of^*_e = (1-\frac{\ohmu_e^t-\hmu(\off^t,\oss^t,\ovnu^t)}{\ohmu_e^t})\of^t_e=\frac{\hmu(\off^t,\oss^t,\ovnu^t)}{\ohmu_e^t}\of^t_e,
\end{equation}
for each arc $e$. Observe that $\off'$ is feasible in $\hG$ (i.e., $\of_e'\geq 0$, for all $e$) and 
\begin{equation}
\label{eq:perfect_centering_of_intermediate}
\frac{\of_e'\os_e^t}{\onu_e^t} = \frac{\hmu(\off^t,\oss^t,\ovnu^t)}{\ohmu_e^t\onu_e^t}\of^t_e \os_e^t = \hmu(\off^t,\oss^t,\ovnu^t),
\end{equation}
for each arc $e$. That is, $(\off',\oss^t,\ovnu^t)$ is $0$-centered with $\hmu(\off',\oss^t,\ovnu^t)=\hmu(\off^t,\oss^t,\ovnu^t)$.

So, this solution would be a perfect candidate for $(\ff^{t+1},\ss^{t+1},\vnu^{t+1})$ except that the flow $\off'$ does not need to be a $\hvsigma$-flow and thus this solution might not be $\hvsigma$-feasible. 

To fix that -- and obtain our desired solution $(\ff^{t+1},\ss^{t+1},\vnu^{t+1})$ -- let $\tvsigma$ be the demand vector of the flow $\off^*$, and consider an electrical $\tvsigma$-flow $\tff^t$ that corresponds to resistances
\begin{equation}
\label{eq:tfft_resistances}
\tr_e^t := \frac{\os_e^t}{\of'_e},
\end{equation}
for each arc $e$ and let $\tvphi$ be the corresponding vertex potentials.

Let us define $(\ff^{t+1},\ss^{t+1},\vnu^{t+1})$ to be
\begin{eqnarray}
f_e^{t+1} & := & \of'_e + \tf_e^t\notag\\
s_e^{t+1} & := & \os_e^t - (\tphi^t_u-\tphi^t_v) =\os_e^t - \frac{\os_e^t}{\of_e'} \tf_e^t\label{eq:new_ff_t1_def}\\
\nu_e^{t+1}& := & \onu^t_e,\notag
\end{eqnarray}
for each arc $e=(v,u)$.

Clearly, now $\ff^{t+1}$ is a $\hvsigma$-flow, as desired. Let us analyze its centrality. To this end, let us fix some arc $e$, and notice that
\begin{eqnarray*}
\hmu^{t+1}_e &=& \frac{f_e^{t+1}s_e^{t+1}}{\nu_e^{t+1}}=\frac{(\of'_e + \tf_e^t)(\os_e^t - \frac{\os_e^t}{\of_e'}\tf_e^t)}{\nu_e^{t+1}}\\ &=&(\onu_e^{t})^{-1}\left(\of'_e\os_e^t+\os_e^t\tf_e^t-\os_e^t\tf_e^t-\frac{\os_e^t}{\of_e'} (\tf_e^t)^2\right)\\
&=&\hmu(\off^t,\oss^t,\ovnu^t)-\frac{\os_e^t}{\onu^t_e\of_e'} (\tf_e^t)^2,
\end{eqnarray*}
where we used \eqref{eq:perfect_centering_of_intermediate}. So, we see in particular that 
\[
\hmu(\ff^{t+1},\ss^{t+1},\vnu^{t+1})\leq \hmu(\off^t,\oss^t,\ovnu^t)\leq (1-\delta^t)\hmu(\ff^{t},\ss^{t},\vnu^{t}),
\]
as needed. 

Now, by our derivation above, we have that
\[
\norm{\hvmu^{t+1}-\hmu(\ff^{t+1},\ss^{t+1},\vnu^{t+1})\onev}{\vnu^{t+1},2}^2\leq \norm{\hvmu^{t+1}-\hmu(\off^{t},\ss^{t},\ovnu^{t})\onev}{\vnu^{t+1},2}^2 = \sum_e \onu_e^t \left(\frac{\os_e^t}{\onu^t_e\of_e'} (\tf_e^t)^2\right)^2.
\]

To bound the resulting expression, let us note that by Cauchy-Schwarz inequality and the fact that measures are always at least $1$ we have 
\[
\sum_e \onu_e^t \left(\frac{\os_e^t}{\onu^t_e\of_e'} (\tf_e^t)^2\right)^2\leq \left(\max_{e} \frac{\os_e^t}{\onu^t_e\of_e'} (\tf_e^t)^2\right)\left(\sum_{e} \frac{\os_e^t}{\of_e'} (\tf_e^t)^2\right)\leq \left(\sum_{e} \frac{\os_e^t}{\of_e'} (\tf_e^t)^2\right)^2.
\]

Now, the key insight here is that by \eqref{eq:tfft_resistances}, 
\[
\left(\sum_{e} \frac{\os_e^t}{\of_e'} (\tf_e^t)^2\right)^2=\left(\energy{\trr^t}{\tff^t}\right)^2.
\]
So, by bounding the energy of the electrical flow $\tf_e^t$ we will be able to bound the centrality of our solution $\hmu(\ff^{t+1},\ss^{t+1},\vnu^{t+1})$. 
To bound this energy, we will first bound the energy $\energy{\trr^t}{\off^*}$ of the flow $\off^*$ and use the fact that both $\off^*$ and $\tff^t$ are $\tvsigma$-flows and thus, by definition, $\tff^t$ is minimizing the energy among all the $\tvsigma$-flows. 

Observe that by definition \eqref{eq:def_of_offstar} of the flow $\off^*$, the fact that $(\off^t,\oss^t,\ovnu^t)$ is $3\hgamma$-centered -- cf. \eqref{eq:centrality_increase} -- and Fact \ref{fa:central_vs_max_min}, we have that
\begin{eqnarray}
\energy{\trr^t}{\off^*} &=& \sum_e \frac{\os_e^t}{\of'_e}\left(\frac{\ohmu_e^t-\hmu(\off^t,\oss^t,\ovnu^t)}{\ohmu_e^t}\of^t_e\right)^2 \leq \sum_e \frac{\onu_e^t\ohmu_e^t}{\of'_e}\frac{(\ohmu_e^t-\hmu(\off^t,\oss^t,\ovnu^t))^2\of_e^t}{(1-3\hgamma)\ohmu_e^t \hmu(\off^t,\oss^t,\ovnu^t)}\notag\\
&\leq & \frac{1}{(1-3\hgamma)}\sum_e \onu_e^t \frac{(\ohmu_e^t-\hmu(\off^t,\oss^t,\ovnu^t))^2\of_e^t}{\hmu(\off^t,\oss^t,\ovnu^t)\of_e'} = \frac{1}{(1-3\hgamma)}\sum_e \onu_e^t \frac{(\ohmu_e^t-\hmu(\off^t,\oss^t,\ovnu^t))^2\ohmu_e^t}{\hmu(\off^t,\oss^t,\ovnu^t)^2}\notag\\
&\leq & \frac{(1+3\hgamma)}{(1-3\hgamma)}\sum_e \onu_e^t \frac{(\ohmu_e^t-\hmu(\off^t,\oss^t,\ovnu^t))^2}{\hmu(\off^t,\oss^t,\ovnu^t)} =  \frac{(1+3\hgamma)}{(1-3\hgamma)} \frac{\norm{\ohvmu^t-\hmu(\off^t,\oss^t,\ovnu^t)\onev}{\ovnu^t,2}^2}{\hmu(\off^t,\oss^t,\ovnu^t)} \label{eq:energy_estimate_final} \\
& \leq  & \frac{(1+3\hgamma)}{(1-3\hgamma)} 9 \hgamma^2 \hmu(\off^t,\oss^t,\ovnu^t) \leq 10 \hgamma^2 \hmu(\off^t,\oss^t,\ovnu^t)\leq 20 \hgamma^2 \hmu(\ff^{t+1},\ss^{t+1},\vnu^{t+1}),\notag
\end{eqnarray}
where we also used the definition \eqref{eq:offprime_def} of the flow $\off'$.

In the light of the above discussion, we can conclude that 
\begin{eqnarray*}
\norm{\hvmu^{t+1}-\hmu(\ff^{t+1},\ss^{t+1},\vnu^{t+1})\onev}{\vnu^{t+1},2}^2 &\leq & \sum_e \onu_e^t \left(\frac{\os_e^t}{\onu^t_e\of_e'} (\tf_e^t)^2\right)^2 \leq \left(\sum_{e} \frac{\os_e^t}{\of_e'} (\tf_e^t)^2\right)^2\\
&\leq &  \left(20 \hgamma^2 \hmu(\ff^{t+1},\ss^{t+1},\vnu^{t+1})\right)^2 \leq \hgamma^2 \hmu(\ff^{t+1},\ss^{t+1},\vnu^{t+1})^2,
\end{eqnarray*}
as $\hgamma\leq \frac{1}{20}$. So, indeed $(\ff^{t+1},\ss^{t+1},\vnu^{t+1})$ is $\hgamma$-centered. 

Now, to prove that $(\ff^{t+1},\ss^{t+1},\vnu^{t+1})$ is also $\hvsigma$-feasible, we just need to show that for any arc $e$, 
\[
\rho(\tff^t,\off')_e=\frac{|\tf^t_e|}{\of'_e}\leq \frac{1}{2}.
\]
To this end, note that by \eqref{eq:perfect_centering_of_intermediate} and \eqref{eq:energy_estimate_final} we have
\begin{equation}
\label{eq:interior_bound_1}
\nu_e^t\hmu(\off^t,\oss^t,\ovnu^t)\rho(\tff^t,\off')_e^2 = \frac{\os_e^t\of'_e}{(\of'_e)^2}(\tf_e^t)^2 = \frac{\os_e^t}{\of'_e}(\tf_e^t)^2\leq \energy{\trr^t}{\tff^t}\leq  \energy{\trr^t}{\off^*} \leq 10 \hgamma^2 \hmu(\off^t,\oss^t,\ovnu^t)\leq \frac{1}{1600} \hmu(\off^t,\oss^t,\ovnu^t).
\end{equation}

Thus, indeed, we can conclude that we obtained a $\hgamma$-centered $\hvsigma$-feasible solution $(\ff^{t+1},\ss^{t+1},\vnu^{t+1})$ with $\hmu(\ff^{t+1},\ss^{t+1},\vnu^{t+1})\leq (1-\delta^t)\hmu(\ff^{t},\ss^{t},\vnu^{t})$, as desired. 

This concludes the proof of the first part of the Theorem \ref{thm:main_interior_point}. The proof of the second part appears in Appendix \ref{app:interior}.

%% file: files/rounding.tex
\section{Rounding Fractional Bipartite $\bb$-Matchings}\label{sec:rounding}

In this section, we show how given a fractional $\bb$-matching $\xx$ in some bipartite graph $G=(P\cup Q, E)$ with $m=|E|$ edges, one can find in $\tO{m}$ time an integral $\bb$-matching $\xx^*$ in $G$ whose size is at least $\floor{\onorm{\xx}}$. In other words, we prove Theorem \ref{thm:rounding_matchings}.

\subsubsection*{Rounding Perfect Matchings}
Let us first consider the case when $\xx$ is just a fractional perfect matching, i.e., $b_v=1$ for all vertices and the size $\onorm{\xx}$ of $\xx$ is $\frac{\onorm{\bb}}{2}$, i.e., the fractional degree of each vertex in $\xx$ is $1$. We claim that in this case we can just use Theorem \ref{thm:regular_bipartite_matchings} to obtain an integral perfect matching in $\tO{m}$ time. 

To see why this is the case, consider a $|P|\times |Q|$ matrix $\MM^{\xx}$ in which rows and columns are indexed by vertices from $P$ and $Q$, respectively, and the entries are given by $M_{p,q}^{\xx}:=x_{(p,q)}$ if the edge $(p,q)$ exists in $G$; and $0$, otherwise. Observe that if $\xx$ is perfect and all $b_v$ are equal to $1$ then we need to have $|P|=|Q|$. Thus, $\MM^{\xx}$ is a square matrix. Furthermore, $\MM^{\xx}$ needs to be also doubly-stochastic, as for any row indexed by vertex $p\in P$ (resp. column indexed by vertex $q\in Q$), the sum $\sum_{q'\in Q} M^{\xx}_{p,q'}$ (resp. $\sum_{p'\in P} M^{\xx}_{p',q}$) of the entries in this row (resp. column) is equal to $\sum_{e\in E(p)} x_{e}=b_p=1$ (resp. $\sum_{e\in E(q)} x_{e}=b_q=1$). So, invoking Theorem \ref{thm:regular_bipartite_matchings}, we can obtain in $\tO{m}$ time an integral matching $\xx^*$ in the support of $\MM^{\xx}$ that is also the support of the edge set $E$ of our graph $G$. 

\subsubsection*{Rounding Non-Perfect Matchings}

Now, to recover the desired integral matching in the case when $\xx$ is not necessarily perfect (but still all $b_v$ are equal to $1$), our first step is to extend $\xx$ to a perfect matching $\oxx$ in a certain augmented graph $\oG$ that is created from $G$ by adding some dummy edges and vertices to it.

More precisely, let $d_P$ (resp. $d_Q$) be the total deficits of vertices in $P$ (resp. in $Q$), i.e., 
\[
d_P:=|P| - \sum_{e\in E(p), p\in P} x_e \mbox{  \ \ \ \  and \ \ \ \ } d_Q:=|Q| - \sum_{e\in E(q),q\in Q} x_e.
\]
 Note that the size $\onorm{\xx}$ of $\xx$ has to be exactly $|P|-d_P=|Q|-d_Q$. We add to the vertex set $Q$, $\ceil{d_P}$ (resp. to the vertex set $P$, $\ceil{d_Q}$) dummy vertices $\oq_1,\ldots, \oq_{\ceil{d_P}}$ (resp. $\op_1,\ldots, \op_{\ceil{d_Q}}$). Next, we extend the fractional matching $\xx$ to $\oxx$ by going over each non-dummy vertex $p\in P$ (resp. $q\in Q$) and fractionally matching it to the dummy vertices $\oq_1,\ldots ,\oq_{\ceil{d_P}}$ (resp. $\op_1,\ldots, \op_{\ceil{d_Q}}$), so to ensure that its fractional degree becomes $1$ and the fractional degree of dummy vertices never exceeds one. It is not hard to see that by employing a simple greedy approach we can achieve this goal in $\tO{m}$ time and, furthermore, ensure that: (1) each non-dummy vertex is matched to at most two dummy vertices in $\oxx$; (2) at the end, there are at most two dummy vertices, say, $\op_{\ceil{d_Q}}$ and $\oq_{\ceil{d_P}}$, (one on each side of the bipartition) that are yet not fully matched in $\oxx$. To alleviate the latter problem, we just match these two dummy vertices to each other (one can check that their deficits have to be equal) and take the set of edges $\oE$ of our augmented graph $\oG$ to be the support of the matching $\oxx$. (Note that by property (1), the size of this support will be still $O(m)$.)

Clearly, $\oxx$ is a perfect matching in $\oG$, so we can use the $\tO{m}$-time procedure we described above to get an integral perfect matching $\oxx^*$ in that graph. Once we do that, we take our desired integral matching $\xx^*$ in $G$ to be $\oxx^*$ after we removed from it all the edges of $\oxx^*$ that are not in $G$, i.e., all the edges that are incident to dummy vertices. Obviously, $\xx^*$ is a feasible matching in $G$ and it is integral. To see that its size is at least $\floor{\onorm{\xx}}$, note that, as there is at most $\ceil{d_P}+\ceil{d_Q}$ dummy vertices in $\oG$, there could be at most that many edges incident to these vertices in $\oxx^*$. But, as $\oxx^*$ is perfect, its size is equal to 
\[
\frac{|P|+\ceil{d_P}+|Q|+\ceil{d_Q}}{2}=\frac{|P|-\ceil{d_P}+|Q|-\ceil{d_Q}}{2} + \ceil{d_P}+\ceil{d_Q} = \floor{\onorm{\xx}} + \ceil{d_P}+\ceil{d_Q},
\]
where we used the fact that $|P|-d_P=|Q|-d_Q=\onorm{\xx}$. Thus, indeed after removing at most $\ceil{d_P}+\ceil{d_Q}$ edges from $\oxx^*$, the resulting integral matching $\xx^*$ will have its size $\onorm{\xx^*}$ to be at least $\floor{\onorm{\xx}}$, as desired.

\subsubsection*{Rounding $\bb$-Matchings}

In the light of the above, it remains to show how to deal with the case when in the demand vector $\bb$ there are some $b_v$ that are bigger than $1$ (and thus some of the entries of $\xx$ could be bigger than $1$, as well). To this end, let us observe first that if there is an edge $e=(p,q)$ with $x_e\geq 1$, we can just subtract $\floor{x_e}$ copies of this edge from our matching right away, while decreasing the demands $b_p$ and $b_q$ of $e$'s endpoints accordingly, i.e., by $\floor{x_e}$. (Note that by feasibility of $\xx$, $b_p,b_q\geq \floor{x_e}$.) So, one can see that if $\oxx$ is the fractional matching $\xx$ after we made such transformation and $\obb$ are the corresponding demands, then once we compute an integral $\obb$-matching $\oxx^*$ of size at least $\floor{\onorm{\oxx}}$ from $\oxx$, we can just add back these subtracted $\floor{x_e}$ copies of edge $e$ to $\oxx^*$ to obtain the desired integral $\bb$-matching $\xx^*$ of size at least $\floor{\onorm{\oxx}}+\floor{x_e}=\floor{\onorm{\xx}}$. 

Therefore, we can assume from now on that in our $\bb$-marching $\xx$ all $x_e$s are smaller than one (but still we can have some demands $b_v$ to be bigger than one). To round such fractional $\bb$-matchings, for each vertex $v\in V$ that has its demand $b_v$ bigger than $1$, we split it into $b_v$ vertices $v^{1},\ldots, v^{b_v}$ -- each with demand one. Next, for every edge $e$ that was previously incident to $v$, we connect it to the new vertices and distribute its fractional weight $x_e$ in $\xx$ among these new vertices. Again, by applying a simple greedy approach we can ensure that each edge is connected to at most two among the vertices $v^{1},\ldots, v^{b_v}$ and none of these vertices has its fractional degree bigger than $1$. (Note that this means, in particular, that once we apply such splitting to all vertices with $b_v>1$ then the support of the corresponding ``split'' fractional matching is at most by a factor of four larger than the support of $\xx$.) Clearly, at this point, we are again in situation where we just need to round a fractional bipartite matching (with all demands being at most $1$). Thus, we can use our rounding procedure we described above and recover the integral matching we are seeking. This finishes the proof of Theorem \ref{thm:rounding_matchings}.

%% file: files/acknowledgments.tex
\vspace{10pt}

\noindent{\bf Acknowledgments.} We are grateful to Andrew Goldberg, Jonathan Kelner, Lap Chi Lau, Gary Miller, Richard Peng, Seth Pettie, Daniel Spielman, and Shang-Hua Teng for a number of helpful discussions on this topic. We also thank Monika Henziger, Satish Rao, and Jens Vygen for useful feedback on the manuscript. 

%% file: files/appendix.tex
\section{Proof of Lemma \ref{lem:effective_conductance}}\label{app:effective_conductance}
Let $C^*$ be the value of right-hand side of the equality we need to establish, and - for notational convenience - let us denote the energy $\energy{\rr}{\ff^*}$ as $E^*$. So, our goal is to show that $C^*=1/E^*$ and that taking $\tvphi$ attains the minimum $C^*$.

We start by noting that, for any vertex potentials $\vphi$, we have
\begin{equation}\label{eq:conductance_proof_identity}
\sum_{(u,v)\in E} f^*_{(u,v)} (\phi_v-\phi_u)=\sum_{v} \phi_v (\sum_{e\in E^+(v)} f_{e}^* - \sum_{e\in E^-(v)} f_{e}^*)=\sum_v \phi_v \sigma_v = \vsigma^T \vphi,
\end{equation}
where we used the fact that $\ff^*$ is a $\vsigma$-flow (cf. \eqref{eq:conservation_constraints}).

Note that by the above calculations and the definition of $\tvphi$ we have 
\begin{equation}
\label{eq:effective_conductance_2}
\vsigma^T \tvphi = \frac{1}{E^*} \sum_{(u,v)\in E} f^*_{(u,v)} (\phi_v^*-\phi_u^*) = \frac{1}{E^*} \sum_{(u,v)\in E} r_{(u,v)} (f^*_{(u,v)})^2 = 1,
\end{equation}
where we used \eqref{eq:potential_flow_def} and the definition of energy \eqref{eq:def_energy_flow}. Therefore, we see that $C^*\leq 1/E^*$ as by \eqref{eq:def_energy_potentials}
\begin{equation}
\label{eq:effective_conductance_3}
\sum_{e=(u,v)\in E} \frac{(\tphi_v-\tphi_u)^2}{r_{e}}=\frac{1}{(E^*)^2} \sum_{e=(u,v)\in E} \frac{(\phi_v^*-\phi_u^*)^2}{r_{e}} =1/E^*.
\end{equation}

Now, let $\hvphi$ be the potential such that $\sum_{(u,v)\in E} \frac{(\hphi_v-\hphi_u)^2}{r_{(u,v)}}=C^*$ and let $\hff$ with $\hf_{(u,v)} := \frac{\hphi_v-\hphi_u}{r_{(u,v)}}$, for each $(u,v)\in E$, be the corresponding flow induced via \ref{eq:potential_flow_def}. (Note that in principle  $\hff$ does not need to be a $\vsigma$-flow).

From \eqref{eq:conductance_proof_identity} we get that 
\begin{equation}
\label{eq:effective_conductance_1}
(\ff^*)^T \RR \hff = \sum_{e} r_e f^*_e \hf_e = \sum_{(u,v)\in E} f^*_{(u,v)} (\hphi_v-\hphi_u) = \vsigma^T\hvphi= 1,
\end{equation}
where we again used \eqref{eq:potential_flow_def} and the fact that $\vsigma^T\hvphi=1$ by definition.

We claim that the energy $\energy{\rr}{\hff}$ of $\hff$ (and thus the value of $C^*$) is at least $1/E^*$. To this end, let us note that
\begin{eqnarray*}
\energy{\rr}{\hff} &= &\hff^T \RR \hff = \left( \frac{\ff^*}{E^*} + \hff - \frac{\ff^*}{E^*} \right)^T\RR\left( \frac{\ff^*}{E^*} + \hff - \frac{\ff^*}{E^*} \right)\\
&=&\frac{(\ff^*)^T \RR \ff^*}{(E^*)^2} - 2 \frac{(\ff^*)^T}{E^*}\RR\left(\hff - \frac{\ff^*}{E^*} \right) + \left( \hff - \frac{\ff^*}{E^*} \right)^T\RR\left(\hff - \frac{\ff^*}{E^*} \right).
\end{eqnarray*}
As we have seen in \eqref{eq:effective_conductance_1}, $(\ff^*)^T \RR \hff=1$, thus $\frac{(\ff^*)^T}{E^*}\RR\left(\hff - \frac{\ff^*}{E^*} \right)=0$ and we can write
\begin{eqnarray*}
\energy{\rr}{\hff} &=&\frac{(\ff^*)^T\RR \ff^*}{(E^*)^2} - 2 \frac{(\ff^*)^T}{E^*}\RR\left(\hff - \frac{\ff^*}{E^*} \right) + \left( \hff - \frac{\ff^*}{E^*} \right)^T\RR\left(\hff - \frac{\ff^*}{E^*} \right)\\
&=&\frac{1}{E^*} + \left( \hff - \frac{\ff^*}{E^*} \right)^T\RR\left(\hff - \frac{\ff^*}{E^*} \right)\geq \frac{1}{E^*},
\end{eqnarray*}
as $\ff^T\RR \ff\geq 0$ for any $\ff$. 

So, $C^*\geq 1/E^*$ too and thus $C^*=1/E^*$. Also, by \eqref{eq:effective_conductance_2} and \eqref{eq:effective_conductance_3} we see that $\tvphi$ indeed attains the minimum, as desired.

\section{Proof of Corollary \ref{col:rounding_flows}}\label{app:col_rounding_flows}

Let $\ff$ be a fractional feasible $s$-$t$ flow of value $F$ in $G$ and let us consider first the case when $F$ is integral. Recall that the reduction presented in Section \ref{sec:reduction} allows one to obtain in $\tO{m}$ time an instance of bipartite $\bb$-matching problem -- corresponding to some bipartite graph $\oG$ -- that has a property that if there exists a feasible $s$-$t$ flow of value $F$ in $G$ then $\oG$ has a perfect $\bb$-matching. Now, the crucial observation is that the proof of that property presented in Section \ref{sec:reduction} is fully constructive and, in particular, provides an $\tO{m}$-time algorithm that produces such a perfect $\bb$-matching in $\oG$ out of a feasible $s$-$t$ flow in $G$ of value $F$. Furthermore, this construction also works for fractional flows, it just produces a perfect $\bb$-matching that is fractional. 

In the light of the above, we can simply apply this transformation to our flow $\ff$ and get a fractional perfect $\bb$-matching $\xx$ in $\oG$. Next, we can use the rounding procedure from Theorem \ref{thm:rounding_matchings} to obtain in $\tO{m}$ time a perfect $\bb$-matching $\xx^*$ in $\oG$ that is integral. (Note that since $\bb$ is always integral, so is the size of any perfect $\bb$-matching.) This, in turn, allows us to utilize another property of the graph $\oG$ that was established in Section \ref{sec:reduction}. Namely, that out of any integral perfect $\bb$-matching in $\oG$, one can extract -- in $\tO{m}$ time -- an integral and feasible $s$-$t$ flow $\ff^*$ in $G$ of value $F$. Clearly, by combining all of the above steps, we get our desired integral $s$-$t$ flow. 

Finally, to deal with the case when $F$ is not integral, we just add an arc $(s,t)$ to $G$, set its capacity to $1$, and put a flow of $\ceil{F}-F\leq 1$ on it. Obviously, now we have a feasible $s$-$t$ flow of value $\ceil{F}$ in such modified graph $G$ and $\ceil{F}$ is integral. Therefore, we can use our approach we described above to get an integral and feasible $s$-$t$ flow $\ff^*$ in this graph and $\ff^*$ will have a value of $\ceil{F}$. Note that $\ff^*$ can have non-zero flow on the arc $(s,t)$ that we added, but as this arc has capacity of $1$, there can be exactly one unit of flow on this arc. So, if we simply remove it from $\ff^*$, we will get an integral and feasible $s$-$t$ flow in the original graph $G$ and the value of $\ff^*$ will be $\ceil{F}-1=\floor{F}$, as desired. This concludes the proof of the corollary.

%% file: files/app_reduction.tex
\section{Appendix to Section \ref{sec:reduction}}\label{app:reduction}

\subsection{Correctness Analysis}

It is easy to verify that the produced $\bb$-matching instance is indeed bipartite (we have edges only between different sides of bipartition $P$ and $Q$), has exactly $2(m+n-1)=\Theta(m)$ vertices, $3m+n-2\leq 4m$ edges, and $\onorm{\bb}\leq 4\onorm{\uu}$. So, we just need to establish the claimed connection to existence of feasible $s$-$t$ flows in the graph $G$.

\subsubsection*{From Flow $\ff$ to Perfect $\bb$-matching $\xx$}

To this end, assume that there exists a feasible $s$-$t$ flow $\ff$ in $G$ of value $F$. To see that a perfect $\bb$-matching in $\oG$ exists, consider a $\bb$-matching $\xx$ that, for each arc $e=(u,v)$ in $G$, takes exactly $f_e$ edges $(p_e,q_e)$ and $u_e-f_e$ edges $(q_u,p_e)$ and $(q_e,p_v)$. Then, for every vertex $v$ of $G$ other than $s$ and $t$, $\xx$ takes $\sum_{e\in E^+(v)} f_e$ copies of the edge $(p_v,q_v)$. 

To see that $\xx$ is indeed a perfect $\bb$-matching, observe that due to feasibility of $\ff$ (cf. \eqref{eq:capacity_constraints}),  $0\leq f_e \leq u_e$ for each arc $e$, and thus $\xx\geq 0$. Also, by the construction of $\xx$, all vertices $p_e$ and $q_e$ have exactly $u_e$ edges adjacent to them in $\xx$. So, they are fully matched. To see that all vertices $p_v$ and $q_v$ are fully matched too, consider some $v\neq s,t$. Indeed, by definition of $\xx$, we have exactly $\sum_{e\in E^+(v)} u_e-f_e+\sum_{e\in E^+(v)} f_e=\sum_{e\in E^+(v)} u_e=b_{p_v}$ (resp. $\sum_{e\in E^-(v)} u_e-f_e+\sum_{e\in E^-(v)} f_e=\sum_{e\in E^-(v)} u_e=b_{q_v}$) edges adjacent to $p_v$ (resp. $q_v$), where we used the fact that $\sum_{e\in E^+(v)} f_e=\sum_{e\in E^-(v)} f_e$, as $\ff$ obeys flow conservation constraints \eqref{eq:conservation_constraints}. Finally, in the case of vertex $q_s$ (resp. $p_t$) we have that their degree in $\xx$ is exactly $\sum_{e\in E^-(s)} u_e-f_e = (\sum_{e\in E^-(s)} u_e) - F = b_{q_s}$ (resp. $\sum_{e\in E^+(t)} u_e-f_e = (\sum_{e\in E^+(t)} u_e)- F = b_{p_t}$), due to the value $\sum_{e\in E^-(s)} f_e = \sum_{e\in E^+(t)} f_e$ of the flow $\ff$ being exactly $F$. So, indeed such $\xx$ is perfect, as claimed.

\subsubsection*{From Perfect $\bb$-matching $\xx$ to Flow $\ff$}

Now, to see that given a perfect $\bb$-matching $\xx$ in $\oG$ we can quickly, i.e., in $\tO{m}$ time, recover an $s$-$t$ flow of value $F$ that is feasible in $G$, consider a flow $\ff$ given by $f_e=x_{(p_e,q_e)}$ for each arc $e$ in $G$. That is, the flow $f_e$ on an arc $e$ is equal to the number of times a copy of an edge $(p_e,q_e)$ appears in $\xx$. Note that as the demands $b_{p_e}$ and $b_{q_e}$ of the endpoints of each edge $(p_e,q_e)$ are equal to $u_e$, $\ff$ is feasible in $G$.

Finally, to prove that $\ff$ also preserves flow conservation constraints (cf. \eqref{eq:conservation_constraints}), note that as $\xx$ is perfect, it has to be that for any vertex $v$ and $e\in E^+(v)$ (resp. $e\in E^-(v)$) $x_{(p_v,q_e)}=b_{q_e}-x_{(p_e,q_e)} = u_e - f_e$ (resp. $x_{(p_e,q_v)}=b_{p_e}-x_{(p_e,q_e)} = u_e - f_e$). So, if we do not take into account the edges $(p_v,q_v)$, each vertex $p_v$ (resp. $q_v$) has exactly $\sum_{e\in E^+(v)} b_{q_e}-x_{(p_e,q_e)} = \sum_{e\in E^+(v)} u_e - f_e$ (resp. $\sum_{e\in E^{-}(v)} b_{p_e}-x_{(p_e,q_e)} = \sum_{e\in E^{-}(v)} u_e - f_e$) edges adjacent to it in $\xx$. This means, in particular, that in case of $q_s$ (resp. $p_t$) we need to have that $\sum_{e\in E^{-}(s)} u_e-f_e = b_{q_s} = (\sum_{e\in E^{-}(s)} u_e) - F$ (resp. $\sum_{e\in E^{+}(t)} u_e-f_e = b_{p_t} = (\sum_{e\in E^{+}(t)} u_e) - F$) and thus $\sum_{e\in E^{-}(s)}=F$ (resp. $\sum_{e\in E^{+}(t)} f_e=F$), i.e., the value of $\ff$ is $F$. Furthermore, for any vertex $v$ other than $s$ and $t$, as $\xx$ is perfect,  w need to have that $\sum_{e\in E^+(v)} u_e = b_{p_v}=x_{(p_v,q_v)}+\sum_{e\in E^+(v)} u_e - f_e$ (resp. $\sum_{e\in E^-(v)} u_e=b_{q_v}=x_{(p_v,q_v)}+\sum_{e\in E^-(v)} u_e - f_e$). Therefore, $\sum_{e\in E^+(v)} f_e = x_{(p_v,q_v)} = \sum_{e\in E^-(v)} f_e$, i.e., $\ff$ obeys all flow conservation constraints. So, indeed $\ff$ is a feasible $s$-$t$ flow of value $F$ in $G$, as desired.

Lastly, it is worth pointing out that even though in the above proof we assume that both the $s$-$t$ flow $\ff$ and the $\bb$-matching $\xx$ are integral, the proof goes through unchanged in the case when $\ff$ and $\xx$ are fractional. We just will have that if $\ff$ is fractional then so will be the corresponding $\bb$-matching $\xx$ and vice versa.

%% file: files/app_basic_algorithm.tex
\section{Appendix to Section \ref{sec:simple}}\label{app:basic_algorithm}

\subsection{Proof of Lemma \ref{lem:initial_solution}}\label{app:initial_solution}

Let us take $\ss^0$ to be the all-ones vector $\onev$ (this corresponds to $\yy^0$ assigning zero value to all vertices). Next, let the flow $\ff^0$ and measures $\vnu^0$ be defined as follows. 

For each arc of the form $(s_p,t_q)$ in $\hG$, we give it a measure of one in $\vnu^0$ and a flow of one unit is sent through it in $\ff^0$. Now, for each vertex $s_p$ (resp. $t_q$) in $\hG$, let $r_p:=|\hE^-(s_p)|-1-b_p$ (resp. $r_q:=|\hE^+(t_q)|-1-b_q$). If $r_p\geq 0$ (resp. $r_q\geq 0$) then we put a flow of one and measure of one on the arc $(s_p,v^*)$ (resp. $(v^*,t_q)$) and a flow and measure of $r_p+1$ (resp. $r_q+1$)  on the arc $(v^*,s_p)$ (resp. $(t_q,v^*)$). On the other hand, if $r_p<0$ (resp. $r_q<0$) then we put a flow and measure of $1-r_p$ (resp. $1-r_q$) on the arc $(s_p,v^*)$ (resp. $(v^*,t_q)$) and a flow and measure of one on the arc $(v^*,s_p)$ (resp. $(t_q,v^*)$).

One can verify that the resulting flow $\ff^0$ is indeed a $\hvsigma$-flow (again, one needs to use here the fact that $\sum_p b_p=\sum_q b_q$, as otherwise there is no perfect $\bb$-matching in $G$) and thus the solution is primal-dual feasible. 

Also, the total measure $\sum_e \nu^0_e$ of all the arcs is at most $\hm+\onorm{\bb}\leq 3\hm$. Finally, we have that $\hmu_e^0=\frac{f_e^0s_e^0}{\nu_e^0}=1$ for all arcs $e$ and thus the solution is indeed $0$-centered and $\hmu(\ff^0,\ss^0,\vnu^0)=1$, as desired.

%% file: files/app_improved_algorithm.tex
\section{Appendix to Section \ref{sec:improved}}\label{app:improved_algorithm}

\subsection{Proof of Lemma \ref{lem:choosing_beta}}\label{app:choosing_beta}

That $(1-\gamma)\alpha \leq \beta \leq (1+\gamma)\alpha$ follows directly from Fact \ref{fa:central_vs_max_min}.

Let us show now that $\hmu(\ff',\ss',\vnu')= \hmu(\ff,\ss,\vnu)$. To this end, let us define $\gamma_e$ so as 
\[
(1+\gamma_e):= \frac{\beta}{\alpha} = \frac{f_e s_e}{\nu_e\hmu(\ff,\ss,\vnu)}.
\]

By definition \eqref{eq:def_hmu}, we have that
\[
\hmu(\ff',\ss',\vnu') = \frac{\sum_g f_g' s_g'}{\sum_g \nu_g'} =  \frac{(\sum_{g\neq e} f_g s_g) + (1+\alpha) f_e s_e }{(\sum_{g\neq e} \nu_g) + (1+\beta) \nu_e}=\hmu(\ff,\ss,\vnu)\left(\frac{(\sum_{g} \nu_g) + \alpha \nu_e (1+\gamma_e) }{(\sum_{g} \nu_g) + \alpha \nu_e (1+\gamma_e)}\right)=\hmu(\ff,\ss,\vnu)
\]

Now, to bound the centrality of $(\ff',\ss',\vnu')$, we need to bound the value of $\norm{\hvmu'-\hmu(\ff',\ss',\vnu')\onev}{\vnu',2}$. However, as $\hmu(\ff',\ss',\vnu')= \hmu(\ff,\ss,\vnu)$ and the two solution coincide on all the arcs except $e$, it suffices to analyze the change in contribution of arc $e$ to the centrality of the solution. Namely, we just need to show that 
\[
\nu_e' \left(\frac{f_e's_e'}{\nu_e'} - \hmu(\ff,\ss,\vnu)\right)^2 = (1+\beta)\nu_e \left(\frac{(1+\alpha)f_es_e}{(1+\beta)\nu_e} - \hmu(\ff,\ss,\vnu)\right)^2\leq \nu_e (\frac{f_es_e}{\nu_e} - \hmu(\ff,\ss,\vnu))^2.
\]

To this end, observe the right side of the above inequality is just
\[
\nu_e \left(\frac{f_es_e}{\nu_e} - \hmu(\ff,\ss,\vnu)\right)^2=\nu_e \gamma_e^2\hmu(\ff,\ss,\vnu)^2.
\]

So, we need to show that 
\[
\frac{\nu_e'}{\nu_e\hmu(\ff,\ss,\vnu)^2} \left(\frac{f_e's_e'}{\nu_e'} - \hmu(\ff,\ss,\vnu)\right)^2 =  (1+\beta)\left(\frac{(1+\alpha)(1+\gamma_e)}{(1+\beta)} - 1\right)^2\leq \gamma_e^2.
\]

But this is true as 
\[
(1+\beta)\left(\frac{(1+\alpha)(1+\gamma_e)}{(1+\beta)} - 1\right)^2=\frac{\left((1+\alpha)(1+\gamma_e)-1-(1+\gamma_e)\alpha\right)^2}{1+(1+\gamma_e)\alpha}=\frac{\gamma_e^2}{1+(1+\gamma_e)\alpha}\leq \gamma_e^2,
\]
where we used that fact that $(1+\gamma_e)\geq \frac{1}{2}$ since, by Fact \ref{fa:central_vs_max_min}, $(1+\gamma_e)\geq (1-\gamma)\geq \frac{1}{2}$. 

\subsection{Proof of Lemma \ref{lem:bound_on_s}}\label{app:bound_on_s}

Note that by construction of the graph $\hG$ and by Invariant \ref{inv:length_upper}, we have that for any two vertices $v$, $v'$ in $\hG$ there is a directed path from $v$ to $v'$, as well as, a one from $v'$ and $v$, with each one of them consisting of at most two arcs and having length at most $4$. As $(\ff,\ss,\vnu)$ is $\vsigma$-feasible then it is, in particular, dual feasible. So, this implies that if $\yy$ is the embedding of the vertices of $\hG$ into a line corresponding to the slack variables $\ss$, then $|y_{v}-y_{v''}|$ is at most $4$ as well. Thus, we can conclude that for any arc $e=(v,v')$ in $\hG$, we have that $s_e = \hl_e - y_{v'} + y_{v}\leq \hl_e + |y_{v}-y_{v''}|$ is at most $6$, as desired.

\subsection{Proof of Lemma \ref{lem:perturbed_solution}}\label{app:perturbed_solution}

Note first that all the arcs in the original version of $\hG$ have length $1$ and $\alpha$-stretching can only increase these lengths. This means that it is still true -- as in the proof of Lemma \ref{lem:mincost_to_matchings} -- that $\hl(\ff)-\frac{\onorm{\bb}}{2}$ is an upper bound on the total flow between the vertices $s_p$ and $t_q$ that is not flowing over the direct arcs $(s_p,q_t)$ reflecting the original edges of $G$. 

Furthermore, as by Invariant \ref{inv:length_upper} the total increase in the length of the arcs of $\hG$ is $\tO{\hm^{\frac{1}{2}-\eta}}$, the cost $\hll(\ff^*)$ of the flow that encodes the perfect $\bb$-matching in $G$ (cf. the proof of Lemma \ref{lem:mincost_to_matchings}), can increase to at most $\frac{\onorm{\bb}}{2}+\tO{\hm^{\frac{1}{2}-\eta}}$. (We use here the fact that $f^*$ never flows more than one unit of flow through any of the arcs.) So, we can assume that the cost $\hll(\ff)$ of the flow $\ff$ we have is also $\frac{\onorm{\bb}}{2}+\tO{\hm^{\frac{1}{2}-\eta}}$, as otherwise we could conclude that no perfect $\bb$-matching exists in $G$. 

However, then it must be the case that the total flow in $\ff$ that does not correspond to taking the direct arcs is at most $\tO{\hm^{\frac{1}{2}-\eta}}$. Thus, the fractional $\bb$-matching obtained by taking only the flow that uses these direct flow-paths will still result in a near-perfect $\bb$-matching in $G$.

\subsection{Proof of Lemma \ref{lem:energy_bound}}\label{app:energy_bound}

Note first that the upperbound follows directly from Lemma \ref{lem:hfft_energy_bound} as long as we ensure that $\cenergy >4$ (which indeed will be the case). 

To establish the lowerbound, let us note first that without loss of generality we can assume that in our $\bb$-matching instance $G=(P\cup Q,E)$, $|P|\geq |Q|$ and thus, as our graph $G$ is sparse, $|P|=\Omega(\hm)$. (Otherwise, we just exchange the roles of $P$ and $Q$ in what follows below.)  

Let us define a vector of resistances $\trr^t$ given by
\[
\tr_{e}^t:=\begin{cases}
r_e^t & \mbox{if $e\in \hE(s_p)$, for some $p\in P$}\\
0 & \mbox{otherwise},
\end{cases}
\]
where $\hE(s_p):=\hE^+(s_p)\cup \hE^-(s_p)$ is the set of all arcs incident to $s_p$ in $\hG$. In other words, $\trr^t$ corresponds to setting to zero resistances (i.e., collapsing) of all the arcs that are not adjacent to some vertex in $P$; and making the resistances of arcs that are adjacent to such $s_p$ equal to their original resistances in $\rr^t$. This means, in particular, that $\tr_e^t\leq r_e^t$, for each arc $e$, and thus, by Rayleigh Monotonicity principle (cf. Fact \ref{fa:rayleigh_monotonicity}), we know that if $\tff^t$ is the electrical $\hvsigma$-flow determined by resistances $\trr^t$ then
\[
\energy{\trr^t}{\tff^t} \leq \energy{\rr^t}{\hff^t}.
\]

Therefore, we can just focus on lowerbounding $\energy{\trr^t}{\tff^t}$. To this end, note that after collapsing all the arcs that were not adjacent to some vertex in $P$, we can think of $\hG$ as a graph that consists only of vertices from $P$ and a single vertex $w^*$ that represents the remaining collapsed vertices. Furthermore, as there is no arcs in $\hG$ between different vertices $s_p$, all the arcs in this collapsed graph are of the form $(s_p,w^*)$ or $(w^*,s_p)$ for some $p\in P$.

As a result, $\energy{\trr^t}{\tff^t}$ is equal to
\[
\energy{\trr^t}{\tff^t}=\sum_{p\in P} \sum_{e\in \hE(s_p)} r_e^t (\tf_e^t)^2 = \sum_{p\in P} R_p \sigma_{s_p}^2,
\]
where $R_p$ is the effective resistance between vertex $s_p$ and $w^*$ with respect to resistances $\trr^t$ and the last equality follows as $\tff^t$ is a $\hvsigma$-flow and all arcs are connecting to $w^*$. 

Now, to lowerbound $R_p$, for some $p\in P$, note that by definition of $\rr^t$ \eqref{eq:hf_resistances}, the fact that $\nu_e^t\geq 1$ for all $e$, and Fact \ref{fa:central_vs_max_min}, we have that
\[
\frac{1}{R_p}=\sum_{e\in \hE(s_p)} \frac{1}{r_e^t} = \sum_{e\in \hE(s_p)} \frac{(f_e^t)^2}{\mu_e^t} \leq \sum_{e\in \hE(s_p)} \frac{(f_e^t)^2}{(1-\hgamma)\nu_e^t\hmu(\ff^t,\ss^t,\vnu^t)}\leq \frac{F_p^2}{(1-\hgamma)\hmu(\ff^t,\ss^t,\vnu^t)},
\]
where $F_p:=\sum_{e\in \hE(s_p)} f_e^t$ and we used the well-known formula for effective resistance of a circuit that consists solely of parallel arcs. 

So, all the above considerations allow us to observe that
\[
\energy{\trr^t}{\tff^t}= \sum_{p\in P} R_p \hsigma_{s_p}^2 \geq \sum_{p\in P} \frac{(1-\hgamma)\hmu(\ff^t,\ss^t,\vnu^t)}{F_p^2} \geq  (1-\hgamma)\hmu(\ff^t,\ss^t,\vnu^t) \frac{|P|^3}{F^2}, 
\]
where $F=\sum_p F_p$ and we used the fact that $|\hsigma_{s_p}|\geq b_p\geq 1$, as well as, that for any $n$-dimensional vector $\xx$, $\sum_{i=1}^n \frac{1}{x_i^2}\geq \frac{n^3}{\onorm{\xx}^2}$. 

Thus, it remains to provide an upperbound on $F$. To this end, let us decompose the flow $\ff^t$ into flow-paths (whose endpoints are vertices $s_p$ and $t_q$) and flow-cycles. Clearly, the total contribution of the flow-paths to $F$ can be at most $\onorm{\hvsigma}\leq \onorm{\bb}=O(\hm)$, since our $\bb$-matching instance is balanced. On the other hand, as length of any flow-cycle is at least two and each flow-cycle contributes its whole volume to the duality gap $(\ff^t)^T\ss^t=\sum_e \mu_e^t$ (as flow-cycles do not exist in optimal solution), the total contribution of flow-cycle to $F$ is at most $\frac{1}{2}\sum_e \mu_e^t$. Thus, by \eqref{eq:duality_gap_bound} and Invariant \ref{inv:measure_upperbound}, this contribution is at most $\frac{1}{2}\sum_e \mu_e^t=\frac{1}{2}\hmu(\ff^t,\ss^t,\vnu^t)(\sum_e \nu_e^t) \leq  2\hm \hmu(\ff^t,\ss^t,\vnu^t)\leq 2\hm$. 

Therefore, we can conclude that $F= O(\hm)$ and since $|P|$ is $\Omega(\hm)$ we have
\[
\energy{\trr^t}{\tff^t} \geq  (1-\hgamma)\hmu(\ff^t,\ss^t,\vnu^t) \frac{|P|^3}{F^2} \geq \cenergy^{-1} \hm \hmu(\ff^t,\ss^t,\vnu^t),
\]
where $\cenergy>4$ is an appropriately chosen constant.

\subsection{Proof of Lemma \ref{lem:stretch_boosting_energy_increase}}\label{app:stretch_boosting_energy_increase}

Let us denote by $\rr$, $\vnu$, and $\hff$, respectively, the resistances $\rr^t$, measures $\vnu^t$, and electrical $\hvsigma$-flow $\hff^t$ before stretch-boost and by $\rr'$, $\vnu'$, and $\hff'$ the corresponding objects after stretch-boost. Also, let as define $S^*:=\Cset{l^*}{\hff^t}\cap E_{H}^t$, where $l^*$ is the index of the stretch-boost. In this notation, we want to show that  
\begin{equation}
\label{eq:bound_increase2}
\energy{\rr'}{\hff'}\geq \left(1+\frac{\vnu(S^*)}{36\cdot 2^{2l^*}}\right) \energy{\rr}{\hff} \geq \left(1+\frac{\htheta^2(\vnu(S^*))^{\frac{1}{3}}}{36}\right) \energy{\rr}{\hff} \geq \left(1+\frac{\htheta^2}{36}\right) \energy{\rr}{\hff}.
\end{equation}
Clearly, the last inequality follows as $\nu_e\geq 1$ for all arcs $e$. To see that the second inequality holds, observe that by \eqref{eq:boosting_condition} we have that
\[
\vnu(S^*)\geq \htheta^3 2^{3l^*}
\]
and thus
\[
\frac{\vnu(S^*)}{36\cdot 2^{2l^*}} \geq \frac{(\vnu(S^*))^{\frac{1}{3}}(\htheta^32^{3l^*})^{\frac{2}{3}}}{36\cdot 2^{2l^*}} = \frac{\htheta^2(\vnu(S^*))^{\frac{1}{3}}}{36}.
\]
 Therefore, once the first inequality in \eqref{eq:bound_increase2} is established, our lemma will follow by choosing $\cincrease:=\frac{1}{36}$.

So, let us proceed to establishing that inequality. By Lemma \ref{lem:effective_conductance}, we know that if $\vphi^*$ is the vector of vertex potentials corresponding to the flow $\hff$ and resistances $\rr$ then 
\begin{equation}
\label{eq:contr_hfft}
\frac{1}{\energy{\rr}{\hff}} = \sum_{e=(u,v)\in \hE} \frac{(\tphi_v-\tphi_u)^2}{r_e},
\end{equation}
where $\tphi_v:=\phi_v^*/\energy{\rr}{\hff}$, for each $v$, and $\hvsigma^T\tvphi=1$. 

Now, consider an arc $e\in S^*$. By \eqref{eq:potential_flow_def}, the definition of the sets $\Cset{l^*}{\hff^t}$ \eqref{eq:def_of_C}, and Fact \ref{fa:rho_vs_rt}, we have that
\begin{eqnarray*}
\frac{(\tphi_v-\tphi_u)^2}{r_e} &=& \frac{1}{\energy{\rr}{\hff}^2} r_e \hf^{2}_e \geq \frac{(1-\hgamma)\nu_e \hmu(\ff,\ss,\vnu) \rho(\ff,\hff)^2}{\energy{\rr}{\hff}^2}\\
&\geq & \left( \frac{(1-\hgamma)\nu_e \hm \hmu(\ff,\ss,\vnu)}{4\cdot 2^{2(l^*+1)} \hm \hmu(\ff,\ss,\vnu)}\right)\frac{1}{\energy{\rr}{\hff}}\geq \left( \frac{\nu_e}{18\cdot 2^{2l^*}}\right)\frac{1}{\energy{\rr}{\hff}},
\end{eqnarray*}
where we also used Lemma \ref{lem:hfft_energy_bound}. In other words, the contribution of arc $e$ to the sum in \eqref{eq:contr_hfft} constitutes at least $\frac{\nu_e}{18\cdot 2^{2l^*}}$-fraction of this sum.

Next, observe that, by definition of $1$-stretching, we need to have that the resistance doubles, i.e., $r_e'=2r_e$, for all the arcs $e\in S^*$, and remains the same for other arcs, i.e., $r_e'=r_e$, for $e\notin S^*$.  
This means that
\begin{eqnarray*}
\sum_{e=(u,v)\in \hE} \frac{(\tphi_v-\tphi_u)^2}{r_e'} &=& \frac{1}{2} \sum_{e=(u,v)\in S^*} \frac{(\tphi_v-\tphi_u)^2}{r_e} + \sum_{e=(u,v)\in \hE\setminus S^*} \frac{(\tphi_v-\tphi_u)^2}{r_e'}\\
& = & \frac{1}{\energy{\rr}{\hff}} - \frac{1}{2} \sum_{e=(u,v)\in S^*} \frac{(\tphi_v-\tphi_u)^2}{r_e} \leq \frac{1}{\energy{\rr}{\hff}} \left(1-\frac{\vnu(S^*)}{36\cdot 2^{2l^*}}\right).
\end{eqnarray*}

But, by Lemma \ref{lem:effective_conductance}, we know that the above estimation provides an upper bound on the value of $\frac{1}{\energy{\rr'}{\hff'}}$, i.e., we have
\[
\frac{1}{\energy{\rr'}{\hff'}} = \min_{\vphi| \hvsigma^T \vphi=1} \sum_{e=(u,v)\in \hE} \frac{(\phi_v-\phi_u)^2}{r_{e}'} \leq \sum_{e=(u,v)\in \hE} \frac{(\tphi_v-\tphi_u)^2}{r_e'} \leq \frac{1}{\energy{\rr}{\hff}} \left(1-\frac{\vnu(S^*)}{36\cdot 2^{2l^*}}\right),
\]
where we used the fact that $\hvsigma^T \tvphi=1$, by definition of $\tvphi$. Multiplying both sides by $\frac{\energy{\rr'}{\hff'}\energy{\rr}{\hff}}{\left(1-\frac{\vnu(S^*)}{36\cdot 2^{l^*}}\right)}$ and noticing that $\frac{1}{(1-x)}\geq (1+x)$ for any $x\geq 0$, gives us the desired inequality in \eqref{eq:bound_increase2}. 

\subsection{Proof of Lemma \ref{lem:int_step_energy_decrease}}\label{app:int_step_energy_decrease}

Let us denote the solution $(\ff^t,\ss^t,\vnu^t)$ by $(\ff,\ss,\vnu)$, the associated electrical flow $\hff^t$ by $\hff$, and let $\rr$ be the resistances $\rr^t$ corresponding to this solution (cf. \eqref{eq:hf_resistances}). Also, let $(\ff',\ss',\vnu')$, $\hff'$, and $\rr'$, denote these respective object after the interior-point method step is applied. 

In this notation, our goal is to show that
\[
\energy{\rr'}{\hff'}\geq (1+\cendecrease\htheta^2 \ln \hm)^{-1}\energy{\rr}{\hff}.
\]

To perform such lowerbounding of the energy decrease, we proceed similarly as we did in the proof of Lemma \ref{lem:stretch_boosting_energy_increase}. Namely, by Lemma \ref{lem:effective_conductance}, we know that 
\[
\frac{1}{\energy{\rr}{\hff}} = \sum_{e=(u,v)\in \hE} \frac{(\tphi_v-\tphi_u)^2}{r_e},
\]
where $\tphi_e:=\phi_e^*/\energy{\rr}{\hff}$ and $\vsigma^T\tvphi=1$. We want to show that if we keep the same vertex potentials $\tvphi$ and change the resistances to $\rr'$ then still the corresponding sum -- as in the equation above -- will not increase by too much (and thus provide a good upperbound on $\frac{1}{\energy{\rr'}{\hff'}}$). 

More specifically, recall that by Theorem \ref{thm:main_interior_point}, for any arc $e$,
\[
\frac{r_e'}{r_e}=\frac{(1+\kappa_e^t)}{(1-\delta^t)},
\]
and that $\inorm{\vkappa^t}\leq \frac{1}{2}$. 

So, by Lemma \ref{lem:effective_conductance}, we have that

\[
\frac{1}{\energy{\rr'}{\hff'}} \leq \sum_{e=(u,v)\in \hE} \frac{(\phi_v-\phi_u)^2}{r_{e}'} \leq \sum_{e=(u,v)\in \hE} \frac{(\tphi_v-\tphi_u)^2}{r_e'} \leq \sum_{e=(u,v)\in \hE}  \frac{(1+2|\kappa_e^t|)(\tphi_v-\tphi_u)^2}{r_e},
\]
where we use the fact that $(1-x)^{-1}\leq (1+2x)$ when $x\leq \frac{1}{2}$ and that by definition of $\tvphi$, $\vsigma^T\tvphi=1$.  

Furthermore, by \eqref{eq:potential_flow_def} and definition of $\tvphi$, we have
\[
\sum_{e=(u,v)\in \hE} (1+2|\kappa_e^t|) \frac{(\tphi_v-\tphi_u)^2}{r_e} = \frac{1}{\energy{\rr}{\hff}} \left(1+ 2 \sum_e |\kappa_e^t| \frac{(\tphi_v-\tphi_u)^2}{r_e}\energy{\rr}{\hff}\right) = \frac{1}{\energy{\rr}{\hff}} \left(1+ 2 \sum_e |\kappa_e^t| \frac{r_e (\hf_e)^2}{\energy{\rr}{\hff}}\right).
\]
So, we again just need to show that 
\[
\sum_e |\kappa_e^t| \frac{r_e (\hf_e)^2}{\energy{\rr}{\hff}} \leq \cendecrease \htheta^2 \ln \hm,
\]
and the lemma will follow.

To establish this last claim, note that for any arc $e$, $\frac{r_e (\hf_e)^2}{\energy{\rr}{\hff}}$ is just the fraction of energy of the flow $\hff$ (with respect to resistances $\rr$) that is contributed by the arc $e$. So, by Fact \ref{fa:rho_vs_rt} and Lemma \ref{lem:energy_bound}, we have that 
\begin{equation}\label{eq:ener_decr_bounding_re}
\frac{r_e (\hf_e)^2}{\energy{\rr}{\hff}} \leq \frac{(1+\hgamma) \nu_e \hmu(\ff,\ss,\vnu) \rho(\ff,\hff)_e^2}{\energy{\rr}{\hff}}\leq \frac{2\cenergy \nu_e}{2^{2k}},
\end{equation}
whenever $e\in \Cset{k}{\hff}$ (cf. \eqref{eq:def_of_C}), for some integer $k$. 

As a result, we can conclude that
\begin{eqnarray*}
\sum_e  |\kappa_e^t|\frac{r_e (\hf_e)^2}{\energy{\rr}{\hff}} &\leq& 2 \sum_{l} \sum_{e\in T_l} \frac{r_e (\hf_e)^2}{2^l \energy{\rr}{\hff}}  \leq 2 \left(\sum_{l\leq \floor{\log \htheta^{-2}}} \sum_{e\in T_l} \frac{r_e (\hf_e)^2}{2^l\energy{\rr}{\hff}}+\sum_{l\geq \ceil{\log \htheta^{-2}}} \sum_{e\in T_l} 2 \htheta^{2} \frac{r_e (\hf_e)^2}{\energy{\rr}{\hff}}\right)\\
&\leq & 2 \left(\sum_{l\leq \floor{\log \htheta^{-2}}} \sum_{e\in T_l} \frac{\nu_e r_e (\hf_e)^2}{2^l\nu_e\energy{\rr}{\hff}}\right) +4 \htheta^{2} \leq 4 \cenergy  \left(\sum_{l\leq \floor{\log \htheta^{-2}}} \sum_{k} \sum_{e\in T_l\cap \Cset{k}{\hff}} \frac{\nu_e}{2^l2^{2k}}\right) + 4\htheta^2\\
 & =& 4\cenergy \left(\sum_{l\leq \floor{\log \htheta^{-2}}} \sum_{k} \frac{\vnu(T_l\cap \Cset{k}{\hff})}{2^{2k+l}}\right) + 4 \htheta^2,
\end{eqnarray*}
where $T_l$ denotes $T_l^{\vkappa^t}$ (cf. \eqref{eq:def_t_lambda}), \eqref{eq:ener_decr_bounding_re}, and the fact that 
\[
\sum_{l\geq \ceil{\log \htheta^{-2}}} \sum_{e\in T_l} \frac{r_e (\hf_e)^2}{\energy{\rr}{\hff}}\leq \sum_e \frac{r_e (\hf_e)^2}{\energy{\rr}{\hff}} = \frac{\energy{\rr}{\hff}}{\energy{\rr}{\hff}}=1.
\]
Now, by $2\htheta$-smoothness of $\hff$ (cf. Definition \ref{def:smoothness}), we get that 
\begin{equation}
\label{eq:internal_smoothness_lambda}
\vnu(T_l\cap \Cset{k}{\hff})\leq \floor{\htheta^3 2^{3(k+1)}},
\end{equation}
for each $l$ and $k$. Also, the fact that by Lemma \ref{lem:bounding_smooth} $\vkappa^t$ is $O(1)$-restricted implies that, for any fixed $l$, 
\[
\sum_{k}\vnu(T_l\cap \Cset{k}{\hff})=\vnu(T_l) \leq O(2^{3l}).
\]
Therefore, we can see that for any $l$, we have
\[
\sum_{k} \frac{\vnu(T_l\cap \Cset{k}{\hff})}{2^{2k+l}}\leq \sum_{k=0}^{k'} \frac{\vnu(T_l\cap \Cset{k}{\hff})}{2^{2k+l}},
\]
for some $k'=l+\log \htheta^{-1}+O(1)$. Here, we used the fact that by \eqref{eq:internal_smoothness_lambda}, $\vnu(T_l\cap \Cset{k}{\hff})=\emptyset$, if $k<0$ and that the expression we are bounding will be maximized if the set $T_l$ contains as many high-energy arcs as possible. (Note that due to the constraint $\vnu(T_l)= O(2^{3l})$ and the bound \eqref{eq:internal_smoothness_lambda}, $T_l$ can then only contain all the arcs in sets $\Cset{k}{\hff}$ for all $k\geq 0$ up to $k'$.) So, we can conclude that
\[
 \sum_{k=0}^{\infty} \frac{\vnu(T_l\cap \Cset{k}{\hff})}{2^{k+l}}  \leq \sum_{k=0}^{k'} \frac{\htheta^3 2^{3(k+1)}}{2^{2k+l}} = O(\htheta^3 2^{\frac{k'}{2}-l})=O(\htheta^2).
\]

To finish our overall bound, we just need to note that by our above derivation, as well as, the fact that $T_l=\emptyset$ if $l\leq 0$ (as $\inorm{\vkappa^t}\leq \frac{1}{2}$),
\[
\sum_e  |\kappa_e^t|\frac{r_e (\hf_e)^2}{\energy{\rr}{\hff}}\leq 4\cenergy \left(\sum_{l=1}^{\floor{\log \htheta^{-2}}} \sum_{k} \frac{\vnu(T_l\cap \Cset{k}{\hff})}{2^{2k+l}}\right) + 4 \htheta^2 \leq 4\cenergy\left(O(\log \htheta^{-2}) O(\htheta^2) +\htheta^2\right)=\cendecrease \htheta^2 \ln \hm,
\]
as desired, once $\cendecrease>1$ is chosen to be large enough.

\subsection{Proof of Lemma \ref{lem:auxiliary_flow_growth}}\label{app:auxiliary_flow_growth}

We start by bounding the increase of measure due to freezing. Let us fix some progress step $t$ and some auxiliary arc $e$ that is in $T_{l}$ for some $l\leq \log \htheta^{-2}$, where $T_l$ denotes $T_l^{\ovkappa^t}$, as defined in \eqref{eq:def_t_lambda}. (Note that only arcs in $T_l$ with $l\leq \log \htheta^{-2}$ can be frozen at step $t$.) 

By Lemma \ref{lem:choosing_beta} and definition of $T_l$, the increase of measure resulting from $\hm^{2\eta}|\okappa_e^t|$-stretching $e$ is at most
\[
(1+\hgamma)\hm^{2\eta}|\okappa_e^t| \nu_e^t \leq \frac{2\hm^{2\eta}}{2^l}.
\]
However, by Lemma \ref{lem:bounding_smooth}, we know that the vector $\ovkappa^t$ is $\crestrict$-restricted. Therefore, the total contribution to measure increase of all the frozen arcs in $T_l$ is at most
\[
\frac{2\hm^{2\eta}}{2^l} \crestrict 2^{3l} \leq O(\hm^{2\eta}2^{2l}) = O(\hm^{6\eta}),
\]
where we used the fact that $l\leq \log \htheta^{-2}$. 

So, as there is at most $\log \htheta^{-2}$ different sets $T_l$ that contribute in each progress step, and there is at most $\htheta^{-2}=\hm^{2\eta}$ progress steps, the overall increase of measure due to freezing is at most $\tO{\hm^{8\eta}}$, as required. 

Note that once we establish below that all auxiliary have always $f_e^t$ that is within a factor of $\cfreeze$ of $\fauxiliary$, the fact that $\fauxiliary$ is much smaller than $\fheavy$ will imply that all auxiliary arcs are always light and thus never get stretch-boosted. So, the measure of auxiliary arc can increase only due to freezing and we have already bounded this increase above.

Now, to prove the first part of the lemma, let us fix some auxiliary arc $e$. Initially, $f_e^t$ is equal to $\fauxiliary$. So, one just need to argue that the total multiplicative change of $f_e^t$ during the course of the $\htheta$-improvement phase execution is bounded by a constant. 

To this end, note that the flows on arcs change only during the progress steps. So, by Theorem \ref{thm:main_interior_point}, if we fix some auxiliary arc $e$, its overall flow changes by a factor of at most
\[
\prod_{t=t_0}^{t_f} (1+|\okappa_e^t|).
\]

Therefore, the total change of flow of $e$ during progress steps that have not resulted in freezing it, can bounded by
\[
\prod_{t=t_0}^{t_f} (1+|\okappa_e^t|)\leq (1+\htheta^{2})^{\htheta^{-2}} \leq \exp(1),
\]
that is constant, as desired. 

So, now we just need to focus on bounding the change of flow on $e$ resulting from the remaining progress steps, i.e., the ones in which it was frozen. To this end, recall that whenever $|\okappa_e^t|\geq \htheta^{-2}$ in some step $t$ then freezing $\hm^{2\eta}|\okappa_e^t|$-stretches $e$. By Lemma \ref{lem:choosing_beta}, the resulting increase of measure of $e$ is at least by a factor of 
\[
\left(1+ (1-\hgamma) \hm^{2\eta}|\okappa_e^t|\right),
\] 
while the change of the flow is by a factor of at most
\[
\left(1+|\okappa_e^t|\right).
\]

Therefore, as the former factor is significantly larger than the latter one, $\inorm{\ovkappa^t}\leq \frac{1}{2}$ (by Theorem \ref{thm:main_interior_point}), and as from discussion above we know that once the measure of an arc becomes larger than $\crestrict\htheta^{-6}$ it will never be frozen again, the constant bound of the maximum multiplicative change of the flow of an auxiliary arc follows. This concludes the proof of the lemma.

\subsection{Proof of Lemma \ref{lem:fixing_duality_increase}}\label{app:fixing_duality_increase}

By definition of $\hmu(\ff^{t},\ss^{t},\vnu^{t})$ (cf. \eqref{eq:def_hmu}), we have that
\begin{equation*}
\hmu(\ff^{t},\ss^{t},\vnu^{t})=\frac{\sum_{e} \nu_e^t \hmu_e^t}{\sum_{e} \nu_e^t }= \hmu(\ff^{t},\ss^{t},\vnu^{t}) \left(1+ \frac{\sum_{e} \nu_e^t \lambda_e }{\sum_{e} \nu_e^t}\right),
\end{equation*}
where $\lambda_e:=\frac{(\hmu_e^t-\hmu(\ff^{t},\ss^{t},\vnu^{t}))}{\hmu(\ff^{t},\ss^{t},\vnu^{t})}$, for each arc $e$.

On the other hand, we have that
\begin{equation}\label{eq:fixing_lem_lem}
\hmu(\ff',\ss',\vnu')=\frac{\sum_{e\in S} \nu_e^t \hmu_e^t}{\sum_{e\in S} \nu_e^t }= \hmu(\ff^{t},\ss^{t},\vnu^{t}) \left(1+ \frac{\sum_{e\in S} \nu_e^t \lambda_e }{\sum_{e\in S} \nu_e^t}\right)\leq \hlambda\hmu(\ff^{t_0},\ss^{t_0},\vnu^{t_0})\left(1+ \frac{\sum_{e\in S} \nu_e^t |\lambda_e| }{\vnu^t(S)}\right),
\end{equation}
where $S$ is the set of non-auxiliary arcs of $\oG$. 

Now, observe that by definition of $\hgamma$-centrality (cf. \eqref{eq:def_centrality}) we have
\[
\norm{\vlambda}{\vnu^t,2}=\frac{\norm{\hvmu^t-\hmu(\ff^t,\ss^t,\vnu^t)}{\vnu^t,2}}{\hmu(\ff^t,\ss^t,\vnu^t)}\leq \hgamma.
\]
So, by applying Cauchy-Schwarz inequality we get that
\[
\frac{\sum_{e\in S} \nu_e^t |\lambda_e|}{\vnu^t(S)}\leq \sqrt{\frac{\sum_{e\in S}\nu_e^t \lambda_e^2}{\vnu^t(S)}}=\sqrt{\frac{\norm{\vlambda}{\vnu^t,2}^2}{\vnu^t(S)}}\leq \sqrt{\frac{\hgamma^2}{\vnu^t(S)}}\leq O(\hm^{-\frac{1}{2}}),
\]
where we use the fact that $\vnu^t(S)\geq \hm$. 

By putting the above inequality and \eqref{eq:fixing_lem_lem} together, the lemma follows.

%% file: files/app_improved_algorithm2.tex
\subsection{Handling Approximate Nature of Electrical Flow Computations}\label{app:inexact_elec_flow_disc}

Here, we discuss how one can adjust our algorithm developed in Sections \ref{sec:simple}--\ref{sec:proof_main_interior_point} to nearly-linear time electrical flow computations that are only approximate -- as in Theorem \ref{thm:vanilla_SDD_solver} -- instead of being exact. 

To this end, let us first recall that we are using electrical flow computations in two places of our algorithm. One is our improvement step described in Section \ref{sec:proof_main_interior_point}. There, to make the descent step, we compute the electrical flow $\hff^t$ associated with our solution and then, to make the centering step, we compute the electrical flow $\tff^t$. The other place where we use electrical flow computations is to check the $\htheta$-smoothness condition (cf. Definition \ref{def:smoothness}), that is to check which arcs are in the sets $\Cset{l}{\hff^t}$, for $l\leq \log \htheta^{-3}$. (Note that it is sufficient for us to know this classification only approximately, say up to a constant factor.)

Observe that in all these three cases, we end up computing some electrical $\vsigma$-flows that are determined by some resistances $\rr$ defined as $r_e=\frac{s_e}{f_e}$ (cf. \eqref{eq:hf_resistances}), for each arc $e$, and where $(\ff,\ss,\vnu)$ is some $\gamma$-centered and $\ovsigma'$-feasible solution  with $\gamma\leq \frac{1}{2}$ and both $\onorm{\vsigma}$ and $\onorm{\vsigma'}$ being $O(\hm)$. Furthermore, we always have that $\frac{1}{O(\hm)}\leq \hmu(\ff,\ss,\vnu)\leq 1$, all variables $s_e$ are bounded by a constant (see Lemma \ref{lem:bound_on_s}), and the duality gap is $O(\hm)$. (All the definitions that are relevant here can be found in Section \ref{sec:simple} and at the beginning of Section \ref{sec:improved}.)

This implies that, for any arc $e$, $f_e$ is always polynomially bounded in $\hm$. This is so since, given our polynomially-bounded demands, any flow of value $\omega(\hm)$ would need to consist mostly of flow-cycles, and such flow-cycles would contribute to duality gap (as they cannot exist in optimal solution), which is always $O(\hm)$. 

This, in turn, together with $\gamma$-centrality (see Fact \ref{fa:central_vs_max_min}) and Invariant \ref{inv:measure_upperbound}, allows us to conclude that all the resistances 
\[
r_e=\frac{s_e}{f_e}=(1\pm \gamma)\frac{\nu_e \hmu(\ff,\ss,\vnu)}{f_e^2} = (1\pm \gamma)\frac{s_e^2}{\nu_e \hmu(\ff,\ss,\vnu)}
\]
are within a polynomial in $\hm$ factor of each other. 

It is known (see, e.g., Theorem 2.3 in \cite{ChristianoKMST11}) that once all the resistances are within polynomial of each other, one can afford very good (and fast) approximation to all the major characteristics of the electrical flows (including good approximation to the flow on each of the edges). In particular, one is able to easily perform (approximate) classification of arcs into sets $\Cset{l}{\hff^t}$, for $l\leq \log \htheta^{-3}$. (Note that we want to classify here only arcs that contribute significant portion of the total energy anyway.)  Also, looking at our analysis of our improvement step in Section \ref{sec:proof_main_interior_point}, one can see that the most fundamental requirement there is that the flows $\hff^t$ and $\tff^t$ that we compute are indeed electrical flow, i.e., there are voltages that induce them via \eqref{eq:potential_flow_def}. After all, this is what ensures that our first-order updates to the centrality are canceling out. The fact that these flow might not have the exact demands we requested is of lesser importance. The only effect of the latter will be that our improvement steps will end up perturbing the $\hvsigma$-feasibility of our maintained solution. However, given that we have polynomially bounded resistance ratio and logarithmic dependence on error, we can always make these perturbation very small and just fix them at the end of each $\htheta$-improvement steps via the fixing procedure that we already employ to fix the effects of the preconditioning -- see Section \ref{sec:preconditioning}. 

In the light of the above, we can conclude that indeed, having approximate, instead of exact, electrical flow computations is acceptable for our algorithm, at least as long the dependence of the running time on the error is only logarithmic (which is the case here). 

%% file: files/app_interior_point.tex
\section{Appendix to Section \ref{sec:proof_main_interior_point}}\label{app:interior}

To prove the second part of the theorem, note first that indeed $\vnu^{t+1}=\ovnu^{t}=\vnu^t$. Next, we can check that the cumulative changes of the vectors $\ff^{t}$ and $\ss^t$ are equal to
\begin{eqnarray*}
f_e^{t+1} &=& (1-\delta^t)\left(1+\frac{\delta^t\hf^t_e}{(1-\delta^t)f_e^t}\right)\left(1-\frac{\ohmu_e^t-\hmu(\off^t,\oss^t,\ovnu^t)}{\ohmu_e^t}\right)\left(1+\frac{\tf^t_e}{\of_e'}\right)f_e^t\\
s_e^{t+1} &=& \left(1-\frac{\delta^t\hf^t_e}{(1-\delta^t)f_e^t}\right)\left(1-\frac{\tf^t_e}{\of_e'}\right)s_e^t,
\end{eqnarray*}
for each arc $e$ in $\hG$. 

As a result, by \eqref{eq:hf_resistances}, we have that for each arc $e$,
\[
(1-\delta^t)r_e^{t+1} = \frac{(1-\delta^t) s_e^{t+1}}{f_e^{t+1}} = \frac{\left(1-\frac{\delta^t\hf^t_e}{(1-\delta^t)f_e^t}\right)\left(1-\frac{\tf^t_e}{\of_e'}\right)s_e^t}{\left(1+\frac{\delta^t\hf^t_e}{(1-\delta^t)f_e^t}\right)\left(1-\frac{\ohmu_e^t-\hmu(\off^t,\oss^t,\ovnu^t)}{\ohmu_e^t}\right)\left(1+\frac{\tf^t_e}{\of_e'}\right)f_e^t}.
\]

Recall that from the discussion above we already know that, for each arc $e$,  $|\frac{\delta^t\hf^t_e}{(1-\delta^t)f_e^t}|=\frac{\delta^t\rho(\hff^t,\ff^t)_e}{(1-\delta^t)}\leq 2\sqrt{\hgamma}\leq \frac{1}{10}$ (cf. \eqref{eq:interior_bound_2}), $|\frac{\ohmu_e^t-\hmu(\off^t,\oss^t,\ovnu^t)}{\ohmu_e^t}|\leq \frac{\hgamma}{(1-\hgamma)} \leq \frac{1}{40}$ (cf. Fact \ref{fa:central_vs_max_min}), and $|\frac{\tf^t_e}{\of_e'}|=\rho(\tff^t,\off')_e\leq \frac{1}{40}$ (cf. \eqref{eq:interior_bound_1}). So, as $r_e^t=\frac{s_e^t}{f_e^t}$, we have that 
\[
(1+\kappa_e^t)=\frac{\left(1-\frac{\delta^t\hf^t_e}{(1-\delta^t)f_e^t}\right)\left(1-\frac{\tf^t_e}{\of_e'}\right)}{\left(1+\frac{\delta^t\hf^t_e}{(1-\delta^t)f_e^t}\right)\left(1-\frac{\ohmu_e^t-\hmu(\off^t,\oss^t,\ovnu^t)}{\ohmu_e^t}\right)\left(1+\frac{\tf^t_e}{\of_e'}\right)}
\]
and by performing a simple Taylor expansion approximation we can obtain that 
\[
|\kappa_e^t|\leq 2\left(\frac{\delta^t\rho(\hff^t,\ff^t)_e}{(1-\delta^t)}+ \frac{|\ohmu_e^t-\hmu(\off^t,\oss^t,\ovnu^t)|}{\ohmu_e^t}+\rho(\tff^t,\off')_e\right)\leq 4\left(\delta^t\rho(\hff^t,\ff^t)_e+\hkappa_e^t\right),
\]
for each arc $e$, where $\hkappa_e^t:=\frac{|\ohmu_e^t-\hmu(\off^t,\oss^t,\ovnu^t)|}{\ohmu_e^t}+\rho(\tff^t,\off')_e$. 

Clearly, this means, in particular, that $\inorm{\kappa_e^t}\leq 2(\frac{2}{10}+\frac{1}{40}+\frac{1}{40})=\frac{1}{2}$, as desired. So, we just need to show that $\norm{\hvkappa^t}{\vnu^t,2}\leq \frac{1}{16}$ too. To this end, observe that
\begin{eqnarray*}
\norm{\hvkappa^t}{\vnu^t,2}^2 &\leq& 2 \sum_{e} \nu_e^t \left(\frac{(\ohmu_e^t-\hmu(\off^t,\oss^t,\ovnu^t))^2}{(\ohmu_e^t)^2} + \rho(\tff^t,\off')_e^2\right) \\
&\leq& 2\left( \frac{\norm{\ohvmu^t-\hmu(\off^t,\oss^t,\ovnu^t)\onev}{\ovnu^t,2}^2}{(1-\hgamma)^2\hmu(\off^t,\oss^t,\ovnu^t)^2} + \frac{1}{\hmu(\off^t,\oss^t,\ovnu^t)} \sum_e \hmu(\off^t,\oss^t,\ovnu^t) \rho(\tff^t,\off')_e^2\right)\\
&\leq & 2 \left(9 \hgamma^2 + \frac{\energy{\trr^t}{\tff^t}}{\hmu(\off^t,\oss^t,\ovnu^t)} \right) \leq \frac{1}{256},
\end{eqnarray*}
as desired, where we used a combination of Fact \ref{fa:central_vs_max_min}, the fact that $(\off^t,\oss^t,\ovnu^t)$ is $3\hgamma$-centered, as well as, equations \eqref{eq:perfect_centering_of_intermediate} and  \eqref{eq:energy_estimate_final}.
This concludes the proof of the theorem.